\documentclass[11pt]{article}
\newcommand{\ifabs}[2]{#2}
\usepackage{tikz}
\usepackage{epsf}
\usepackage{epsfig}
\usepackage{graphicx}
\usepackage{color}
\usepackage{mathrsfs}
\usepackage{fullpage}
\usepackage{amsthm}
\usepackage{latexsym}
\usepackage{amssymb}
\usepackage{amsmath}
\usepackage{enumerate}
\usepackage{geometry}
\usepackage{stmaryrd}

\usepackage{epsf}
\usepackage{verbatim}

\usepackage[all,cmtip]{xy}
\usepackage{breqn}

\usepackage{tikz}
\usetikzlibrary{positioning}
\usetikzlibrary{decorations.pathreplacing,calc}

\newtheorem*{fact}{Fact}
\newtheorem{theorem}{Theorem}[section]
\newtheorem{lemma}{Lemma}[section]
\newtheorem{corollary}{Corollary}[section]

\newtheorem{definition}{Definition}[section]
\newtheorem*{proposition*}{Proposition}

\newcommand{\Pin}[3]{\textsc{Pin}_{{#2},{#3}}({#1})}

\newcommand{\mgn}{\mathbb{P}}
\newcommand{\Mgnt}{\widehat{\mathbb{P}}}
\newcommand{\mgnt}[1]{\Mgnt^{({#1})}}

\newcommand{\dli}{\Omega^{e}}

\newcommand{\deriv}[2]{\frac{\operatorname{d}{#1}}{\operatorname{d}{#2}}}
\newcommand{\pderiv}[2]{\frac{\partial{#1}}{\partial{#2}}}

\newcommand{\Hf}[5]{\mathcal{H}({\mathcal{F}}_{#1, #2}^{#3}, {\Lambda}_{#4, #5})}

\begin{document}
\title
{\bf FPTAS for Weighted Fibonacci Gates and Its Applications }

\author{Pinyan Lu\thanks{Microsoft Research. {\tt
      pinyanl@microsoft.com}}
  \and Menghui Wang\thanks{University of Wisconsin-Madison. This work was partially performed when the author was an undergraduate student at Shanghai Jiao Tong Univerisity. {\tt menghui@cs.wisc.edu}}
  \and Chihao Zhang\thanks{Shanghai Jiao Tong Univerisity. {\tt chihao.zhang@gmail.com}}
}

\date{}
\maketitle

\bibliographystyle{plain}

\begin{abstract}
Fibonacci gate problems have severed as computation primitives to solve other problems by holographic algorithm~\cite{FOCS08}
and play an important role in the dichotomy of exact counting for Holant and CSP frameworks~\cite{STOC09}. We  generalize them to weighted cases and allow each vertex function to have different parameters, which is a much boarder family and \#P-hard for exactly counting.
We design a fully polynomial-time approximation scheme (FPTAS) for this generalization by correlation decay technique.
  This is the first deterministic FPTAS for approximate counting in the general Holant framework without a degree bound.
We also formally introduce holographic reduction in the study of approximate counting
  and these weighted Fibonacci gate problems  serve as computation primitives for approximate counting.
  Under holographic reduction, we obtain FPTAS for other Holant problems and spin problems.
  One important application is developing an FPTAS for a large range of ferromagnetic two-state spin systems.
  This is the first deterministic FPTAS in the ferromagnetic range for two-state spin systems without a degree bound.
  Besides these algorithms, we also develop several new tools and techniques to establish the correlation decay property, which are applicable in other problems.
\end{abstract}

  \thispagestyle{empty}
  \newpage
  \setcounter{page}{1}

\section{Introduction}
Holant is a refined framework for counting problems~\cite{FOCS08,STOC09,holant}, which is more expressive than previous frameworks such as counting constraint satisfaction problems (CSP)  in the sense that they can be simulated using Holant instances.
In this paper, we consider a generalization called weighted Holant problems.
A weighted Holant is an extension of a Holant problem where each edge $e$ is assigned an activity $\lambda_e$, and if it is chosen it contributes to the partition function a factor of $\lambda_e$.
Given a graph $G(V,E)$, a family of node functions ${\cal F}=\left\{F_v | v \in V\right\}$, and edge weights $\Lambda=\left\{\lambda_e | e\in E\right\}$, the partition function for a weighted Holant instance $\Omega\left(G,{\cal F},\Lambda\right)$ is the summation of the weights over all configurations $\sigma:E\rightarrow \left\{0,1\right\}$, specifically the value of $$\sum_{\sigma}\left(\prod_{e \in E}\lambda_e\left(\sigma\left(e\right)\right)\prod_{v\in V}F_v\left(\sigma |_{E\left(v\right)}\right)\right).$$
We use Holant(${\mathcal F}, \Lambda$) to denote the class of Holant problems where all functions are taken from ${\mathcal F}$ and all edge weights are taken from $\Lambda$.
For example, consider the {\sc Perfect Matching} problem on $G$.
This problem corresponds to attaching the {\sc Exact-One} function on every vertex of $G$ --- for each 0-1 edge assignment, the product $\prod_{v\in V} F_v(\sigma\mid_{E(v)})$ evaluates to 1 when the assignment is a perfect matching, and 0 otherwise, thereby summing over all 0-1 edge assignments gives us the number of perfect matchings in $G$.
If we use the {\sc At-Most-One} function at each vertex, then we can count all matchings, including those that are not perfect.

A symmetric function $F$ can be expressed by $[f_0,f_1,\ldots,f_k]$, where $f_i$ is the value of $F$ on inputs of hamming weight $i$. The above mentioned {\sc Exact-One} and {\sc At-Most-One} functions are both symmetric and can be expressed as $[0,1,0,0, \ldots]$ and $[1,1,0,0, \ldots]$ respectively. 
A Fibonacci function $F$ is a symmetric function $[f_0,f_1,\ldots,f_k]$, satisfying that $f_i=c f_{i-1}+f_{i-2}$ for some constant $c$. For example, the parity function $[a,b, a,b, \ldots]$ is a special Fibonacci function with $c=0$.
If there are no edge weights (or equivalently all the weights are equal to 1) and all the node functions are Fibonacci functions with a same parameter $c$, we have a polynomial time algorithm to compute
the partition function exactly~\cite{FOCS08}. These problems also form the base for a family of holographic algorithms, where other interesting problems can be reduced to the Fibonacci gate problems~\cite{FOCS08}. Furthermore,
this family of functions is interesting not only because of its tractability, but also because it essentially captures almost all tractable Holant problems with all unary functions available~\cite{STOC09,holant}.
%It also serves as a computational primitive, based on which we design polynomial time algorithms for other problems by holographic reduction~\cite{FOCS08}.

If we allow edges to have non-trivial weights or each functions to have different parameters in Fibonacci gates, then the exact counting problem becomes \#P-hard~\cite{STOC09,holant}. Nevertheless,	it is interesting to study the problem in the approximation setting. We first introduce the solution concepts for approximate counting. A fully polynomial-time approximation scheme (FPTAS) is an algorithm scheme that approximates the answer to a problem within an arbitrarily small relative error in polynomial time. More precisely, an FPTAS is an algorithm scheme such that for any given parameter $\varepsilon>0$, the algorithm produces an output $Z'$ satisfying $(1-\varepsilon)Z<Z'<(1+\varepsilon)Z$,
where $Z$ is the correct answer, and runs in time $poly(n,1/\varepsilon)$.
Its randomized relaxation is called a fully polynomial-time randomized approximation scheme (FPRAS),
 which uses random bits in the algorithm and requires that the final output be within the range $[(1-\epsilon) Z, (1+\epsilon) Z]$ with high probability.
 In contrast to the exact counting setting, the approximability of Holant problem is much less well-understood.  In this paper, we study approximate counting for weighted Fibonacci gate problems.

Another closely related and well-studied model is spin systems.
In this paper, we focus on two-state spin systems.
An instance of a spin system is a graph $G(V,E)$.
A configuration $\sigma: V\rightarrow \{0, 1\}$ assigns every vertex one of the two states.
The contributions of local interactions between adjacent vertices are quantified by a matrix
$A=\begin{bmatrix} A_{0,0} & A_{0,1} \\ A_{1,0} & A_{1,1} \end{bmatrix}=\begin{bmatrix} \beta & 1 \\ 1 & \gamma \end{bmatrix}$, where $\beta, \gamma \geq 0$.
The partition function is defined by
\ifabs{
  $Z_A(G)=\sum_{\sigma \in \{0,1\}^{V}} \prod_{(u,v) \in E} A_{\sigma(u), \sigma(v)}$.
}
{
  \[ Z_A(G)=\sum_{\sigma \in 2^{V}} \prod_{(u,v) \in E} A_{\sigma(u), \sigma(v)}.\]}
There has been a lot of studies on the approximability of the partition function in terms of parameters $\beta$ and $\gamma$.
The problem is exactly solvable in polynomial time if $\beta \gamma=1$.
When $\beta \gamma<1$, the system is called anti-ferromagnetic and we have a complete understanding of its approximability: there is a uniqueness boundary, above which there is an FPTAS~\cite{Weitz06,LLY12,SST,LLY13} and below which it is NP-hard~\cite{Sly10,SS12,galanis2012inapproximability}.

The story is different in ferromagnetic range $\beta\gamma>1$. Jerrum and Sinclair \cite{JS93} gave an FPRAS for Ising model ($\beta=\gamma>1$) based on Markov Chain Monte Carlo (MCMC) method and lately Goldberg et al. extended that to all $\beta\gamma>1$ plane. However, these algorithms are all randomized. Can we design a deterministic FPTAS for it as that for anti-ferromagnetic range?
Indeed,  this is an interesting and important question in general and many effort has been made for derandomizing MCMC based algorithms. For instance, there is an FPRAS for counting matchings~\cite{jerrum1989approximating} but FPTAS is only known for graphs of bounded degree~\cite{BGKNT07}. The situation is similar in computing permanent of nonnegative matrix, although an FPRAS is known \cite{app_JSV04}, the current best deterministic algorithm can only approximate the permanent with an exponential large factor \cite{linial1998deterministic}.
To the best of our knowledge, no deterministic FPTAS was previously known for two-state spin systems in ferromagnetic range. In particular, the correlation decay technique, the main tool to design FPTAS in anti-ferromagnetic range, cannot directly apply.

%To the best of our knowledge, no FPTAS is known for the ferromagnetic two-state spin or Ising model and thus it is interesting to ask whether one can do so.

%For the ferromagnetic range $\beta \gamma>1$, only an FPRAS is known~\cite{JS93,GJP03}.
%It has yet to be determined whether one can design a deterministic FPTAS in the ferromagnetic range.

\subsection{Our Results}
The main results of this paper are a number of FPTAS's for computing the partition function of different Holant problems and spin systems.

\begin{description}
\item[Weighted Fibonacci gates.]  We design an FPTAS for weighted Fibonacci gates when the parameters satisfy certain conditions.
  We have several theorems to cover different ranges.
  In Theorem~\ref{thm:1}, we prove that for any fixed choice of other parameters, we can design an FPTAS as long as the edge weights are close enough to $1$, which corresponds to the unweighted case.
  This result demonstrates a smooth transition from the unweighted case to weighted ones in terms of approximation.

  Another interesting range is that we have an FPTAS for the whole range as long as the Fibonacci parameter $c$ is reasonably large (no less than a constant 1.17) and edge weights are no less than $1$ (which means all the edges prefer to be chosen) (Theorem~\ref{thm:2}).
  It is worth noting that we allow Fibonacci functions on different nodes to have different parameters $c$, which contrasts the exact counting setting where it is crucial for different functions to have the same parameter in order to have a polynomial time algorithm.

\item[Ferromagnetic two-state spin systems.] We design an FPTAS for a large range of ferromagnetic two state spin systems.
  This is the first deterministic FPTAS in the ferromagnetic range for two-state spin systems without a degree bound.
  To describe the tractable range, we present a monotonically increasing function $\Gamma:[1, \infty] \to \mathbb{R}$ with $\Gamma(1)=1$ and $\Gamma(x)\leq x$.
  We have an FPTAS for a ferromagnetic spin system $\begin{bmatrix} \beta & 1 \\ 1 & \gamma \end{bmatrix}$ as long as $\gamma \leq \Gamma(\beta)$ or  $\beta \leq \Gamma(\gamma)$ (Theorem~\ref{thm:spin-holant}).
  The exact formula of $\Gamma$ is complicated and we do not spend much effort to optimize it.
  However, it already enjoys a nice property in that $\lim_{x \to +\infty} \frac{\Gamma(x)}{x}=1$.
  This means that although the range does not cover the Ising model ($\beta=\gamma$), it gets relatively close to that in infinity.
  We also have similar results for two-spin system with external fields.

\item[Other Holant Problems.] We can extend our FPTAS to functions $[f_0,f_1,\ldots,f_d]$ with form
  $f_{i+2}=a f_{i+1} + b f_i$ if the parameters satisfy certain conditions.
  This is a much broader family than Fibonacci gates, since Fibonacci gates corresponds to $b=1$.
\end{description}

\subsection{Our Techniques}
Our main approach for designing FPTAS's is the correlation decay technique as introduced in \cite{BG08} and \cite{Weitz06}.
While the general framework is standard, it is highly non-trivial to design a recursive computational structure and especially to prove the property of exponential correlation decay for a given set of problems.
This is in analog to designing an FPRAS with the Markov Chain Monte Carlo (MCMC) method: though general framework for these algorithms is the same, it is still difficult to design Markov chains for different problems and especially to prove the rapid mixing property~\cite{MC_JA96}.
One powerful technique here is to use a potential function to amortize the decay rate, which has been introduced and used in many problems~\cite{RSTVY11,LLY12,SST,LLY13,LY13} and which we utilize here.
Besides this, to enrich the tool set, we introduce several new techniques to design the recursive computational structure and to prove the correlation decay property.
We believe that these techniques can find applications in other problems.

\begin{description}
\item[Working with dangling edges.] The recursive computational structure for spin problems usually relates a marginal probability of a vertex to that of its neighbors.
  In Holant problems, we are talking about the assignments and marginal probabilities of edges.
  Since an edge has two ends, it has two set of neighbors, which complicates things a lot.
  In this paper, we choose to work on instances with dangling edges, which is a half edge that have neighbors only on one end.
  It is much easier to write recursions on dangling edges.
  This technique works for any Holant problems and we believe it is the right thing to work with in the Holant framework. Indeed, the idea has later been successfully used in~\cite{counting-edge-cover}.

\item[Computation tree with bounded degrees.] Usually, the correlation decay property only implies an FPTAS for systems with bounded degrees.
  One exception is the anti-ferromagnetic two-state spin systems, where a stronger notion of computationally efficient correlation decay is introduced\cite{LLY12}.
  In this paper, we also establish the computationally efficient correlation decay for systems without a degree bound, but via a different approach.
  By making use of the unique property of Fibonacci functions, we can decompose a node into several nodes with constant degrees.
  Thus, at each step of our computation tree, we only involve constant many sub-instances even if the degree of the original system is not bounded.

\item[Bounding range of variables.] After we get a recursion system, the main task is to prove the correlation decay property.
  This is usually achieved by proving that a certain amortized decay rate, which is a function of several variables, is less than one for any choice of these variables in their domain.
  If we can prove that these variables are always within smaller domains, then we only need to prove that the rate is less than one under these smaller domains, which becomes weaker and easier to prove.
  Some naive implementation of this idea already appeared in approximate counting of coloring problems~\cite{GK07,LY13}.
  In this paper, we develop this idea much further.
  We divide sub-instances involved in the computation tree into two classes: deep ones for which we can get a much better estimation of their range and shallow ones for which we can compute their value without error.
  Then we can either compute the exact value or we can safely assume that it is within a smaller domain, which enables us to prove the correlation decay property easier.

\item[Holographic reduction.] We formally introduce holographic reduction in the study of approximate counting.
We use weighted Fibonacci gate problems  as computation primitives for approximate counting and design holographic algorithms for other problems based on them.
 In particular,
%  Under this idea, we establish equivalent relations between Fibonacci gate problems and two-state spin problems.
  we use the FPTAS for Fibonacci gates to obtain an  for ferromagnetic two-state spin systems.
  It is noteworthy that the correlation decay property does not generally hold for ferromagnetic two-state spin systems.
  So we cannot do a similar argument to get the FPTAS in the spin world directly. 
  Moreover, the idea of holographic reduction can apply to any Holant problems, which extends known counting algorithms (both exact and approximate, both deterministic and randomized) to a broader family of problems.
  Indeed, the other direction of holographic reduction is also used in our algorithm.  We design an exact algorithm for shallow sub-instances of Fibonacci instance by a holographic reduction to the spin world.
\end{description}

\subsection{Related Works}
Most previous studies of the Holant framework are for exact counting, and a number of dichotomy theorems were proved~\cite{holant,HuangL12,CaiGW13}.
Holographic reduction was introduced by Valiant in holographic algorithms~\cite{HA_FOCS,STOC07}, which is later also used to prove hardness result of counting problems~\cite{FOCS08,holant,planar}.

For some special Holant problems such as counting (perfect) matchings, their approximate versions are well studied~\cite{BGKNT07,jerrum1989approximating,app_JSV04}.
In particular, \cite{BGKNT07} gave an FPTAS to count matchings but only for graphs with bounded degrees.
It is relatively less studied in the general Holant framework in terms of approximate counting except for two recent work: \cite{YZ13} studied general Holant problems but only for planar graph instances with a bounded degree; \cite{McQuillan2013} gives an FPRAS for several Holant problems.
Another well-known example is the ``sub-graph world" in \cite{JS93}.
It is indeed a weighted Holant problem with Fibonacci functions of $c=0$, for which an FPRAS was given.
In that paper, holographic reduction was also implicitly used, which extends the FPRAS to the Ising model.

Most previous study for FPTAS via correlation decay is on the spin systems.
It was extremely successful in the anti-ferromagnetic two-spin system~\cite{Weitz06,LLY12,SST,LLY13}.
It is also used in multi-spin systems~\cite{GK07,LY13}.
Many more works focused on randomized approximate counting, such as examples~\cite{JS93,app_JSV04,app_GJ11,app_DJV01,IS_DG00,col_Jerrum95,col_Vigoda99}.

\section{Preliminaries}\label{sec:background}
A weighted Holant instance $\Omega = (G(V,E), \{ F_v | v \in V \}, \{\lambda_e | e \in E\} )$ is a tuple.
$G(V,E)$ is a graph. $F_v$ is a function with arity $d_v$: $\{0,1\}^{d_v}\rightarrow \mathbb{R}^+$, where
$d_v$ is the degree of $v$ and $\mathbb{R}^+$ denotes non-negative real numbers. Edge weight $\lambda_e$ is a mapping $\{0,1\} \rightarrow \mathbb{R}^+ $.
A configuration $\sigma$ is a mapping
$E\rightarrow \{0,1\}$ and gives a weight
 \[w_\Omega(\sigma)=\prod_{e \in E} \lambda_e (\sigma(e)) \prod_{v \in  V}
F_v(\sigma\mid_{E(v)}),\]
 where $E(v)$ denotes the incident edges of
$v$.
 The  counting problem on the instance  $\Omega$
is to compute the partition function:
\[Z(\Omega)=\sum_{\sigma}\left(
\prod_{e \in E} \lambda_e (\sigma(e)) \prod_{v\in V} F_v(\sigma\mid_{E(v)})\right).\]

We can represent each function $F_v$  by a vector
in $(\mathbb{R}^+)^{2^{d_v}}$, or a tensor in $((\mathbb{R}^+)^{2})^{\otimes d_v}$. This
is also called a {\it signature}.
% We use following notations to denote
%some special signatures. $=_k$ denotes the equality of
%arity $k$. $\Delta_s$ denotes the unary function which gives value
%1 on inputs $s$, and 0 on all other inputs.
A symmetric function $F$ can be expressed by $[f_0,f_1,\ldots,f_k]$, where $f_j$ is the value
of $F$ on inputs of hamming weight $j$. For example,
the equality function is $[1,0,\ldots,0,1]$.
Edge weight is a unary function, which can be written as $[\lambda_e(0),\lambda_e(1)]$.
Since we do not care about global scale factor, we always normalize that $\lambda_e(0)=1$ and use
the notation  $\lambda_e=\lambda_e(1)$ as a real number.

A Holant problem is parameterized by a set of functions ${\cal F}$ and edge weights $\Lambda$. We denote by  ${\rm Holant}({\cal F}, \Lambda )$ the following computation problem .
\begin{definition}
Given a set of functions ${\cal F}$ and edge weights $\Lambda$, we denote by  ${\rm Holant}({\cal F}, \Lambda )$ the following computation problem.\\
\noindent {\bf Input:} A Holant instance $\Omega = (G(V,E), \{ F_v | v \in V \}, \{\lambda_e | e \in E\} )$, where $F_v \in {\cal F}$ and $\lambda_e \in \Lambda $ ;\\
\noindent {\bf Output:} The partition function $Z(\Omega)$.
\end{definition}

The weights of configurations also give a distribution over all possible configurations:
\[\mgn_\Omega(\sigma) =\frac{w_\Omega (\sigma)}{Z(\Omega)}=\frac{1}{ Z(\Omega)}
\prod_{e \in E} \lambda_e (\sigma(e)) \prod_{v\in V} F_v(\sigma\mid_{E(v)}).\]
This defines the marginal probability of each edge $e_0 \in E$.
 \[\mgn_\Omega(\sigma(e_0)=0) =\frac{\sum_{\sigma:\sigma(e_0)=0}\left(
\prod_{e \in E} \lambda_e (\sigma(e)) \prod_{v\in V} F_v(\sigma\mid_{E(v)})\right)}{ Z(\Omega)}
.\]
Similarly, we can define the marginal probability of a subset of edges. Let $E_0\subset E$ and $e_1, e_2, \ldots, e_{|E_0|}$ be an enumeration of the edges in $E_0$. Then we
can define $\sigma(E_0)=\sigma(e_1)\sigma(e_2)\cdots\sigma(e_{|E_0|})$ as a Boolean string of length $|E_0|$.
Let $\alpha\in \{0,1\}^{|E_0|}$, we define
 \[\mgn_\Omega(\sigma(E_0)=\alpha) =\frac{\sum_{\sigma:\sigma(e_i)=\alpha_i, i=1,2,\ldots, |E_0|}\left(
\prod_{e \in E} \lambda_e (\sigma(e)) \prod_{v\in V} F_v(\sigma\mid_{E(v)})\right)}{ Z(\Omega)}
.\]
We denote the partial summation as
\[Z(\Omega, \sigma(E_0)=\alpha)=\sum_{\sigma:\sigma(e_i)=\alpha_i}\left(
\prod_{e \in E} \lambda_e (\sigma(e)) \prod_{v\in V} F_v(\sigma\mid_{E(v)})\right).\]

We define a dangling instance $\Omega^{D}$ of ${\rm Holant}({\cal F}, \Lambda )$ also as a tuple $(G(V,E\cup D), \{ F_v | v \in V \}, \{\lambda_e | e \in E\} )$, where $G(V,E\cup D)$ is a graph with dangling edges $D$.
A dangling edge can be viewed as a half edge, with one end attached to a regular vertex in $V$ and the
other end dangling  (not considered as a vertex).
A dangling instance $\Omega^{D}$ is the same as a Holant instance except for these dangling edges.
In $G(V,E\cup D)$ each node is assigned a function in ${\cal F}$  (we do not consider ``dangling'' leaf nodes at the end of a
dangling edge among these), each regular edge in $E$ is assigned a weight from $\Lambda$ and we always assume that there is no weight on a dangling edge in this paper.
A dangling instance can be also viewed as a regular instance by attaching a vertex with function $[1,1]$ at the dangling end of each dangling edge. We can define the  probability distribution and   marginal probabilities just as for regular instance.
In particular, we shall use dangling instance $\Omega^e$ with single dangling edge $e$ extensively in this paper. For that, we define
\[R(\Omega^{e})= \frac{\mgn_{\Omega^{e}}(\sigma(e)=1)}{\mgn_{\Omega^{e}}(\sigma(e)=0)}.\]

\ifabs{}
{
\begin{definition}\label{definition-partial-pinning}
Given a Holant instance $\Omega=(G(V,E), \{ F_v | v \in V \}, \{\lambda_e | e \in E\} )$, a vertex $e_0=(u_1,u_2) \in E$ and $\tau \in [0,1]$. We can define a weighted pinning operation $\Pin{\Omega}{e}{\tau}=(G'(V,E-\{e_0\}), \{ F'_v | v \in V \}, \{\lambda_e | e \in E-\{e_0\}\} )$. The graph of $\Pin{\Omega}{e}{\tau}$ is the same as that of $\Omega$ except that $e$ is removed; All the edge weights in the remaining edges are the same in both instances; all the vertex functions are the same except $u_1$ and $u_2$. For $v\in \{u_1, u_2\}$ and $\alpha\in \{0,1\}^{d_v-1}$, $F'_v(\alpha)=(1-\tau) F_v(\alpha 0) +\tau F_v(\alpha 1)$.
\end{definition}

In this definition, we have the coincidence that for $x\in \{0,1\}$,  $\Pin{\Omega}{e}{x}$ is exactly the Holant instance by fixing the edge $e$ to $x$.
%$\Pin{\Omega}{e}{\lambda_e}$ is effectively to set the edge free.

\subsection{Holographic Reduction}
Holographic reduction is powerful reduction among counting problems expressible in Holant framework.
We use $Holant({\cal G}|{\cal R})$ to denote all the counting
problems, expressed as unweighted Holant problems on bipartite graphs
$H=(U,V,E)$, where each signature for a vertex in $U$ or $V$ is from
${\cal G}$ or ${\cal R}$, respectively.   Signatures in ${\cal G}$ are
are denoted by column vectors (or
contravariant tensors); signatures in ${\cal R}_q$ are  denoted by row vectors (or covariant
tensors)~\cite{dodson}. One can perform (contravariant and covariant) tensor
transformations on the signatures, which may
produce exponential cancelations in tensor spaces.
We shall define a simple version of  holographic reductions,
which are invertible.
Suppose $Holant({\cal G}|{\cal R})$ and
$Holant({\cal G'}|{\cal R'})$ are two holant problems defined for the same
family of graphs, and
 $T \in {\bf GL}({\mathbb C})$ is a basis.
%%% i will take C here for the field.
We say that there is a holographic
%reduction from $\#{\cal H}:{\cal G}_q|{\cal R}_q$ to $\#{\cal H}:{\cal
reduction from $Holant({\cal G}|{\cal R})$ and
$Holant({\cal G'}|{\cal R'})$ , if
the {\it contravariant} transformation
$G' = T^{\otimes g} G$ and the {\it covariant} transformation
$R=R' T^{\otimes r}$ map $G\in {\cal G}$ to $G'\in {\cal G'}$
and $R\in {\cal R}$  to $R' \in {\cal R'}$,
where $G$ and $R$ have arity $g$ and $r$ respectively.
(Notice the reversal of directions when the
transformation $T^{\otimes n}$ is applied. This is the meaning of
{\it contravariance} and {\it covariance}.)

\begin{theorem}[Holant Theorem~\cite{HA_J}]\label{thm:holant}
Suppose there is  a holographic reduction from $Holant({\cal G}|{\cal R})$ to
$Holant({\cal G'}|{\cal R'})$  mapping instance $\Omega$
to $\Omega'$, then  $Z(\Omega) = Z(\Omega')$.
\end{theorem}

The proof of this theorem  follows from general
principles of contravariant and covariant tensors~\cite{dodson}.}
 %and their contractions~\cite{dodson}.

\ifabs{
\section{Statement of Main Results}
}
{
\section{Results and Applications}

We first list our FPTAS for various ranges of Fibonacci gates and show their applications in other Holant problems and spin systems. The proof of these theorems shall be given in later sections.

\subsection{Fibonacci Signature}}
A symmetrical function $[f_0,f_1,\ldots,f_d]$ is called a (generalized) Fibonacci function if there exists a constant $c$ such that
\[f_{i+2}=c f_{i+1} + f_i, \ \  \mbox{where }\ \  i=0, 1, \cdots, d-2.\]
We denote this family of function as ${\cal F}_c$, the Fibonacci functions with parameter $c$.
\ifabs{}
{Another useful way to parameterize Fibonacci functions is
\[f_i = A \rho^i + B (-\rho)^{-i}, \]
where $A,B$ are two constants and $\rho$ is the positive root of $t^2= c t +1$.
Thus, there is a one to one correspondence between parameter $c$ and $\rho$. In this paper, when one of them is defined in
a context, we assume that the other one is also defined automatically and accordingly. We shall use the other one directly and freely.

}
We use ${\cal F}_c^{p,q}$ to denote a subfamily of ${\cal F}_c$
 such that $f_{i+1} \geq p f_{i}$ and $f_{i+1} \leq q f_{i}$  for all $i=0, 1, \cdots, d-1 $.
When the upper bound $q$ is not given, we simply write ${\cal F}_c^{p}$.
We use ${\cal F}_{c_1, c_2}^{p,q}$ to denote $\bigcup_{c_1\leq c \leq c_2} {\cal F}_c^{p,q}$.
We use $\Lambda_{\lambda_1,\lambda_2}$ to denote the set of edge weights $\lambda_e$ such that
$\lambda_1\leq \lambda_e \leq \lambda_2$.

Here is a list of FPTAS's we get:

\begin{theorem}
\label{thm:1}
For any $c>0$ and $p>0$, there exists $\lambda_1(p,c)<1$ and $\lambda_2(p,c)>1$ such that
there is an FPTAS for  ${\rm Holant}({\cal F}_c^p, \Lambda_{\lambda_1(p,c),\lambda_2(p,c)} )$.
\end{theorem}

\begin{theorem}
\label{thm:2}
Let $p>0$.  Then there is an FPTAS for  ${\rm Holant}({\cal F}_{1.17, +\infty}^p, \Lambda_{1,+\infty} )$.
\end{theorem}

\begin{theorem}
\label{thm:spin}
Let $\lambda>0$ and $c\geq 2.57$.
There is an FPTAS for ${\rm Holant}({\cal F}_{c}^{c/2,c+2/c}, \Lambda_{\lambda,+\infty} )$.
\end{theorem}
%
%\begin{theorem}
%\label{thm:3}
%Let $\lambda>0,p>0$.  Then there is an FPTAS for  ${\rm Holant}( {\cal F}_{?, +\infty}^p, \Lambda_{\lambda,+\infty} )$.
%\end{theorem}

\ifabs{
As an important application of the above theorems, we get the following FPTAS for ferromagnetic two-state spin system .
\begin{theorem}
\label{thm:spin-holant}
There is a continuous curve $\Gamma(\beta)$ defined on $[1,+\infty)$ such that (1) $\Gamma(1)=1$;
(2) $1<\Gamma(\beta)<\beta$ for all $\beta>1$; and (3) $\lim_{\beta \to +\infty} \frac{\Gamma(\beta)}{\beta}=1$.
There is an FPTAS for the  two-state spin system with local interaction matrix
    $\begin{bmatrix}
        \beta & 1 \\
        1 & \gamma
    \end{bmatrix}$
and external field $\mu \leq 1$ if $\beta\gamma>1$ and $\gamma \leq \Gamma(\beta)$.
\end{theorem}
}
{}

\ifabs{}
{
  \subsection{Beyond Fibonacci}
We use ${\cal L}_{a,b}$ to denote the set of all symmetric functions  $[f_0,f_1,\ldots,f_d]$ which satisfies that
\[f_{i+2}=a f_{i+1} + b f_i, \ \  \mbox{where }\ \  i=0, 1, \cdots, d-2.\]
And we use ${\cal L}$ to denote all these functions for different $a$ and $b$. We shall show that an instance of ${\rm Holant}({\cal L}, \Lambda)$
can be transformed to an instance of Fibonacci gates. Given an instance $\Omega = (G(V,E), \{ F_v | v \in V \}, \{\lambda_e | e \in E\} )$ of ${\rm Holant}({\cal L}, \Lambda)$, we can modify a function  $F_v=[f_0,f_1,\ldots,f_d]\in {\cal L}_{a,b}$ to
\[ [g_0, g_1, \ldots, g_d]=[f_0,\frac{f_1}{\sqrt{b}},\ldots,\frac{f_d}{b^{d/2}}]. \]
Then these $[g_0, g_1, \ldots, g_d]$ satisfies that
\[g_{i+2}=\frac{a}{\sqrt{b}} g_{i+1} +  g_i, \ \  \mbox{where }\ \  i=0, 1, \cdots, d-2,\]
which is a Fibonacci function.
At the same time, we modify the edge weight of each neighbor of $v$ from $\lambda$ to $\lambda \sqrt{b}$. By the definition of partition function, it is easy to verify that the partition function remains the same after these simultaneous modification of vertex function and edge weighs. We can do this for all the vertex functions and edge weights.
This is indeed a holographic reduction under the basis $\begin{bmatrix} 1 & 0 \\ 0 & \sqrt{b} \end{bmatrix}$.
Finally we can get an instance of Fibonacci gate. So all our FPTAS results for Fibonacci gates can be translated to an FPTAS results of a subfamily of  ${\rm Holant}({\cal L}, \Lambda)$.

  \subsection{Holographic reduction and spin world}\label{sec:spin}
Weighted Holant problem can also be interpreted as an (unweighed) Holant problem defined on bipartite graphs. For any Holant instance on a general graph, we can make it
bipartite by adding an additional vertex on each edge, and for the
new vertex on a edge with weight $\lambda$, the function on it is
$[1, 0, \lambda]$. The new bipartite graph is unweighed (no edge weights).
It is clear that this modification does not change the partition function of the instance.
For this bipartite Holant, we can apply a holographic reduction under base $\begin{bmatrix} 1 & t \\ \rho & -\frac{t}{\rho}\end{bmatrix}$ to get the following lemma.

\begin{lemma}
\label{lem:holographic}
Let $\lambda>0$, $\rho\geq 1$, $t(1-\lambda)>0$, and $|t|\leq 1$.
Let $\beta=\frac{1+\lambda \rho^2}{t(1-\lambda)}$ and $\gamma=\frac{t(1+\lambda \rho^{-2})}{1-\lambda}$.
The two spin problem with edge function $\begin{bmatrix} \beta & 1 \\ 1 & \gamma \end{bmatrix}$ and external field $\mu$
is equivalent to ${\rm Holant}( {\cal F}_{\rho-\frac{1}{\rho}}, \Lambda_{\lambda,\lambda} )$, where ${\cal F}_{\rho-\frac{1}{\rho}}$ is a set of Fibonacci functions
with parameter $c=\rho-\frac{1}{\rho}$ and the one of arity $n$ has form
\begin{equation}
    \label{eq:spin-fibo}
    f_k=\rho^k+\mu t^n (-\rho)^{-k}.
\end{equation}
\end{lemma}

Through this reduction, Theorem~\ref{thm:1}-\ref{thm:spin} give a region on the $\beta\mbox{-}\gamma$ plane in which the ferromagnetic two-state spin system problem admits an FPTAS. The explicit range is complicated and not very informative. We use a function
$\Gamma(\beta)$ to denote the combined range of the above three theorems and have the following FPTAS for ferromagnetic two spin system.
%
%Let ${\cal{F}}_c^{p,q}$ denote a family of Fibonacci signatures such that
%\begin{align*}
%    & f_{i+2}=cf_{i+1}+f_i, \\
%    & f_{i+1}\geq pf_i, \\
%    & f_{i+1}\leq qf_i, \\
%\end{align*}
%for all $i=0,1,\ldots,d$.
%
%\begin{lemma}
%Fibonacci signatures reduced from two spin problems \eqref{eq:spin-fibo} satisfies that $f_{i+1} \geq \frac{c}{2}f_i$ and $f_{i+1} \leq (c+\frac{2}{c})f_i$.
%\end{lemma}
%\begin{proof}
%Set $t=1$ when $\lambda \leq 1$ and $t=-1$ otherwise.
%Then use Lemma~\ref{lem:bound}.
%\end{proof}
%
%\begin{theorem}
%\label{thm:spin}
%Let $\lambda>0$ and $c\geq 4\sqrt{\sqrt{2}-1}=2.57$.
%There is an FPTAS for ${\rm Holant}({\cal F}_{c}^{c/2,c+2/c}, \Lambda_{\lambda,\lambda} )$.
%\end{theorem}

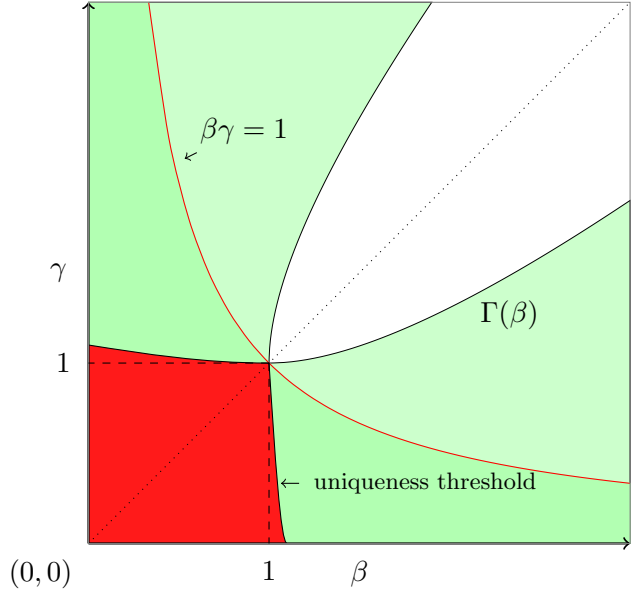
\begin{figure}[ht]
    \centering
    \begin{tikzpicture}[scale=0.8]
%  \draw[help lines] (0,0) grid (9,9);

  \filldraw [domain=3:9,smooth,fill=green!20, very thin] plot(\x,{3*(\x/3-ln(\x/3))})--(9,0) -- (3,0) -- cycle;
  \filldraw [domain=3:9,smooth, variable=\t,fill=green!20, very thin] plot({3*(\t/3-ln(\t/3))},{\t}) -- (0,9) -- (0,3) -- cycle;

  \filldraw [domain=1:9,smooth,color=green!30] plot(\x,9/\x) -- (9,0) -- (0,0) -- (0,9) -- cycle;
%  \filldraw [fill = red!30] (0,0) rectangle (1,1);
  \draw [domain=1:9,smooth,color=red] plot(\x,9/\x);
  \filldraw [fill=red!90] (0,3.3) sin (3,3) sin (3.3,0) -- (0,0) -- cycle; % I use sin curve only for approximation here!
%  \draw (1,1) sin (1.1,0);% node[anchor=west]{\tiny uniqueness threshold}; %
  \draw [dotted] (0,0) -- (9,9);

%  \fill [green!50] (1,1) sin (1.1,0) -- (3,0) -- (3,1/3) -- (hyper);
  \draw [thick,->] (0,0) -- (9,0);
  \draw [thick,->] (0,0) -- (0,9);

  \draw [very thin, color=black!50] (0,0) rectangle (9,9);
  \draw [->] (1.8,6.5)-- (1.6,6.4) node[above right=0.1cm]{$\beta\gamma=1$};
  \draw [->] (3.5,1) -- (3.2,1) node[right=0.3cm]{\footnotesize uniqueness threshold};
  \draw (7,3.8) node{$\Gamma(\beta)$};
  \draw (4.5, -0.5) node{$\beta$};
  \draw (-0.5,4.5) node{$\gamma$};
  \draw (0,0) node[below left=0.1cm]{$(0,0)$};
  \draw [thin, dashed] (3,3) -- (3,0);
  \draw [thin, dashed] (3,3) -- (0,3);
  \draw (0,3) node[left=0.1cm]{$1$};
  \draw (3,0) node[below=0.1cm]{$1$};

\end{tikzpicture}
    \caption{This figure illustrates the rough shape of $\Gamma(\cdot)$ when there is no external field. It also includes anti-ferromagnetic range. Parameters $(\beta, \gamma)$ admit FPTAS in green region and hard to approximate in red region. }
\end{figure}

\begin{theorem}
\label{thm:spin-holant}
There is a continuous curve $\Gamma(\beta)$ defined on $[1,+\infty)$ such that (1) $\Gamma(1)=1$;
(2) $1<\Gamma(\beta)<\beta$ for all $\beta>1$; and (3) $\lim_{\beta \to +\infty} \frac{\Gamma(\beta)}{\beta}=1$.
There is an FPTAS for the  two-state spin system with local interaction matrix
    $\begin{bmatrix}
        \beta & 1 \\
        1 & \gamma
    \end{bmatrix}$
and external field $\mu \leq 1$ if $\beta\gamma>1$ and $\gamma \leq \Gamma(\beta)$.
\end{theorem}

\begin{proof}
The main idea is to make use of the holographic reduction as stated in Lemma \ref{lem:holographic} to transform 
 FPTAS for the Fibonacci function $f_k=\rho^k+\mu t^n (-\rho)^{-k}$ with edge weight $\lambda$ to a FPTAS for spin system with 
 parameters $\beta=\frac{1+\lambda \rho^2}{t(1-\lambda)}$,  $\gamma=\frac{t(1+\lambda \rho^{-2})}{1-\lambda}$ and external field $\mu$. In the following, we first choose some parameters $\rho$, $\lambda$ and $|t|=1$ in the tractable range of Theorem~\ref{thm:spin}, Theorem~\ref{thm:2}, and Theorem~\ref{thm:1} to define the boundary $\Gamma(\beta)$ by the holographic reduction. Then we cover the below area by choosing some suitable $|t|<1$.

  We first specify the boundary curve 
  %$\Gamma(\beta)$ that satisfies the conditions in the statement of the theorem and there exsits an FPTAS for two-state spin system with local interaction matrix
%      $\begin{bmatrix}
%        \beta & 1 \\
%        1 & \gamma
%    \end{bmatrix}$
%    and external field $\mu\le 1$ if $\beta\gamma>1$ and $\gamma=\Gamma(\beta)$. In fact, 
$\Gamma(\beta)=\max\left\{\Gamma_1(\beta),\Gamma_2(\beta),\Gamma_3(\beta)\right\}$ where $\Gamma_1(\beta),\Gamma_2(\beta)$ are curves parameterized by $\lambda$ and $\Gamma_3(\beta)$ is a curve parameterized by $\rho$ defined as follows.
    \begin{align*}
      \Gamma_1&=\left(\beta=\frac{1+2.92^2\lambda}{1-\lambda},\gamma=\frac{1+2.92^{-2}\lambda}{1-\lambda}\right),\quad \lambda\in(0,1);\\
      \Gamma_2&=\left(\beta=\frac{1+1.75^2\lambda}{\lambda-1},\gamma=\frac{1+1.75^{-2}\lambda}{\lambda-1}\right),\quad \lambda\in(1,\infty);\\
      \Gamma_3&=\left(\beta=\frac{1+\rho^2\lambda_2(\rho)}{\lambda_2(\rho)-1},\gamma=\frac{1+\rho^2\lambda_2(\rho)}{\lambda_2(\rho)-1}\right),\quad \rho\in(1,\infty)
    \end{align*}
    where $\lambda_2(\cdot)$ is the one in Theorem~\ref{thm:1}.

    $\Gamma_1$ is obtained from Lemma~\ref{lem:holographic} combined with Theorem~\ref{thm:spin} by taking $t=1$ and $\rho=2.92$ (equivalently $c=2.57$) as it is easy to verify that the condition $c/2\le\frac{f_{i+1}}{f_i}\le c+2/c$ hold in this case. $\Gamma_2$ is obtained from Lemma~\ref{lem:holographic} combined with Theorem~\ref{thm:2} by taking $t=-1$ and $\rho=1.75$ (equivalently $c=1.17$). $\Gamma_3$ is obtained from Lemma~\ref{lem:holographic} combined with Theorem~\ref{thm:1} by taking $t=-1$ and $\lambda= \lambda_2(\rho)$. Note that although in the statement of Theorem~\ref{thm:1}, $\lambda_2(\cdot)$ is a function of $p$ and $\rho$, $p$ is also a function of $\rho$ for fixed $t$ and $\mu$ in our case. Thus $\lambda_2(\cdot)$ is a function of $\rho$.

    Now we can discuss the shape of $\Gamma(\beta)$ on $\beta\mbox{-}\gamma$ plane. The maximum in the definition of $\Gamma(\beta)$ is achieved by $\Gamma_1,\Gamma_2,\Gamma_3$ consecutively for $\beta$ from $1$ to $\infty$.
    \begin{itemize}
    \item When $\beta$ is relatively small, $\Gamma(\beta)=\Gamma_1(\beta)$ , which starts from the point $(1,1)$. 
    \item As $\beta$ grows, $\Gamma(\beta)=\Gamma_2(\beta)$ as the slope $\frac{\Gamma_1(\beta)}{\beta}$ approaches $
    \frac{1+2.92^{-2}}{1+2.92^{2}}\approx 0.117$ while $\frac{\Gamma_2(\beta)}{\beta}$ approaches $
    \frac{1+1.75^{-2}}{1+1.75^{2}}\approx 0.3265$.  
    \item We $\beta$ is large enough, we have $\Gamma(\beta)=\Gamma_3(\beta)$ with the slope approaches $1$: $\lim_{\beta\to+\infty}\frac{\Gamma(\beta)}{\beta}=\lim_{\lambda\to 1^+,\rho\to 1^+}\frac{1+\lambda\rho^{-2}}{1+\lambda\rho^{2}}=1$.
    \end{itemize}

    It remains to prove that an FPTAS exists for $\gamma<\Gamma(\beta)$. It is easy to verify that for fixed choice of $\rho$ and $\lambda$ as above, if we choose a $t$ with the same sign but smaller absolute value, it remains in the tractable range of  Theorem~\ref{thm:spin}, Theorem~\ref{thm:2}, and Theorem~\ref{thm:1}. For any pair $(\beta, \gamma)$ with $\frac{1}{\beta} < \gamma < \Gamma(\beta)$, there exist a pair $(\beta^*, \gamma^*)$ in the curve $\Gamma$ such that $\beta \gamma =\beta^*\gamma^*$. By the definition of $\Gamma$, we know that $\beta^*=\frac{1+\lambda\rho^2}{t^*(1-\lambda)}$ and $\gamma^*=\frac{t^*(1+\lambda\rho^{-2})}{1-\lambda}$ for some $\rho,\lambda$ and $t^*=1$ or $-1$, for which the Fibonacci gates $f_k=\rho^k+\mu (t^*)^n(-\rho)^{-k}$ has an FPTAS. By our observation, we still have FPTAS if we replace $t^*$ by a $t$ with the same sign but smaller absolute value. In particular, if we choose $t=\frac{t^* \beta^*}{\beta}$, we get 
$\beta=\frac{1+\lambda\rho^2}{t(1-\lambda)}$ and $\gamma=\frac{t(1+\lambda\rho^{-2})}{1-\lambda}$. So $(\beta,\gamma)$ also admits an FPTAS by holographic reduction.

\end{proof}

}

\section{Computation Tree Recursion}
\label{sec:recursion}
In the exact polynomial time algorithm for Fibonacci gates without edge weights, one crucial property of a set of Fibonacci functions with a fixed parameter is
that it is closed when two nodes are connected together~\cite{FOCS08}. This is no longer true if we have non-trivial edge weights or when
different Fibonacci function have different parameters. However, we can still use the special property of a Fibonacci function
to decompose a vertex, which is the key property for all FPTAS algorithms in our paper.

Let $\Omega = (G(V,E), \{ F_v | v \in V \}, \{\lambda_e | e \in E\} )$ be an instance of ${\rm Holant}({\cal F}_{c_1,c_2}^{p,q}, \Lambda_{\lambda_1,\lambda_2} )$, $v\in V$ be a vertex of the instance with degree $d_1+d_2$ ($d_1, d_2\geq 1$) and $e_1, e_2, \ldots, e_{d_1+d_2}$ be its incident edges. We can construct a new Holant instance $\Omega'$:  $\Omega'$ is the same as $\Omega$
except that $v$ is decomposed into two vertices $v', v''$. $e_1, e_2, \ldots, e_{d_1}$ are connected to $v'$ and $e_{d_1+1}, e_{d_1+2}, \ldots, e_{d_1+d_2}$ are connected to $v''$. There is a new edge $e$ connecting $v'$ and $v''$.  If the function on
the original $v$ is $[f_0, f_1, \ldots, f_{d_1+d_2}]$, a Fibonacci function with parameter $c$, then the function on $v'$ is
$[f_0, f_1, \ldots, f_{d_1}]$ and the function on $v''$ is $[1,0, 1, c \ldots]$, also a Fibonacci function with parameter $c$.
The edge weight on the new edge $e$ is $1$. The functions on all other nodes and edge weights on all other edges (except the new $e$) remain the same as that in $\Omega$. We use the following notation to denote this decomposition operation $$\Omega'=D(\Omega, v, \{e_1, e_2, \ldots, e_{d_1}\}, \{e_{d_1+1}, e_{d_1+2}, \ldots, e_{d_1+d_2}\}).$$
\ifabs{Using the special property of Fibonacci function, we have the following lemma. }
{
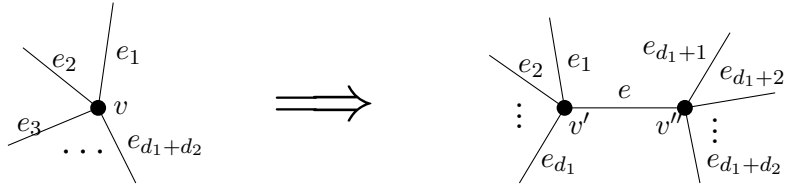
\begin{figure}[ht]
  \centering
  \begin{tikzpicture}[scale=2]
\coordinate (m) at (-1.5, 0);
\fill [black] (m) circle (1.5pt);
\draw (m) -- (-1.4,.7) node [right,midway] {$e_1$};
\draw (m) -- (-2,.4) node [right,near end] {$e_2$};
\draw (m) -- (-2.1,-.25) node [left,midway] {$e_3$};
\draw (m) -- (-1.25,-.5) node [right,midway] {$e_{d_1+d_2}$};
\node[font=\Large] at (-1.6,-.3) {$\cdots$};
\node[right=2pt of m] {$v$};

\coordinate (l) at (1.6, 0);
\coordinate (r) at (2.4, 0);
\node at (1.7,-.1) {$v'$};
\node at (2.3,-.1) {$v''$};
\fill [black] (l) circle (1.5pt);
\fill [black] (r) circle (1.5pt);
\draw (l) -- (1.5,.6) node [right,midway] {$e_1$};
\draw (l) -- (1.1,.35) node [right,near end] {$e_2$};
\draw (l) -- (1.3,-.5) node [right,near end] {$e_{d_1}$};
\node[font=\Large] at (1.3,0) {$\vdots$};
\draw (l) -- (r) node [above,midway] {$e$};
\draw (r) -- (2.7,.5) node [left,near end] {$e_{d_1+1}$};
\draw (r) -- (3,.1) node [above,near end] {$e_{d_1+2}$};
\draw (r) -- (2.5,-.5) node [right,near end] {$e_{d_1+d_2}$};
\node[font=\Large] at (2.6,-.1) {$\vdots$};

\node[font=\Huge] at (0,0) {$\Longrightarrow$};
\end{tikzpicture}
  \caption{Vertex decomposition}
\end{figure}}

\begin{lemma}\label{lem:decomposition}
  Let  $\Omega'=D(\Omega, v, E_1, E_2)$. Then $Z(\Omega)=Z(\Omega')$ and for all $e\in E$, $\mgn_\Omega(\sigma(e)=0)=\mgn_{\Omega'}(\sigma(e)=0)$.
\end{lemma}

\ifabs{}
{
\begin{proof}
 There is a natural one-to-two correspondence of configuration $\sigma$ of $\Omega$ to $\sigma'_0$ and $\sigma'_1$ of $\Omega'$:
  $\sigma'_0$ and $\sigma'_1$ are identical to $\sigma$ on $E$ while $\sigma'_0(e)=0$ and $\sigma'_1(e)=1$ for the additional edge $e$ in $\Omega'$. Then our conclusion follows from the fact that
  \[w_\Omega(\sigma)=w_{\Omega'}(\sigma'_0)+w_{\Omega'}(\sigma'_1).\]
  We verify this in the following.
  The contribution of all the other vertex function and edges weights are the same in both sides. So, we only need to verify that \[F_v(\sigma(E_1+E_2))= F_{v'}(\sigma(E_1)0) F_{v''}(\sigma(E_2)0) +F_{v'}(\sigma(E_1)0) F_{v'}(\sigma(E_1)0).\]
  or
  \[f_{|\sigma(E_1+E_2)|}=f_{|\sigma(E_1)|} g_{|\sigma(E_2)|} + f_{|\sigma(E_1)|+1} g_{|\sigma(E_2)|+1}, \]
  where $\{g_i\}$ in the Fibonacci function of $v''$. Then the above identity can be verified by the definition of $f$ and $g$.
\end{proof}}

Let $\dli$ be a dangling instance of ${\rm Holant}({\cal F}_{c_1,c_2}^p, \Lambda_{\lambda_1,\lambda_2} )$. Let $v$ be the attaching vertex of the dangling edge $e$ and $e_1, e_2, \ldots, e_d$ be other incident edges of $v$.
We compute $R(\dli)$ by smaller instances depending on $d$. If $d=0$, then $R(\dli)$ can be computed directly.
If $d=1$, we construct a smaller dangling instance $\Omega^{e_1}$ by removing $e_0$ and $v$ from $G$ and make $e_1$ be the new dangling edge and remove its weight.
\begin{equation}\label{eq:rec_single}
  R(\dli)=\frac{f_1+ \lambda_{e_1} f_2 R(\Omega^{e_1})}{f_0+ \lambda_{e_1} f_1 R(\Omega^{e_1})}.
\end{equation}
\ifabs{}
{
We define
\[h(x)=\frac{f_1+ \lambda_{e_1} f_2 x}{f_0+ \lambda_{e_1} f_1 x} \]
}

If $d\geq 2$, we use the above lemma to decompose the vertex $v$ into $v'$ and $v''$ and let $e$ and $e_1$ connect to $v''$
and the remaining edges connect to $v'$. We use $e'$ to denote the edge between $v'$ and $v''$.
By removing $e$ and $v''$ from $\Omega'$ , we get a dangling instance $\Omega^{e', e_1}$ with two dangling edges $e',e_1$.
\ifabs{}
{
\begin{figure}[ht]
  \centering
  \begin{tikzpicture}[scale=2]
\coordinate (ltop) at (-1.5, 1);
\coordinate [label=left:$v$] (lv) at (-1.5, 0);
\coordinate (l1) at (-2.1, -.8);
\coordinate (l2) at (-1.6, -.8);
\coordinate (l3) at (-.9, -.8);
\fill [black] (lv) circle (1.5pt);
\draw[thick] (ltop) -- (lv) node [right,midway] {$e$};
\draw[thick] (lv) -- (l1) node [left,midway] {$e_3$};
\draw[thick] (lv) -- (l2) node [right,midway] {$e_2$};
\draw[thick] (lv) -- (l3) node [right,midway] {$e_1$};
\draw[thick,dashed,white,gray,rounded corners=10] (-2.3,-.8) rectangle (-.7,-1.5);
\node at (-1.5, -1.7) {$\Omega^e$};

\node [font=\Huge] at (0,-.6) {$\Longrightarrow$};

\coordinate (rtop) at (1.5, 1);
\coordinate (rv) at (1.5, 0);
\coordinate [label=left:$v'$] (rt) at (1.2, -.4);
\coordinate (r1) at (.9, -.8);
\coordinate (r2) at (1.5, -.8);
\coordinate (r3) at (2.1, -.8);
\fill [black] (rv) circle (1.5pt);
\fill [black] (rt) circle (1.5pt);
\draw[thick] (rtop) -- (rv) node [right,midway] {$e$};
\draw[thick] (rv) -- (rt) node [left,midway] {$e'$};
\draw[thick] (rt) -- (r1) node [left,midway] {$e_3$};
\draw[thick] (rt) -- (r2) node [right,midway] {$e_2$};
\draw[thick] (rv) -- (r3) node [right,midway] {$e_1$};
\draw[thick,dashed,white,gray,rounded corners=10] (.7,-.8) rectangle (2.3,-1.5);
\node [right=0 of rv] {$v''$};

\node [font=\Huge] at (3,-.6) {$\Longrightarrow$};

\coordinate (rrtopl) at (4.2, .4);
\coordinate (rrv) at (4.2, -.2);
\coordinate (rr1) at (3.9, -.8);
\coordinate (rr2) at (4.5, -.8);
\coordinate (rrtopr) at (5.1, 0);
\coordinate (rr3) at (5.1, -.8);
\fill [black] (rrv) circle (1.5pt);
\draw[thick] (rrtopl) -- (rrv) node [left,midway] {$e'$};
\draw[thick] (rrv) -- (rr1) node [left,midway] {$e_3$};
\draw[thick] (rrv) -- (rr2) node [right,midway] {$e_2$};
\draw[thick] (rrtopr) -- (rr3) node [left,midway] {$e_1$};
\node [right=0 of rrv] {$v'$};
\draw[thick,dashed,white,gray,rounded corners=10] (3.7,-.8) rectangle (5.3,-1.5);
\node at (4.5, -1.7) {$\Omega^{e',e_1}$};

\end{tikzpicture}
  \caption{Vertex decomposition ($d=3$)}
\end{figure}
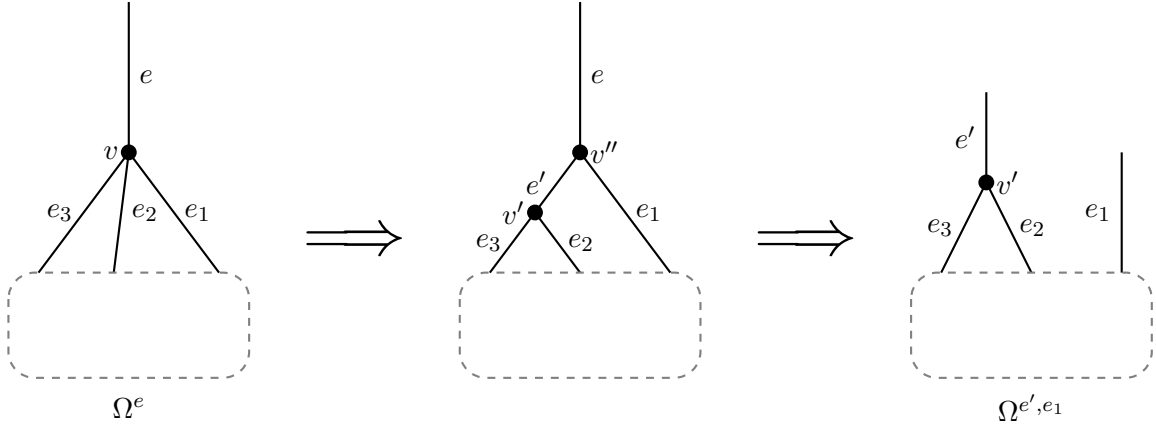
}
\begin{align*}
R(\dli)
  &=\frac{Z(\dli, \sigma(e)=1)} {Z(\dli, \sigma(e)=0)}\\
%  &=\frac{0Z(\Omega^{e', e_1}, \sigma(e' e_1)=00)+\lambda_{e_1} Z(\Omega^{e', e_1}, \sigma(e' e_1)=01)+Z(\Omega^{e', e_1}, \sigma(e' e_1)=10)+c \lambda_{e_1} Z(\Omega^{e', e_1}, \sigma(e' e_1)=11)} {Z(\Omega^{e', e_1}, \sigma(e' e_1)=00)+0 Z(\Omega^{e', e_1}, \sigma(e' e_1)=01)+0 Z(\Omega^{e', e_1}, \sigma(e' e_1)=10)+\lambda_{e_1} Z(\Omega^{e', e_1}, \sigma(e' e_1)=11)}\\
  &=\frac{\lambda_{e_1} Z(\Omega^{e', e_1}, \sigma(e' e_1)=01)+Z(\Omega^{e', e_1}, \sigma(e' e_1)=10)+c \lambda_{e_1} Z(\Omega^{e', e_1}, \sigma(e' e_1)=11)} {Z(\Omega^{e', e_1}, \sigma(e' e_1)=00)+\lambda_{e_1} Z(\Omega^{e', e_1}, \sigma(e' e_1)=11)}\\
  &=\frac{\lambda_{e_1} \mgn_{\Omega^{e', e_1}}(\sigma(e' e_1)=01)+\mgn_{\Omega^{e', e_1}}(\sigma(e' e_1)=10)+c \lambda_{e_1} \mgn_{\Omega^{e', e_1}}(\sigma(e' e_1)=11)} {\mgn_{\Omega^{e', e_1}}(\sigma(e' e_1)=00)+\lambda_{e_1} \mgn_{\Omega^{e', e_1}}(\sigma(e' e_1)=11)}.
\end{align*}

In the above recursion, the marginal probability of the original instance is written as that of smaller instances but with two dangling edges.
In order to continue the recursive process, we need to convert them into instances with single dangling edge.
This can be done by pinning one of the two dangling edges, or just leaving one of the edges free (in which case the dangling end of the free edge can be treated as a regular vertex with signature $[1,1]$).
\ifabs{We use $\Pin{\Omega}{e}{x}$ to denote the new instance obtained by pinning the edge $e$ of the instance $\Omega$ to $x$. }
{}
There are many choices in deciding which edge to pin, and to what state the edge is pinned to.
Each choice leads to different recursions and consequently have an impact on the following analysis.
Here we give an example which is used in the proof of Theorem \ref{thm:1} and Theorem \ref{thm:spin}.
In the proof of Theorem \ref{thm:2}, we use a different one.

Set $\Omega^{e'}=\Pin{\Omega^{e', e_1}}{e_1}{0}$, $\Omega^{e_1}=\Pin{\Omega^{e', e_1}}{e'}{0}$ and $\widetilde{\Omega}^{e_1}=\Pin{\Omega^{e', e_1}}{e'}{1}$. By the definitions, we have
\ifabs{$\mgn_{\Omega^{e'}}(\sigma(e')=0)=\mgn_{\Omega^{e', e_1}}(\sigma(e')=0|\sigma(e_1)=0)$,
$\mgn_{\Omega^{e_1}}(\sigma(e_1)=0)=\mgn_{\Omega^{e', e_1}}(\sigma(e_1)=0|\sigma(e')=0)$,
and $\mgn_{\widetilde{\Omega}^{e_1}}(\sigma(e_1)=0)=\mgn_{\Omega^{e', e_1}}(\sigma(e_1)=0|\sigma(e')=1)$.}
{$$\mgn_{\Omega^{e'}}(\sigma(e')=0)=\mgn_{\Omega^{e', e_1}}(\sigma(e')=0|\sigma(e_1)=0),$$
$$\mgn_{\Omega^{e_1}}(\sigma(e_1)=0)=\mgn_{\Omega^{e', e_1}}(\sigma(e_1)=0|\sigma(e')=0),$$
$$\mgn_{\widetilde{\Omega}^{e_1}}(\sigma(e_1)=0)=\mgn_{\Omega^{e', e_1}}(\sigma(e_1)=0|\sigma(e')=1).$$}
Given these relation and the fact that
\[\mgn_{\Omega^{e', e_1}}(\sigma(e' e_1)=00)+\mgn_{\Omega^{e', e_1}}(\sigma(e' e_1)=01)+\mgn_{\Omega^{e', e_1}}(\sigma(e' e_1)=10)+\mgn_{\Omega^{e', e_1}}(\sigma(e' e_1)=11)=1.\]
\ifabs{We can solve these marginal probabilities and substitute these into the above recursion to get }
{
We can solve these marginal probabilities and get
\[\mgn_{\Omega^{e', e_1}}(\sigma(e' e_1)=00)=\frac{1}{1+R(\Omega^{e'})+R(\Omega^{e_1})+R(\Omega^{e'}) R(\widetilde{\Omega}^{e_1})}.\]
\[\mgn_{\Omega^{e', e_1}}(\sigma(e' e_1)=01)=\frac{R(\Omega^{e_1})}{1+R(\Omega^{e'})+R(\Omega^{e_1})+R(\Omega^{e'}) R(\widetilde{\Omega}^{e_1})}.\]
\[\mgn_{\Omega^{e', e_1}}(\sigma(e' e_1)=10)=\frac{R(\Omega^{e'})}{1+R(\Omega^{e'})+R(\Omega^{e_1})+R(\Omega^{e'}) R(\widetilde{\Omega}^{e_1})}.\]
\[\mgn_{\Omega^{e', e_1}}(\sigma(e' e_1)=11)=\frac{R(\Omega^{e'}) R(\widetilde{\Omega}^{e_1})}{1+R(\Omega^{e'})+R(\Omega^{e_1})+R(\Omega^{e'}) R(\widetilde{\Omega}^{e_1})}.\]
Substituting these into the above recursion, we get}
\begin{equation}\label{eq:rec}
  R(\dli) =\frac{\lambda_{e_1} R(\Omega^{e_1}) +R(\Omega^{e'}) +c \lambda_{e_1} R(\Omega^{e'}) R(\widetilde{\Omega}^{e_1})} {1+\lambda_{e_1}R(\Omega^{e'}) R(\widetilde{\Omega}^{e_1})}
\end{equation}
\ifabs{}
{
We define
\[g(x,y,z)=\frac{\lambda_{e_1} y +x +c \lambda_{e_1} x z} {1+\lambda_{e_1}x z}.\]
}

If $e'$ and $e_1$ are in different connected components of  $\Omega^{e', e_1}$, then the marginal probability of $e_1$ is independent of $e'$ and as a result  $R(\widetilde{\Omega}^{e_1})=R({\Omega}^{e_1})$. So in this case, we have

\begin{equation}\label{eq:rec_separate}
  R(\dli) =\frac{\lambda_{e_1} R(\Omega^{e_1}) +R(\Omega^{e'}) +c \lambda_{e_1} R(\Omega^{e'}) R({\Omega}^{e_1})} {1+\lambda_{e_1}R(\Omega^{e'}) R({\Omega}^{e_1})}
\end{equation}
\ifabs{}
{
We define
\[\hat{g}(x,y)=\frac{\lambda_{e_1} y +x +c \lambda_{e_1} x y} {1+\lambda_{e_1}x y}.\]
}

Starting from an dangling instance $\dli$, we can compute $R(\dli) $ by one of (\ref{eq:rec_single}), (\ref{eq:rec}) and (\ref{eq:rec_separate}) recursively. We note that if $\dli \in {\rm Holant}({\cal F}_{c_1,c_2}^{p,q}, \Lambda_{\lambda_1,\lambda_2} )$, the instances involved in the recursion are also in the same family.
\ifabs{We define three functions according to these three recursions:
\[ h(x)=\frac{f_1+ \lambda_{e_1} f_2 x}{f_0+ \lambda_{e_1} f_1 x} ,\ \ \ \ \  g(x,y,z)=\frac{\lambda_{e_1} y +x +c \lambda_{e_1} x z} {1+\lambda_{e_1}x z},\ \ \ \mbox{ and }\ \ \  \hat{g}(x,y)=\frac{\lambda_{e_1} y +x +c \lambda_{e_1} x y} {1+\lambda_{e_1}x y}.\]
}
{} By expanding this recursion, we get a computation tree recursion to compute  $R(\dli)$.
We need one more step to compute the marginal probability of an edge in a regular instance.
\ifabs{
This can be done similarly and we have the following lemma.}
{Let $e=(u,v)$ be an edge in a regular instance $\Omega$. We can use Lemma \ref{lem:decomposition} to decompose vertices $u$ and $v$ in two smaller ones if their degrees are larger than three. Therefore, we can assume that the degrees of $u$ and $v$ are both less than four. In the following, we assume $d(u)=d(v)=3$. Other cases are similar and simpler. We denote the other two incident edges
of $u$ as $e_1$ and $e_2$, the other two incident edges of $v$ as $e_3$ and $e_4$.
The function on $u$ is $F_u$ and the function on $v$ is $F_v$.
We use $\Omega^D=\Omega^{e_1, e_2, e_3, e_4}$ to denote the dangling instance by removing $u$, $v$ and the edge $e=(u,v)$ from $\Omega$.
Then it follows from the definition that
\begin{align*}
  &\ \ \mgn_\Omega(\sigma(e)=0)\\
  &=\frac{\sum_{x_1,x_2, x_3, x_4\in \{0,1\}}  Z( \Omega^D, \sigma(e_1 e_2 e_3 e_4) = x_1 x_2 x_3 x_4) F_u(x_1 x_2 0) F_v(x_3 x_4 0)}{\sum_{x_1,x_2, x_3, x_4\in \{0,1\}} \left( Z( \Omega^D, \sigma(e_1 e_2 e_3 e_4) = x_1 x_2 x_3 x_4) (F_u(x_1 x_2 0) F_v(x_3 x_4 0)+ F_u(x_1 x_2 1) F_v(x_3 x_4 1) ) \right)}\\
  &=\frac{\sum_{x_1,x_2, x_3, x_4\in \{0,1\}}  \mgn_{\Omega^D}(\sigma(e_1 e_2 e_3 e_4) = x_1 x_2 x_3 x_4) F_u(x_1 x_2 0) F_v(x_3 x_4 0)}{\sum_{x_1,x_2, x_3, x_4\in \{0,1\}} \left( \mgn_{\Omega^D}(\sigma(e_1 e_2 e_3 e_4) = x_1 x_2 x_3 x_4) (F_u(x_1 x_2 0) F_v(x_3 x_4 0)+ F_u(x_1 x_2 1) F_v(x_3 x_4 1) ) \right)},\\
\end{align*}
where $\mgn_{\Omega^D}(\sigma(e_1 e_2 e_3 e_4) = x_1 x_2 x_3 x_4) $ can be further written as a product of four probability for dangling instances with one single dangling edge each.
\[\mgn_{\Omega^D}(\sigma(e_1 e_2 e_3 e_4) = x_1 x_2 x_3 x_4)=\prod_{k=1,2,3,4} \mgn_{\Omega^{D_k}}(\sigma(e_k)=x_k), \]
where $\Omega^{e_k}$ is obtained by pinning $\Omega^D$: $e_1, e_2, \ldots e_{k-1}$ are pinned to $x_1, x_2, \ldots x_{k-1}$ respectively;  $e_{k+1}, e_{k+2}, \ldots e_4$ are all pinned with weight $\frac{1}{2}$ (see them free).
Thus if we can estimate the marginal probabilities of dangling instances in sufficient precision, we can use the above relation to compute  $\mgn_\Omega(\sigma(e)=0)$. Since this recursion only involves constant many sub-instance and their derivatives are all bounded, we conclude the following lemma.
}
\begin{lemma}
  \label{lem:dangling-to-reg}
  If we can $\epsilon$ approximate $R(\dli)$ for any dangling instance $\dli$ of ${\rm Holant}({\cal F}_{c_1,c_2}^{p,q}, \Lambda_{\lambda_1,\lambda_2} )$
  in time $poly(n, \frac{1}{\epsilon})$, we can also  $\epsilon$ approximate the marginal probability of any edge of a regular instance of ${\rm Holant}({\cal F}_{c_1,c_2}^{p,q}, \Lambda_{\lambda_1,\lambda_2} )$ in time $poly(n, \frac{1}{\epsilon})$.
\end{lemma}

\section{Algorithm}\label{sec:algo}
\ifabs{The procedure from  marginal probabilities to
partition function is rather standard and we have the following lemma.}
{
The general framework of the algorithm is standard. We use the marginal probabilities to
compute the partition function and use the computation tree recursion to estimate the
marginal probabilities.}

\begin{lemma}\label{lemma:prob-to-FPTAS}
If for any $\epsilon>0$ and any $\dli$ of ${\rm Holant}\left({\cal F}_{c_1,c_2}^{p,q}, \Lambda_{\lambda_1,\lambda_2} \right)$, we have a deterministic algorithm to get $\widehat{P}$ in time $poly\left(n, \frac{1}{\epsilon}\right)$  such that $|\widehat{P}-\mgn_{\dli} (\sigma(e)=0) |\leq \epsilon $, we have an FPTAS for ${\rm Holant}({\cal F}_{c_1,c_2}^{p,q}, \Lambda_{\lambda_1,\lambda_2})$.
\end{lemma}

\ifabs{}
{
\begin{proof}
By Lemma~\ref{lem:dangling-to-reg}, if we can compute an $\epsilon$ additive approximation of the marginal probability of a dangling instance in time $poly\left(n, \frac{1}{\epsilon}\right)$ , we can also compute a $\epsilon$ additive approximation of the marginal probability of an edge in a regular instance in $poly\left(n, \frac{1}{\epsilon}\right)$, and further compute a $\frac{\epsilon}{6m}$ additive approximation in $poly\left(n, \frac{1}{\epsilon}\right)$.

The partition function can be approximated from estimations of marginal probabilities by the following standard procedure.
Let $e_1,e_2,\ldots,e_m$ be an enumeration of the edges $E$.
\begin{enumerate}
\item
Let $\Omega_1=\Omega$.
For $k=1,2,\ldots,m$, assuming that %$x_{v_1},x_{v_2},\ldots,x_{v_{k-1}}$ are specified and
the $\Omega_k$ is well-defined, use the algorithm to  compute $\hat{\mgn}_{\Omega_k}(\sigma(e_k)=0)$.
If $\hat{\mgn}_{\Omega_k}(\sigma(e_k)=0)\geq \frac{1}{2}$, set $x_k=0$; otherwise set $x_k=1$.
Construct $\Omega_{k+1}$ by pinning $e_k$ of $\Omega_k$ to $x_k$.
\item
Compute $\widehat{Z}(\Omega)=\frac{w_\Omega(x_1 x_2 \cdots x_m)}{\prod_{k=1}^m\hat{\mgn}_{\Omega_k}(\sigma(e_k)=x_k)}$ and return $\widehat{Z}(\Omega)$.
\end{enumerate}
It is clear that the running time is in $poly\left(n, \frac{1}{\epsilon}\right)$.
By the construction, we have that $\hat{\mgn}_{\Omega_k}(\sigma(e_k)=x_k)\geq \frac{1}{2}$. Since it is a $\frac{\epsilon}{6m}$ additive approximation of  $\mgn_{\Omega_k}(\sigma(e_k)=x_k)$, we have that
$\mgn_{\Omega_k}(\sigma(e_k)=x_k)> \frac{1}{3}$. Thus
\[ \frac{\mgnt{t}_{\Omega_k}(\sigma(e_k)=x_k)}{\mgn_{\Omega_k}(\sigma(e_k)=x_k)} \in [1-\frac{\epsilon}{2m}, 1+\frac{\epsilon}{2m}].  \]
By definition we have $\mgn_{\Omega}(x_1 x_2 \cdots x_m)=\frac{w_\Omega(x_1 x_2 \cdots x_m)}{Z(\Omega)}$, thus
$Z(\Omega)
=
\frac{w_\Omega(x_1 x_2 \cdots x_m)}{\prod_{k=1}^m\hat{\mgn}_{\Omega_k}(\sigma(e_k)=x_k)}$.
Therefore, we have
\[
1-\epsilon
\le
\left(1-\frac{\epsilon}{2m}\right)^m
\le
\frac{Z(\Omega)}{\widehat{Z}(\Omega)}
=
\prod_{k=1}^m\frac{\mgnt{t}_{\Omega_k}(X_{v_k}=x_{v_k})}{\mgn_{\Omega_k}(X_{v_k}=x_{v_k})}
\le
\left(1+\frac{\epsilon}{2m}\right)^m
\le
1+\epsilon,
\]
which is simplified as that $1-\epsilon\le\frac{\widehat{Z}(\Omega)}{Z(\Omega)}\le 1+\epsilon$.
This completes the proof.
\end{proof}}

Before we use the computation tree recursion to compute the marginal probability, we need the following lemma
to handle shallow instances separately.   We denote by $SP(\dli)$ the longest simple path containing $e$ in $G$.

\begin{lemma}\label{lemma:bounded-simple-path}
Let $L$ be a constant.  We have a polynomial time algorithm to compute  $R(\dli)$ for all $\dli$ of ${\rm Holant}({\cal F}_{c_1,c_2}^p, \Lambda_{\lambda_1,\lambda_2} )$ with  $SP(\dli)\leq L$.
\end{lemma}

The proof of the above Lemma uses holographic reduction to spin world and makes use of the self-avoiding walk tree~\cite{Weitz06} for two-state spin systems. The length of the longest simple path is the same as the depth of the  self-avoiding walk tree.
\ifabs{See the full version for more details.}
{In order to make the argument through,
we define an extended two state spin system to be a two state spin system where the vertex weight could be any real number and the edge function could be any (not necessary symmetric) real function. In this system, we can also define partition function as usually. By that, we can algebraically define formal marginal probability which can be any real number. Under these definitions, the technique of self-avoiding  walk tree is still valid and can be used to compute the partition function of extended two state spin systems. This conclude the following lemma.

\begin{lemma}
The partition function of extended two state spin system with bounded simple path can be computed in polynomial time.
\end{lemma}

Any instance of ${\rm Holant}\left({\cal F}_{c_1,c_2}^{p,q}, \Lambda_{\lambda_1,\lambda_2} \right)$ can be transform to an instance of extended two spin system with same partition function under holographic reduction. If we can compute the partition function, we can also compute marginal probabilities.  This proves Lemma \ref{lemma:bounded-simple-path}.
}

%

%We postpone the proof to Section \ref{sec:spin} where we use a holographic reduction to spin world and use the self-avoiding walking tree in spin world. The length of the longest simple path is the same as the depth of the  self-avoiding walking tree.
%

Now we give out formal procedure to estimate $\mgn_{\Omega^e}(\sigma(e)=0)$.
Since there is a one to one relation between $\mgn_{\Omega^e}(\sigma(e)=0)$ and  $R(\dli)$, we can define our recursion on $R(\dli)$, and at the final step we convert $R(\dli)$ back to $\mgn_{\Omega}(\sigma(e)=0)$.
%This can be done by the computation tree recursion proposed in the previous section.
Let bounds $R_1, R_2$ and depth $L$ be obtained for the family of dangling instance in the sense that
for any dangling instance with $SP(\dli)\geq L$, we have $R(\dli)\in [R_1, R_2]$.
Formally, for $t\ge 0$, the quantity $R^{t}(\dli)$ is recursively defined as follows:
\begin{itemize}
\item
If $SP(\dli)\leq 2L$, we compute $R^{t}(\dli)=R(\dli)$ by Lemma \ref{lemma:bounded-simple-path}.
\item
Else If $t=0$, let $R^{0}(\dli)=R_1$.
\item
Else If $t>0$, use one of the recursion to get  $\tilde{R}^t(\dli) = g(R^{t-1}(\Omega^{e'}), R^{t-1}(\Omega^{e_1}),R^{t-1}(\widetilde{\Omega}^{e_1})$,
$\tilde{R}^t(\dli)=h(R^{t-1}(\Omega^{e_1}))$ or  $\tilde{R}^t(\dli) = \hat{g}(R^{t-1}(\Omega^{e'}), R^{t-1}(\Omega^{e_1}))$. Return the median of $R_1, \tilde{R}^{t}(\dli), R_2$:
  $R^{t}(\dli)=Med(R_1, \tilde{R}^{t}(\dli), R_2)$.
\end{itemize}
There are three possible recursions and we define four amortized decay rates:
\ifabs{
\[\alpha_1(x) = \frac{\Phi(x) \left|\deriv{h}{x}\right|}{\Phi(h(x))},\ \ \ \alpha_3(x,y) =\frac{\left|\pderiv{\hat{g}}{x}\right|\Phi(x)}{\Phi(\hat{g}(x,y))},\ \ \ \
\alpha_4(x,y) =\frac{\left|\pderiv{\hat{g}}{y}\right|\Phi(y)}{\Phi(\hat{g}(x,y))},\]
\[ \alpha_2(x,y,z) =\frac{1}{\Phi(g(x,y,z))}\left(\left|\pderiv{g}{x}\right|\Phi(x)+ \left|\pderiv{g}{y}\right|\Phi(y)+\left|\pderiv{g}{z}\right|\Phi(z)\right),\]
}
{
\begin{align*}
\alpha_1(x) &= \frac{\Phi(x) \left|\deriv{h}{x}\right|}{\Phi(h(x))},\\
\alpha_2(x,y,z) &=\frac{1}{\Phi(g(x,y,z))}\left(\left|\pderiv{g}{x}\right|\Phi(x)+ \left|\pderiv{g}{y}\right|\Phi(y)+\left|\pderiv{g}{z}\right|\Phi(z)\right),\\
\alpha_3(x,y) &=\frac{\left|\pderiv{\hat{g}}{x}\right|\Phi(x)}{\Phi(\hat{g}(x,y))},\\
\alpha_4(x,y) &=\frac{\left|\pderiv{\hat{g}}{y}\right|\Phi(y)}{\Phi(\hat{g}(x,y))},\\
\end{align*}}
where $\Phi(\cdot)$ is a potential function.

\begin{definition}
We call a function $\Phi: (0, +\infty)\rightarrow (0, +\infty)$ \emph{nice} if there is some function $f: [1,+\infty) \to (0,+\infty)$ such that for any $c\geq 1$ and  $x,y>0$ with $\frac{x}{c}\leq y \leq c x$, we have $\frac{\Phi(x)}{\Phi(y)}\leq f(c)$.
\end{definition}
\ifabs{}
{
For any fixed constant $d$, $\Phi(x)=x^d$ is a nice function while $\Phi(x)=2^x$ is not.}
\begin{lemma}
\label{lem:algo}
Let bounds $R_1, R_2$ and depth $L$ be obtained for dangling instances of ${\rm Holant}({\cal F}_{c_1,c_2}^{p,q}, \Lambda_{\lambda_1,\lambda_2} )$ such that for any dangling instance with $SP(\dli)\geq L$, we have $R(\dli)\in [R_1, R_2]$.
If there exist a nice function $\Phi(\cdot)$  and a constant $\alpha<1$ such that
$\alpha_1(x) \leq \alpha$ for all $x \in [R_1, R_2]$, $\alpha_2(x,y,z) \leq \alpha$ for all $x,y,z \in [R_1, R_2]$,
$\alpha_3(x,y)\leq \alpha$ for all $x\in [R_1, R_2]$, and $\alpha_4(x,y)\leq \alpha$ for all $y\in [R_1, R_2]$.
Then there is an FPTAS for ${\rm Holant}({\cal F}_{c_1,c_2}^{p,q}, \Lambda_{\lambda_1,\lambda_2} )$.
\end{lemma}

\begin{proof}
By Lemma \ref{lemma:prob-to-FPTAS}, it is enough to give a $poly\left(n, \frac{1}{\epsilon}\right)$ algorithm to get $\widehat{P}$ such that $|\widehat{P}-\mgn_{\dli} (\sigma(e)=0) |\leq \epsilon $. We shall use the above recursive algorithm to compute an estimation of $R(\dli)$ and then to compute $\widehat{P}$.

Given any $\dli$ and constant $L$, we can test if $SP(\dli)<2 L$ in polynomial time.
Let $\phi=\int \frac{1}{\Phi(x)} dx$ be a monotonously increasing function.
We prove by induction that $$|\phi(R^t(\dli))-\phi(R(\dli))| \leq \alpha^t |\phi(R_1)-\phi(R(\dli))| .$$

For the base case $t=0$, if $SP(\dli)\leq 2L$, then it is trivially true since $R^0(\dli)=R(\dli)$.
Otherwise, it is also trivial since we set $R^0(\dli)=R_1$.

Now we assume that the inequality is true for $t-1$ and prove it for $t$.
If $SP(\dli)\leq 2L$, then this is trivially true since $R^t(\dli)=R(\dli)$.  Now we assume that $R^t(\dli)>2L$ and as a result it is computed by a recursion. It is enough to prove for the case that $R^{t}(\dli)=Med(R_1, \tilde{R}^{t}(\dli), R_2)=\tilde{R}^{t}(\dli)$. In other cases,
$R^{t}(\dli)$ is even closer to  $R(\dli)$ since  $R(\dli)\in [R_1, R_2]$.
There are three cases to consider:
\begin{enumerate}
  \item $R^t(\dli)=h(R^{t-1}(\Omega^{e_1}))$. If $SP(\Omega^{e_1})\leq 2L$, then by the calculation  $R^{t-1}(\Omega^{e_1})=R(\Omega^{e_1})$ and as a result $R^t(\dli)=h(R(\Omega^{e_1}))=R(\dli)$. Otherwise, we have
      that $R^{t-1}(\Omega^{e_1}), R(\Omega^{e_1})\in [R_1, R_2]$.
      \[ | \phi(R^t(\dli))-\phi(R(\dli))| = |\phi(h(R^{t-1}(\Omega^{e_1})))-\phi(h(R(\Omega^{e_1})))|
     =\frac{\Phi(x) |\deriv{h}{x}|}{\Phi(h(x))} |\phi(R^{t-1}(\Omega^{e_1})-\phi(R(\Omega^{e_1}))|,\]
     by mean value theorem, where $x$ is between $R^{t-1}(\Omega^{e_1})$ and $R(\Omega^{e_1})$ and as a result $x\in [R_1, R_2]$. By the fact that
     $\alpha_1(x) = \frac{\Phi(x) \left|\deriv{h}{x}\right|}{\Phi(h(x))}\leq \alpha$ for  $x\in [R_1, R_2]$, we get
       \[
     | \phi(R^t(\dli))-\phi(R(\dli))| \leq \alpha|\phi(R^{t-1}(\Omega^{e_1})-\phi(R(\Omega^{e_1}))| \leq \alpha^t |\phi(R_1)-\phi(R(\dli))|,\]
      where the last inequality uses induction hypothesis.
  \item $R^t(\dli) = g(R^{t-1}(\Omega^{e'}), R^{t-1}(\Omega^{e_1}), R^{t-1}(\widetilde{\Omega}^{e_1})$. In this case, we know that $e_1$ and $e'$ are connected and thus $$\frac{SP(\Omega^{e'})}{2} \leq SP(\Omega^{e_1})= SP(\widetilde{\Omega}^{e_1}) \leq 2 SP(\Omega^{e'}). $$
      If $\min\{SP(\Omega^{e_1}),SP(\widetilde{\Omega}^{e_1}), SP(\Omega^{e'})\}>L$, we know that $ R(\Omega^{e_1}),R(\widetilde{\Omega}^{e_1}), R(\Omega^{e'})\in [R_1, R_2]$ and by a similar argument as above we get that the conclusion by the fact that $\alpha_2(x,y,z) \leq \alpha,  \forall x,y,z \in [R_1, R_2]$.
       Otherwise, we have that $\max\{SP(\Omega^{e_1}),SP(\widetilde{\Omega}^{e_1}), SP(\Omega^{e'})\}\leq 2L$ and we have
       $R^{t-1}(\Omega^{e'})=R(\Omega^{e'}), R^{t-1}(\Omega^{e_1})=R(\Omega^{e_1})$, and $ R^{t-1}(\widetilde{\Omega}^{e_1})=R(\widetilde{\Omega}^{e_1})$. Therefore, we have $R^t(\dli) =R(\dli)$.
  \item  $R^t(\dli) = \hat{g}(R^{t-1}(\Omega^{e'}), R^{t-1}(\Omega^{e_1}))$. In this case,
      if $\max\{SP(\Omega^{e_1}), SP(\Omega^{e'})\}\leq 2L$, we have
       $R^{t-1}(\Omega^{e'})=R(\Omega^{e'})$ and $ R^{t-1}(\Omega^{e_1})=R(\Omega^{e_1})$.
       If $\min\{SP(\Omega^{e_1}), SP(\Omega^{e'})\}> 2L$, we know that both $R(\Omega^{e'})$ and $R(\Omega^{e_1})$ are in $[R_1,R_2]$. Then it is a weaker version of the above recursion of $g$ and we get the result.
       The remaining case is that $\min\{SP(\Omega^{e_1}), SP(\Omega^{e'})\}\leq  2L$ and $\max\{SP(\Omega^{e_1}), SP(\Omega^{e'})\}> 2L$. For that, one of  $R(\Omega^{e'})$ and $R(\Omega^{e_1})$
       is in $[R_1,R_2]$ and the other one is equal to the correct value without error. We get our conclusion by the fact that $\alpha_3(x,y)\leq \alpha, \forall x\in [R_1, R_2]$ or $\alpha_4(x,y)\leq \alpha, \forall y\in [R_1, R_2]$ respectively.
\end{enumerate}
This completes the induction proof for
\ifabs{$|\phi(R^t(\dli))-\phi(R(\dli))| \leq \alpha^t |\phi(R_1)-\phi(R(\dli))|$ .}
{
$$|\phi(R^t(\dli))-\phi(R(\dli))| \leq \alpha^t |\phi(R_1)-\phi(R(\dli))| .$$}
Since
\[|\phi(R^t(\dli))-\phi(R(\dli))|= \frac{1}{\Phi(x)}  |R^t(\dli)-R(\dli)|, \mbox{ and }\ |\phi(R_1)-\phi(R(\dli))|= \frac{1}{\Phi(y)}  |R_1-R(\dli)|,\]
for some $x,y\in [R_1, R_2]$ by the Mean Value Theorem. Given the fact that $\Phi(\cdot)$ is nice and $\frac{R_2}{R_1}$ is bounded by a constant, we conclude that there is a constant $C$ such that
$$|R^t(\dli)-R(\dli)| \leq C \alpha^t |R_1-R(\dli)| .$$
Let $\widehat{P}=\frac{1}{1+R^t(\dli)}$ then we have that
\begin{align*}
|\widehat{P} -\mgn_{\dli} (\sigma(e)=0) | &= |\frac{1}{1+R^t(\dli)} - \frac{1}{1+R(\dli)}|\\
&= \frac{|R^t(\dli)-R(\dli)|}{(1+R^t(\dli)) (1+R(\dli)) } \\
& \leq  \frac{C \alpha^t |R_1-R(\dli)|}{(1+R^t(\dli)) (1+R(\dli)) } \\
&\leq C \alpha^t .
\end{align*}
Thus by an appropriate choice of $t=O\left(\log \frac{1}{\epsilon} \right)$ , we have
$|\widehat{P} -\mgn_{\dli} (\sigma(e)=0) | \leq \epsilon $.
\end{proof}

\ifabs
{
  \input{outline}
}
{
  \section{Bounds}
In this section, we shall prove various upper and lower bounds for $R(\dli)$. These bounds are crucial to obtain
the correlation decay property and hence FPTAS. We start with the following straightforward bounds which work
for any dangling Holant instance.

\begin{lemma}
  Let $\Omega^e$ be a dangling Holant instance, $v$ be the vertex attaching $e$ and the function on $v$ be $F_v=[f_0,f_1,\dots,f_{d+1}]$. Then
  \[\min_{k=0,1,\ldots, d} \frac{f_{k+1}}{f_k} \leq R(\dli) \leq \max_{k=0,1,\ldots, d}  \frac{f_{k+1}}{f_k} .\]
\end{lemma}

\begin{proof}
  Let $D=\{e_1,e_2,\dots,e_d\}$ be other incident edges of $v$. For any fixed configuration $\pi \in \{0,1\}^D$,
  the  $R(\dli_\pi)=\frac{f_{|\pi|+1}}{f_{|\pi|}}$. Average over all the possible configurations $\pi \in \{0,1\}^D$,
  we know that that $R(\dli)$ is sandwiched between two extreme configurations.
\end{proof}

In the above argument, we used the worst configuration for the edges  $e_1,e_2,\dots,e_d$. If we already establish that
the marginal probabilities of these edges are within certain range, we can get a more accurate estimation of $R(\dli)$.
Recursively using this idea, we can get better and better bounds. This is the main approach to get better bounds in this section.

\begin{lemma}\label{lemma:rec-bound}
  If $R(\dli) \in [R_1, R_2] $ for any dangling instance $\dli$ from a family  ${\rm Holant}({\cal F}_{c_1, c_2}^{p_0}, \Lambda_{\lambda_1,\lambda_2} )$ with $SP(\dli)\geq L$. Then for a dangling instance $\dli$ of  ${\rm Holant}({\cal F}_{c_1, c_2}^p, \Lambda_{\lambda_1,\lambda_2} )$ with $SP(\dli)\geq L+1$, we have
  \[  \min_{p\geq p_0, c\in [c_1, c_2] , \lambda \in [\lambda_1, \lambda_2]  , x\in [R_1, R_2]  }  \frac{ p +(1+c p )\lambda x}{1+\lambda x p }   \leq R(\dli) \leq  \max_{p\geq p_0, c\in [c_1, c_2] , \lambda \in [\lambda_1, \lambda_2]  , x\in [R_1, R_2]  }  \frac{ p +(1+c p )\lambda x}{1+\lambda x p } . \]
\end{lemma}

\begin{proof}
  Formally, let $\Omega^D$ be the dangling instance obtained from $\Omega^e$ by removing $v$ and thus $D=\{e_1,e_2,\dots,e_d\}$ consists of $d$ dangling edges. $v_i$ is the vertex in $\Omega^D$ that attaches $e_i$ for all $1\le i\le d$.
  Without loss of generality, we assume that in one longest simple path, $e$ is followed by $e_1$.
  We define $D'=\{e_2,e_3,\dots,e_d\}$ and assume $F_v=[f_0,f_1,\dots,f_d]$.

  Then we have
  \begin{align*}
    R(\Omega^e)
    &=\frac{\mgn_{\Omega^e}(\sigma(e)=1)}{\mgn_{\Omega^e}(\sigma(e)=0)} = \frac{Z(\Omega^e,\sigma(e)=1)}{Z(\Omega^e,\sigma(e)=0)}\\
    &=\frac{\sum_{\pi\in\{0,1\}^D}\left(Z(\Omega^D,\sigma(D)=\pi)\cdot\prod_{i=1}^d\lambda_{(v,v_i)}^{\pi(e_i)}\cdot f_{\|\pi\|+1}\right)}{\sum_{\pi\in\{0,1\}^D}\left(Z(\Omega^D,\sigma(D)=\pi)\cdot\prod_{i=1}^d\lambda_{(v,v_i)}^{\pi(e_i)}\cdot f_{\|\pi\|}\right)}\\
    &=\frac{\sum_{\pi\in\{0,1\}^{D'}}\prod_{i=2}^d\lambda_{(v,v_i)}^{\pi(e_i)}\left(Z(\Omega^D,\sigma(D)=0\pi)\cdot f_{\|\pi\|+1}+\lambda_{(v,v_1)}\cdot Z(\Omega^D,\sigma(D)=1\pi)\cdot f_{\|\pi\|+2}\right)}{\sum_{\pi\in\{0,1\}^{D'}}\prod_{i=2}^d\lambda_{(v,v_i)}^{\pi(e_i)}\left(Z(\Omega^D,\sigma(D)=0\pi)\cdot f_{\|\pi\|}+\lambda_{(v,v_1)}\cdot Z(\Omega^D,\sigma(D)=1\pi)\cdot f_{\|\pi\|+1}\right)}.\\
    &=\frac{\sum_{\pi\in\{0,1\}^{D'}}\prod_{i=2}^d\lambda_{(v,v_i)}^{\pi(e_i)}\left(\mgn_{\Omega^D}(\sigma(D)=0\pi)\cdot f_{\|\pi\|+1}+\lambda_{(v,v_1)}\cdot \mgn_{\Omega^D}(\sigma(D)=1\pi)\cdot f_{\|\pi\|+2}\right)}{\sum_{\pi\in\{0,1\}^{D'}}\prod_{i=2}^d\lambda_{(v,v_i)}^{\pi(e_i)}\left(\mgn_{\Omega^D}(\sigma(D)=0\pi)\cdot f_{\|\pi\|}+\lambda_{(v,v_1)}\cdot \mgn_{\Omega^D}(\sigma(D)=1\pi)\cdot f_{\|\pi\|+1}\right)}.
  \end{align*}

  Thus
  \begin{align}
    R(\Omega^e)&\le\max_{\pi\in\{0,1\}^{D'}}\frac{\mgn_{\Omega^D}(\sigma(D)=0\pi)\cdot f_{\|\pi\|+1}+\lambda_{(v,v_1)}\cdot \mgn_{\Omega^D}(\sigma(D)=1\pi)\cdot f_{\|\pi\|+2}}{\mgn_{\Omega^D}(\sigma(D)=0\pi)\cdot f_{\|\pi\|}+\lambda_{(v,v_1)}\cdot \mgn_{\Omega^D}(\sigma(D)=1\pi)\cdot f_{\|\pi\|+1}}.\label{eq:Rmax}\\
    R(\Omega^e)&\ge\min_{\pi\in\{0,1\}^{D'}}\frac{\mgn_{\Omega^D}(\sigma(D)=0\pi)\cdot f_{\|\pi\|+1}+\lambda_{(v,v_1)}\cdot \mgn_{\Omega^D}(\sigma(D)=1\pi)\cdot f_{\|\pi\|+2}}{\mgn_{\Omega^D}(\sigma(D)=0\pi)\cdot f_{\|\pi\|}+\lambda_{(v,v_1)}\cdot \mgn_{\Omega^D}(\sigma(D)=1\pi)\cdot f_{\|\pi\|+1}}\label{eq:Rmin}.
  \end{align}

  For a fixed $\pi\in\{0,1\}^{D'}$, we can define a new dangling instance $\Omega_\pi^{e_1}$ with dangling edge $e_1$ by
  pinning the configurations of $D'$ to $\pi$. Then we have
  \[R(\Omega_\pi^{e_1})= \frac{\mgn_{\Omega^D}(\sigma(D)=1\pi)}{\mgn_{\Omega^D}(\sigma(D)=0\pi)}.\]
  By our choice of $e_1$, we have that $SP(\Omega_\pi^{e_1})\geq L$.  As a result, $R(\Omega_\pi^{e_1}) \in [R_1, R_2]$

  By the definition of Fibonacci function, we have  $f_{\|\pi\|+2}= c f_{\|\pi\|+1} +f_{\|\pi\|}$. Let $p=\frac{f_{\|\pi\|+1}}{f_{\|\pi\|}}$, we get that claimed bounds.

  % to be the one that \eqref{eq:Rmax} achives its maximum. Then
  % \begin{align*}
  %   R(\Omega^e)
  %   &\le\frac{\mgn_{\Omega^D}(\sigma(D)=0\tilde\pi)\cdot f_{\|\tilde\pi\|+1}+\lambda_{(v,v_1)}\cdot \mgn_{\Omega^D}(\sigma(D)=1\tilde\pi)\cdot f_{\|\tilde\pi\|+2}}{\mgn_{\Omega^D}(\sigma(D)=0\tilde\pi)\cdot f_{\|\tilde\pi\|}+\lambda_{(v,v_1)}\cdot \mgn_{\Omega^D}(\sigma(D)=1\tilde\pi)\cdot f_{\|\tilde\pi\|+1}}\\
  %   &=\frac{\mgn_{\Omega^D}(\sigma(e_1)=0\mid\sigma(D')=\tilde\pi)\cdot f_{\|\tilde\pi\|+1}+\lambda_{(v,v_1)}\cdot \mgn_{\Omega^D}(\sigma(e_1)=1\mid\sigma(D')=\tilde\pi)\cdot f_{\|\tilde\pi\|+2}}{\mgn_{\Omega^D}(\sigma(e_1)=0\mid\sigma(D')=\tilde\pi)\cdot f_{\|\tilde\pi\|}+\lambda_{(v,v_1)}\cdot \mgn_{\Omega^D}(\sigma(e_1)=1\mid\sigma(D')=\tilde\pi)\cdot f_{\|\tilde\pi\|+1}}\\
  % \end{align*}
\end{proof}

We denote by
\[h^c_{\lambda, p}(x) = \frac{ p +(1+c p )\lambda x}{1+\lambda x p }.\]
We use $\Hf{c_1}{c_2}{p}{\lambda_1}{\lambda_2}$ to denote the family $\{h^c_{\lambda, p }\mid c_1\leq c \leq c_2, \lambda_1\le\lambda\le\lambda_2,u\ge u_0\}$.
By recursive using Lemma \ref{lemma:rec-bound}, we can get the following bound.

\begin{lemma}\label{lemma:bound-funtion-to-instance}
  If for any $x\geq 0$ and any $h_1, h_2, \dots, h_L \in \Hf{c_1}{c_2}{p}{\lambda_1}{\lambda_2}$, we have $h_1 h_2 \cdots h_L(x) \in [R_1, R_2]$. Then
  for any dangling instance $\dli$ of  ${\rm Holant}({\cal F}_{c_1, c_2}^p, \Lambda_{\lambda_1,\lambda_2} )$ with $SP(\dli)\geq L$, we have $R(\dli)\in [R_1, R_2]$.
\end{lemma}

\subsection{For Theorem~\ref{thm:2}}

\begin{lemma}\label{lem:warm1}
  Let $c_0>0, p_0>0,  L\geq c_0^2 + \frac{c_0}{p_0}$ and $h_1, h_2, \ldots, h_L \in \{h^c_{\lambda, p} | c\geq c_0, \lambda\geq 1, p \geq p_0\}$. Then for any $x\geq 0$, we have
  \[
  h_L h_{L-1} \dots h_1(x) \geq c_0.
  \]
\end{lemma}
\begin{proof}
  We denote $x_i= h_i h_{i-1} \dots h_1 (x)$ and $x_i =h_i(x_{i-1}) =\frac{ p_i +(1+c_i p_i )\lambda_i x_{i-1}}{1+\lambda_i x_{i-1} p_i }$.

  If there exists some $x_{i-1} \lambda_i \geq c_0$, then
  \begin{align*}
    h_i(x_{i-1}) &=\frac{ p_i +(1+c_i p_i )\lambda_i x_{i-1}}{1+\lambda_i x_{i-1} p_i }\\
    & \geq \frac{ p_i +(1+c_0 p_i )\lambda_i x_{i-1}}{1+\lambda_i x_{i-1} p_i } \\
    &= c_0 + \frac{ p_i +\lambda_i x_{i-1} -c_0}{1+\lambda_i x_{i-1} p_i } \\
    &\geq c_0 + \frac{ p_i +c_0 -c_0}{1+\lambda_i x_{i-1} p_i } \\
    &> c_0.
  \end{align*}
  Then it remains above  $c_0$.

  Now we assume $x_{i-1} \lambda_i <c_0$ for all $i=1, 2, \ldots, L$. In this case, we have
  \begin{align*}
    x_{i}-x_{i-1} &= h_i(x_{i-1})  -x_{i-1} \\
    &=\frac{ p_i +(1+c_i p_i )\lambda_i x_{i-1}}{1+\lambda_i x_{i-1} p_i }-x_{i-1}\\
    & \geq \frac{ p_i +(1+c_0 p_i )\lambda_i x_{i-1} -x_{i-1} -\lambda_i x_{i-1}^2 p_i}{1+\lambda_i x_{i-1} p_i } \\
    & \geq \frac{ p_i +(\lambda_i  -1) x_{i-1} }{1+\lambda_i x_{i-1} p_i } \\
    & \geq \frac{ p_i }{1+\lambda_i x_{i-1} p_i } \\
    & \geq \frac{ p_0 }{1+c_0 p_0 } .
  \end{align*}
  So at every step, it is increased by at least $\frac{ p_0 }{1+c_0 p_0 }$. So if $L\geq \frac{c_0 (1+c_0 p_0)} {p_0}$, we can conclude that $x_L \geq c_0$.
\end{proof}

By Lemma \ref{lemma:bound-funtion-to-instance} and the above bound, we have the following bound which is used in the
proof of Theorem \ref{thm:2}.

\begin{corollary}
  Let $c_0>0, p_0>0,  L\geq c_0^2 + \frac{c_0}{p_0}$ and $\dli$ be an instance of  ${\rm Holant}({\cal F}_{c_0, \infty}^{p_0}, \Lambda_{1,\infty} )$ with $SP(\dli)\geq L$. Then $R(\dli)\geq c_0$.
\end{corollary}

\subsection{For Theorem~\ref{thm:1}}

% In this section, when we use the notation $c$, we implicitly define a varible $\rho>0$ such that $c=\rho-\frac{1}{\rho}$. We may freely use $\rho$ or $c$ for the sake of convenience.

\begin{lemma}\label{lem:warm}
  Let $h_{\lambda_1,\mu_1}, h_{\lambda_2,\mu_2}\in\Hf{c}{c}{\mu}{\lambda}{+\infty}$ be two functions, then for any $x\ge 0$,
  \[
   \min\left\{\frac{\lambda\mu^2}{1+\lambda\mu^2}\cdot c,x^*\right\}\le h_{\lambda_1,\mu_1}h_{\lambda_2,\mu_2}(x)%
   \le\max\left\{\frac{\mu+(1+c\mu)\lambda c}{\mu\lambda c},c+\frac{1}{\mu},x^*\right\}.
  \]
  where $x^*$ is the larger fixpoint of $h_{\lambda_1,\mu_1}$.
\end{lemma}
\begin{proof}
  We only prove the lower bound, the proof of the upper bound is analogous.

  If $\mu_1\ge\rho$ then the lemma obviously holds since $h_{\lambda_1,\mu_1}(x)\ge c$ for any $x\ge 0$. Thus we assume $\mu_1<\rho$, then we distinguish between two cases:
  \begin{itemize}
  \item [(1)] $\mu_2\ge\rho$, then the lemma follows from the fact that $h_{\lambda_2,\mu_2}(x)\ge c$ for any $x$ and $h_{\lambda_1,\mu_1}(x)>x$ when $x<x^*$.
  \item [(2)] $\mu_2<\rho$, then we have
    \[
    h_{\lambda_1,\mu_1}h_{\lambda_2,\mu_2}(x)\ge h_{\lambda_1,\mu_1}(\mu),
    \]
    % Consider $h_{\lambda_1,\mu_1}(\mu)$ as a function on $\mu_1$ for fixed $\lambda_1$ and $\mu$, it can be either increasing or decreasing. In both case, we have
    and thus
    \[
    h_{\lambda_1,\mu_1}(\mu)=\frac{\mu_1+(1+c\mu_1)\lambda_1\mu}{1+\lambda_1\mu_1\mu}\ge\frac{\lambda\mu^2}{1+\lambda\mu^2}\cdot c.
    \]
  \end{itemize}
\end{proof}

In the following, we say a number $x$ is \emph{warm} if $\frac{\lambda\mu^2}{1+\lambda\mu^2}\cdot c\le x\le \max\left\{\frac{\mu+(1+c\mu)\lambda c}{\mu\lambda c},c+\frac{1}{\mu}\right\}$ when we work with functions in $\Hf{c}{c}{\mu}{\lambda}{+\infty}$.

\begin{lemma}\label{lem:fpbound}
  Let $\mu,\lambda,c>0$ be three numbers, let $x^*$ be the larger fixpoint of $h_{\lambda,\mu}^c$, then
  \[
  \left|x^*-\rho\right|\le \frac{4\rho\left|\lambda-1\right|\left|\mu-\rho\right|}{(\rho^2+1)\lambda\mu}
  \]
\end{lemma}
\begin{proof}
  Solving the equation $h_{\lambda,\mu}^c(x^*)=x^*$ and taking the larger root, we obtain
  \[
  x^*=\rho+\frac{(\lambda-1)\rho-\lambda(\rho^2+1)\mu+\sqrt{\left(\lambda-1\right)^2\rho^2+\left(\lambda^2\rho^4-2\lambda(\lambda-2)\rho^2+\lambda^2\right)\mu^2+2\left(\lambda(\lambda-1)\rho(\rho^2-1)\right)\mu}}{\lambda\rho\mu}
  \]
  Take
  \begin{align*}
    A&=\lambda(\rho^2+1)\mu-(\lambda-1)\rho,\\
    B&=\left(\lambda-1\right)^2\rho^2+\left(\lambda^2\rho^4-2\lambda(\lambda-2)\rho^2+\lambda^2\right)\mu^2+2\left(\lambda(\lambda-1)\rho(\rho^2-1)\right)\mu
  \end{align*}
  Then $x^*-\rho=\frac{\sqrt{B}-A}{\lambda\rho\mu}$ and it holds that
  \begin{align*}
    B-A^2
    &=B-\left(\lambda^2(\rho^2+1)^2\mu^2-2\lambda(\lambda-1)(\rho^2+1)\mu\rho+(\lambda-1)^2\rho^2\right)\\
    &=4(\lambda-1)\lambda\rho^2\mu(\rho-\mu),\\
    \sqrt{B}+A&\ge\lambda(\rho^2+1)\mu
  \end{align*}

  Notice that if $\lambda=1$ or $\mu=\rho$, then $x^*=\rho$. We need to distinguish between four cases
  \begin{enumerate}[(1)]
  \item $\lambda>1$ and $\mu>\rho$;
  \item $\lambda>1$ and $\mu<\rho$;
  \item $\lambda<1$ and $\mu>\rho$;
  \item $\lambda<1$ and $\mu<\rho$.
  \end{enumerate}
  We only prove (1), the other cases are analogous.
  If $\lambda>1$ and $\mu>\rho$, then $x^*<\rho$ and we have
  \[
    \rho-x^*
    =\frac{A-\sqrt{B}}{\lambda\rho\mu}
    =\frac{A^2-B}{(\sqrt{B}+A)\lambda\rho\mu}
    \le\frac{4(\lambda-1)\rho(\mu-\rho)}{(\rho^2+1)\lambda\mu}
    \]
\end{proof}

\begin{lemma}\label{lem:smallmu}
  Let $h_{\lambda_0,\mu_0}\in \Hf{c}{c}{\mu}{\lambda}{+\infty}$ and $\mu_0\le\rho$. Let $k$ be a number such that $k(1+(1-k)\mu^2)<c(1-k^2)\mu^3$. Then for every warm $x$, if $\max\{|\lambda-1|,|\lambda_0-1|\}\le k$, then
  \[
  |h_{\lambda_0,\mu_0}(x)-\rho|\le \alpha_1|x-\rho|+\delta_1
  \]
  for $\alpha_1=\frac{1+k}{1+\frac{c(1-k^2)\mu^3}{1+(1-k)\mu^2}}<1$ and $\delta_1=k\rho$.
\end{lemma}
\begin{proof}
  \begin{align*}
    |h_{\lambda_0,\mu_0} - \rho|
    &=\left|\frac{\rho-\mu_0}{\rho}\cdot\frac{\lambda_0}{1+x\mu_0\lambda_0}\left(x-\rho\right)+(\lambda_0-1)\frac{\rho - \mu_0}{1+x\lambda_0\mu_0}\right|\\
    &\le\frac{1+k}{1+\frac{c(1-k^2)\mu^3}{1+(1-k)\mu^2}}\left|x-\rho\right| + k\rho.
  \end{align*}
\end{proof}

\begin{lemma}\label{lem:bigmu}
  Let $k<\rho^2-1$ be a number and $h_{\lambda_1,\mu_1},h_{\lambda_2,\mu_2}\in \Hf{c}{c}{\mu}{\lambda}{+\infty}$ where $\mu\ge\rho$. Assume $\max\{|\lambda_1-1|,|\lambda_2-1|,|\lambda-1|\}\le k$, then for every warm $x$,
  \[
  |h_{\lambda_1,\mu_1}h_{\lambda_2,\mu_2}(x)-\rho|\le \alpha_2|x-\rho|+\delta_2
  \]
  for $\alpha_2=\frac{1+k}{\rho^2} < 1$ and $\delta_2=\left(\frac{1+(1-k)\mu^2}{(1-k)^2\mu^2c}+\frac{1}{\rho}\right)k$.
\end{lemma}

\begin{proof}
  \begin{align}
    &\left|h_{\lambda_1,\mu_1}h_{\lambda_2,\mu_2}(x) - \rho\right|\notag\\
    =&\left|\frac{\rho-\mu_1}{\rho}\cdot\frac{\lambda_1}{1+h_{\lambda_2,\mu_2}(x)\mu_1\lambda_1}\left(h_{\lambda_2,\mu_2}(x)-\rho\right)+(\lambda_1-1)\frac{\rho - \mu_1}{1+h_{\lambda_2,\mu_2}(x)\lambda_1\mu_1}\right|\notag\\
    =&\left|\frac{\rho-\mu_1}{\rho}\cdot\frac{\lambda_1}{1+h_{\lambda_2,\mu_2}(x)\mu_1\lambda_1}\left(
        \frac{\rho-\mu_2}{\rho}\cdot\frac{\lambda_2}{1+x\mu_2\lambda_2}\left(x-\rho\right)+(\lambda_2-1)\frac{\rho - \mu_2}{1+x\lambda_2\mu_2}\right) + \right.\notag\\
    &\quad\;\left.(\lambda_1-1)\frac{\rho - \mu_1}{1+h_{\lambda_2,\mu_2}(x)\lambda_1\mu_1}\right|\label{eqn:twice}
  \end{align}
  Take
  \begin{align*}
    A& = \left|\frac{\lambda_1\lambda_2}{\rho^2}\cdot\frac{\mu_1-\rho}{1+h_{\lambda_2,\mu_2}(x)\mu_1\lambda_1}\cdot\frac{\mu_2-\rho}{1+x\mu_2\lambda_2}\right|,\\
    B& = \left|\frac{\rho-\mu_1}{1+h_{\lambda_2,\mu_2}(x)\mu_1\lambda_1}\right|,\\
    C& = \left|\frac{\mu_2-\rho}{1+x\mu_2\lambda_2}\cdot\frac{\mu_1-\rho}{1+h_{\lambda_2,\mu_2}(x)\mu_1\lambda_1}\cdot\frac{\lambda_1}{\rho}\right|.
  \end{align*}
  It holds that
  \begin{align*}
    A
    &\le \frac{\lambda_1\lambda_2\mu_1\mu_2}{\rho^2}\frac{\rho}{\lambda_1\rho\mu_1\mu_2+\left(\lambda_1\lambda_2\rho\mu_1+\left(\left(\lambda_1\lambda_2\left(\rho^2-1\right)\right)\mu_1+\lambda_2\rho\right)\mu_2\right)x+\rho}\\
    &\le \frac{\lambda_1\lambda_2\mu_1\mu_2}{\rho^2}\frac{1}{\lambda_1\mu_1\mu_2+1}\\
    &\le\frac{\lambda_2}{\rho^2}\le\frac{1+k}{\rho^2},\\
    B&\le\left|\frac{\mu_1}{1+h_{\lambda_2,\mu_2}(x)\mu_1\lambda_1}\right|\le\frac{1}{h_{\lambda_2,\mu_2}(x)(1-k)}\le\frac{1+(1-k)\mu^2}{(1-k)^2\mu^2c},\\
    C&=A\cdot\frac{\rho}{\lambda_2}\le\frac{1}{\rho}.
  \end{align*}
  Then
  \begin{align*}
    (\ref{eqn:twice})
    &= \left|A\cdot|x-\rho| + B\cdot(\lambda_1-1) + C\cdot(\lambda_2-1)\right|\\
    &\le |A|\cdot|x-\rho|+(|B|+|C|)\cdot k\\
    &\le\frac{1+k}{\rho^2}|x-\rho|+\left(\frac{1+(1-k)\mu^2}{(1-k)^2\mu^2c}+\frac{1}{\rho}\right)k.
  \end{align*}
  Then $\alpha_2 = \frac{1+k}{\rho^2} < 1$ and $\delta_2 = \left(\frac{1+(1-k)\mu^2}{(1-k)^2\mu^2c}+\frac{1}{\rho}\right)k$.
\end{proof}

\begin{lemma}\label{lem:criterion}
  Consider functions in $\Hf{c}{c}{\mu}{\lambda}{\lambda'}$ and define $k = \max\{|\lambda-1|,|\lambda'-1|\}$. We assume that $k$ satisfies $k<\rho^2-1$ and $k(1+(1-k)\mu^2)<c(1-k^2)\mu^3$. There exist constants $M,\delta,\alpha<1$ such that for any sequence of $d>0$ functions $h_1,h_2,\dots,h_d\in\Hf{c}{c}{\mu}{\lambda}{\lambda'}$ and any warm $x$, if the sequence satisfies one of following three criterions:
  \begin{itemize}
  \item [(1)] $d=1$ and $h_1$ has its corresponding $\mu\le\rho$;
  \item [(2)] $d\le M$ and exact $h_1$ and $h_d$ in the sequence have their corresponding $\mu>\rho$;
  \item [(3)] $d=M$ and exact $h_d$ has its corresponding $\mu>\rho$,
  \end{itemize}
  then
  \[
  \left|h_dh_{d-1}\cdots h_1(x)-\rho\right|\le\alpha\left|x-\rho\right|+\delta.
  \]
\end{lemma}
\begin{proof}
  Assume $h_i=h_{\lambda_i,\mu_i}$ for every $1\le i\le d$. We consider three criterions respectively:
  \begin{itemize}
  \item [(1)] We can take $\alpha=\alpha_1$ and $\delta=\delta_1$.
  \item [(2)] For every $1\le i\le d$, define $\gamma_i=\frac{4\rho\left|\lambda_i-1\right|\left|\mu_i-\rho\right|}{(\rho^2+1)\lambda_i\mu}$. For every $2\le i\le d-1$, $h_i$ is an increasing function, then due to Lemma \ref{lem:fpbound}, for any $x\ge 0$,
    \begin{itemize}
    \item If $x\le\rho$, then
      \[
      \min\{x,\rho-\gamma_i\}\le h_i(x)\le\rho+\gamma_i.
      \]
    \item If $x\ge\rho$, then
      \[
      \rho-\gamma_i\le h_i(x)\le\max\{x,\rho+\gamma_i\}.
      \]
    \end{itemize}
    Let $\gamma=\max_{1\le i\le d}\gamma_i$, then for any $x\ge 0$, one of following two must be true:%\cremark{I need to explain more here.}
    \begin{itemize}
    \item [(a)]
      \[
      \left|h_dh_{d-1}\dots h_1(x)-\rho\right|\le\left|h_dh_1(\rho)\right|.
      \]
    \item [(b)]
      \[
      \left|h_dh_{d-1}\dots h_1(x)-\rho\right|\le\max\left\{\left|h_d(\rho+\gamma)-\rho\right|,\left|h_d(\rho-\gamma)-\rho\right|\right\}.
      \]
    \end{itemize}
    Notice that $h_d(x)=h_{\lambda_d,\mu_d}(x)$ is monotone on $\mu_d$ for fixed $x$, thus
    \begin{align*}
      \min\left\{%
        c+\frac{1}{\lambda_d(\rho+\gamma)},%
        \frac{\mu+(1+c\mu)\lambda_d(\rho+\gamma)}{1+\mu\lambda_d(\rho+\gamma)}%
      \right\}&\le h_d(\rho+\gamma)\\
      &\le\max\left\{
        c+\frac{1}{\lambda_d(\rho+\gamma)},%
        \frac{\mu+(1+c\mu)\lambda_d(\rho+\gamma)}{1+\mu\lambda_d(\rho+\gamma)}%
      \right\},\\
      \min\left\{%
        c+\frac{1}{\lambda_d(\rho-\gamma)},%
        \frac{\mu+(1+c\mu)\lambda_d(\rho-\gamma)}{1+\mu\lambda_d(\rho-\gamma)}%
      \right\}&\le h_d(\rho-\gamma)\\
      &\le\max\left\{%
        c+\frac{1}{\lambda_d(\rho-\gamma)},%
        \frac{\mu+(1+c\mu)\lambda_d(\rho-\gamma)}{1+\mu\lambda_d(\rho-\gamma)}%
      \right\}.
    \end{align*}
    Therefore we can take $\alpha=\alpha_2$ and
    \begin{align*}
      \delta&=\max\left\{%
        \delta_2,%
        \left|\frac{1}{\lambda_d(\rho+\gamma)}-\frac{1}{\rho}\right|,%
        \left|\frac{\mu+(1+c\mu)\lambda_d(\rho+\gamma)}{1+\mu\lambda_d(\rho+\gamma)}-\rho\right|,\right.\\%
      &\left.\quad\quad\quad\quad\quad%
        \left|\frac{1}{\lambda_d(\rho-\gamma)}-\frac{1}{\rho}\right|,%
        \left|\frac{\mu+(1+c\mu)\lambda_d(\rho-\gamma)}{1+\mu\lambda_d(\rho-\gamma)}-\rho\right|
      \right\}.
    \end{align*}

  \item [(3)] Assume $h_1=h_{\lambda_1,\mu_1}$, then
    \begin{align*}
      \left|h_1(x)-\rho\right|
      &\le \left|\frac{\rho-\mu}{\rho}\cdot\frac{\lambda_1}{1+x\mu\lambda_1}\right|\left|x-\rho\right|+\left|\lambda_1-1\right|\left|\frac{\rho - \mu}{1+x\lambda_1\mu}\right|\\
      &\le \frac{1}{\rho x}\left|x-\rho\right|+\frac{k}{\lambda_1 x}.
    \end{align*}
    Let $\alpha' = \frac{1}{\rho x}$ and $M$ be the number such that $\alpha'\alpha_1^M<\alpha_1<1$, then we can take $\alpha = \alpha'\alpha^M$ and take $\delta = \alpha'\cdot\frac{\delta_1}{1-\alpha_1}+\frac{k}{\lambda_1 x}$.
  \end{itemize}
\end{proof}

Let $h_1,h_2,\dots,h_d\in\Hf{c_1}{c_2}{p}{\lambda_1}{\lambda_2}$ be a sequence of functions. If for every function $h_i$ and every $x\ge 0$, we have $|h_i(x)-\rho|\le \alpha |x_i-\rho|+\delta$ holds for some $\alpha<1$ and $\delta$, then for every $x\ge 0$,
\[
|h_dh_{d-1}\dots h_1(x)-\rho|<\alpha^d|x-\rho|+\frac{\delta}{1-\alpha}
\]
holds.

Consdier functions in $\Hf{c}{c}{p}{\lambda_1}{\lambda_2}$ and define $k=\max\left\{|\lambda_1-1|,|\lambda_2-1|\right\}$. Assume $k<1/2$, then for a sequence of functions $f_1,\dots,f_d\in\Hf{c}{c}{p}{\lambda_1}{\lambda_2}$ that satisfies one of three criterions in Lemma~\ref{lem:criterion}, it holds that for every warm $x$,
\[
|h_dh_{d-1}\dots h_1(x)-\rho|\le\alpha|x-\rho|+\delta
\]
where
\begin{align*}
  \delta
  &\le\max\left\{\Delta_1,\Delta_2,\Delta_3,\Delta_4,\Delta_5,\Delta_6\right\}
\end{align*}
and
\[
\alpha\le\max\{\alpha_1,\alpha_2\},
\]
for
\begin{align*}
  \Delta_1&=\delta_2=\left(\frac{1+(1-k)p^2}{(1-k)^2p^2c}+\frac{1}{\rho}\right)k,\\
  \Delta_2&=\frac{1}{\rho x}\cdot\frac{\delta_1}{1-\alpha_1}+\frac{k}{(1-k)x},\\
  \Delta_3&=\max_{\lambda\in[\lambda_1,\lambda_2]}\left\{\left|\frac{1}{\lambda(\rho+\gamma)}-\frac{1}{\rho}\right|\right\},\\
  \Delta_4&=\max_{\lambda\in[\lambda_1,\lambda_2]}\left\{\left|\frac{p+(1+cp)\lambda(\rho+\gamma)}{1+p\lambda(\rho+\gamma)}-\rho\right|\right\},\\
  \Delta_5&=\max_{\lambda\in[\lambda_1,\lambda_2]}\left\{\left|\frac{1}{\lambda(\rho-\gamma)}-\frac{1}{\rho}\right|\right\},\\
  \Delta_6&=\max_{\lambda\in[\lambda_1,\lambda_2]}\left\{\left|\frac{p+(1+cp)\lambda(\rho-\gamma)}{1+p\lambda(\rho-\gamma)}-\rho\right|\right\},\\
  \alpha_1&=\frac{1+k}{1+\frac{c(1-k^2)p^3}{1+(1-k)p^2}},\quad\delta_1= k\rho,\\
  \alpha_2&= \frac{1+k}{\rho^2},\\
  \gamma&\le\max_{\lambda\in[\lambda_1,\lambda_2],\mu\in[p,+\infty]}\left\{\frac{4\rho|\lambda-1||\mu-\rho|}{(\rho^2+1)\lambda\mu}\right\}\le\frac{4\rho^2k}{(1+\rho^2)(1-k)p}
\end{align*}

In the following, we shall bound $\frac{1}{1-\alpha_1}$, $\frac{1}{1-\alpha_2}$ and each $\Delta_i$ respectively. Since
\[
\alpha_1
=\frac{1+k}{1+\frac{c(1-k^2)p^3}{1+(1-k)p^2}}
\le\frac{1+k}{1+\frac{3}{4}\cdot\frac{cp^3}{1+p^2}}
=\frac{4(1+k)(1+p^2)}{4+4p^2+3cp^3},
\]
we have
\[
\frac{1}{1-\alpha_1}
\le\frac{4+4p^2+3cp^3}{3cp^4-4(1+p^2)k}.
\]
If we require that $k<\min\{\frac{cp^4}{2(1+p^2)},\frac{c^2}{2}\}$, then
\[
\frac{1}{1-\alpha_1}<\frac{4+4p^2+3cp^3}{cp^4}<\frac{11(1+p^3)(1+c)}{p^4c}.
\]
Using the fact that $\rho^2\ge c^2+1$, we have
\begin{align*}
  \frac{1}{1-\alpha_2}
  &=\frac{\rho^2}{\rho^2-k-1}\\
  &\le\frac{c^2+1}{c^2-k}\\
  &\le\frac{2(c^2+1)}{c^2}
\end{align*}

\begin{align*}
  \Delta_1
  &=\left(\frac{1+(1-k)p^2}{(1-k)^2p^2c}+\frac{1}{\rho}\right)k\\
  &\le\left(\frac{1+(1-k)p^2}{(1-k)^2p^2c}+\frac{1}{c}\right)k\\
  &=\frac{1+(1-k)p^2+(1-k)^2p^2}{(1-k)^2p^2c}\cdot k\\
  &\le\frac{1+p^2+p^2}{(1-k)^2p^2c}\cdot k\\
  &\le\frac{8(1+p^2)k}{p^2c}
\end{align*}

\begin{align*}
  \Delta_2
  &=\frac{1}{\rho x}\cdot\frac{\delta_1}{1-\alpha_1}+\frac{k}{(1-k)x}\\
  &\le\frac{k}{x(1-\alpha_1)}+\frac{k}{(1-k)x}\\
  &=\frac{k}{x}\left(\frac{1}{1-\alpha_1}+\frac{1}{1-k}\right)\\
  &\le\frac{\left(1+(1-k)p^2\right)k}{(1-k)p^2c}\left( \frac{11(1+p^3)(1+c)}{p^4c} +\frac{1}{1-k}\right)\\
  &\le\frac{2\left(1+p^2\right)k}{p^2c}\left( \frac{11(1+p^3)(1+c)}{p^4c}+2\right)\\
  &\le\frac{70(1+p^2)(1+p^4)(1+c)k}{p^6c^2}\\
  &\le\frac{210(1+p^6)(1+c)k}{p^6c^2}
\end{align*}

It follows from monotonicity that
\begin{align*}
  &\max\left\{\Delta_3,\Delta_4,\Delta_5,\Delta_6\right\}\\
  \le&\max\left\{%
    \frac{p+(1+cp)(1+k)(\rho+\gamma)}{1+p(1+k)(\rho+\gamma)}-\rho,
    \rho-\frac{p+(1+cp)(1-k)(\rho-\gamma)}{1+p(1-k)(\rho-\gamma)},\right.\\
  &\quad\quad\quad\left.
    \frac{1}{(1-k)(\rho-\gamma)}-\frac{1}{\rho},
    \frac{1}{\rho}-\frac{1}{(1+k)(\rho+\gamma)}
  \right\}.
\end{align*}
If we require that $k<\frac{p(1+\rho^2)}{16\rho}$, then we have
\begin{align*}
  \frac{p+(1+cp)(1+k)(\rho+\gamma)}{1+p(1+k)(\rho+\gamma)}-\rho
  &=\frac{p-\rho+(1+k)(\rho+\gamma)(1+cp-\rho p)}{1+p(1+k)(\rho+\gamma)}\\
  &\le\frac{k\rho+\gamma(1+k)}{p(\rho+\gamma)}\\
  &\le\frac{k\left(1+\frac{12\rho}{(1+\rho^2)p}\right)}{p}\\
  &\le\frac{12(1+p)k}{p^2}
\end{align*}

\begin{align*}
  \rho-\frac{p+(1+cp)(1-k)(\rho-\gamma)}{1+p(1-k)(\rho-\gamma)}
  &=\frac{\rho-p+(1-k)(\rho-\gamma)(\rho p-c\mu-1)}{1+p(1-k)(\rho-\gamma)}\\
  &\le\frac{2(k\rho+\gamma)}{p(\rho-\gamma)}\\
  &\le\frac{32(1+p)k}{p^2}
\end{align*}

\begin{align*}
  \frac{1}{(1-k)(\rho-\gamma)}-\frac{1}{\rho}
  &=\frac{\rho-(1-k)(\rho-\gamma)}{(1-k)\rho(\rho-\gamma)}\\
  &\le\frac{2(\gamma+k\rho)}{\rho(\rho-\gamma)}\\
  &\le\frac{32(1+p)k}{cp}
\end{align*}
\begin{align*}
  \frac{1}{\rho}-\frac{1}{(1+k)(\rho+\gamma)}
  &=\frac{(1+k)(\rho+\gamma)-\rho}{\rho(1+k)(\rho+\gamma)}\\
  &\le\frac{k\rho+\gamma(1+k)}{\rho(\rho+\gamma)}\\
  &\le\frac{12(1+p)k}{cp}
\end{align*}

Take all bounds into account, we have
\begin{align*}
  \delta\le&\max\left\{\frac{11(1+p^3)(1+c)}{p^4c}, \frac{2(c^2+1)}{c^2}\right\}\cdot\\%
  &\max\left\{%
    \frac{8(1+p^2)}{cp^2},
    \frac{210(1+p^6)(1+c)}{p^6c^2},
    \frac{32(1+p)}{p^2},
    \frac{32(1+p)}{cp}
  \right\}\cdot k%
\end{align*}

\begin{lemma}
  Let $c,p,\lambda_1,\lambda_2,L>0$ and $h_1,h_2,\dots,h_L\in\Hf{c}{c}{p}{\lambda_1}{\lambda_2}$ and define $k=\max\left\{|\lambda_1-1|,|\lambda_2-1|\right\}$. If $k<\min\left\{1/2, \frac{cp^4}{2(1+p^2)},\frac{c^2}{2}, \frac{p(1+\rho^2)}{16\rho}\right\}$, then for any warm $x$,
  \[
  h_Lh_{L-1}\dots h_1(x)\in[R_1,R_2]
  \]
  for $R_1=\rho-\Delta, R_2=\rho+\Delta$ where
  \begin{align*}
  \Delta
  &=\max\left\{\frac{11(1+p^3)(1+c)}{p^4c}, \frac{2(c^2+1)}{c^2},\right\}\cdot\max\left\{%
    \frac{8(1+p^2)}{cp^2},
    \frac{210(1+p^6)(1+c)}{p^6c^2},
    \frac{32(1+p)}{p^2},
    \frac{32(1+p)}{cp}
  \right\}\cdot k\\%
  &\quad\quad+\left(\max\{\frac{4(1+k)(1+p^2)}{4+4p^2+3cp^3},\frac{1+k^2}{1+c^2}\}\right)^{g(L)}\cdot|x-\rho|
  \end{align*}
  and $g:\mathbb{N}\to\mathbb{N}$ is a non-decreasing and unbounded function.
\end{lemma}
\begin{proof}
  The lemma follows from previous discussion and the fact that any sequence of $L$ functions can be consecutively grouped such that each group satisfies one of three criterions in Lemma~\ref{lem:criterion}. Thus
  \[
  g(L)=\min_{\mbox{a sequence of $L$ functions in $\Hf{c}{c}{p}{\lambda_1}{\lambda_2}$}}\{\mbox{number of groups in $f_L,f_{L-1},\dots,f_1$}\}.
  \]
\end{proof}

\begin{corollary}\label{cor:conclusion}
  Let $c,p,\lambda_1,\lambda_2,L>0$ and $\dli$ be an instance of  ${\rm Holant}({\cal F}_{c, c}^{p}, \Lambda_{\lambda_1,\lambda_2} )$ with $SP(\dli)\geq L+2$. Define $k=\max\left\{|\lambda_1-1|,|\lambda_2-1|\right\}$. If $k<\min\left\{1/2, \frac{cp^4}{2(1+p^2)},\frac{c^2}{2}, \frac{p(1+\rho^2)}{16\rho}\right\}$. Then $R(\dli)\in[R_1,R_2]$ for $R_1=\rho-\Delta, R_2=\rho+\Delta$ where
  \begin{align*}
  \Delta
  &=\max\left\{\frac{11(1+p^3)(1+c)}{p^4c}, \frac{2(c^2+1)}{c^2},\right\}\cdot\max\left\{%
    \frac{8(1+p^2)}{cp^2},
    \frac{210(1+p^6)(1+c)}{p^6c^2},
    \frac{32(1+p)}{p^2},
    \frac{32(1+p)}{cp}
  \right\}\cdot k\\%
  &\quad\quad+\left(\max\{\frac{4(1+k)(1+p^2)}{4+4p^2+3cp^3},\frac{1+k^2}{1+c^2}\}\right)^{g(L)}\cdot\max\left\{\left|\frac{\lambda\mu^2}{1+\lambda\mu^2}\cdot c-\rho\right|,\left|\frac{\mu+(1+c\mu)\lambda c}{\mu\lambda c}-\rho\right|,\left|\frac{1}{\mu}-\frac{1}{\rho}\right|\right\}
  \end{align*}
  and $g:\mathbb{N}\to\mathbb{N}$ is a non-decreasing and unbounded function. Moreover, it holds that $\lim_{L\to\infty,k\to 0}\Delta=0$.
\end{corollary}
\begin{proof}
  For any two functions $h_1,h_2\in\Hf{c}{c}{p}{\lambda_1}{\lambda_2}$ and any $x\ge 0$, it follows from Lemma~\ref{lem:warm} that $h_1h_2(x)$ is either warm or lies in the range $[\rho-\gamma,\rho+\gamma]$ for $\gamma=\max_{\lambda\in[\lambda_1,\lambda_2],\mu\in[p,+\infty]}\left\{\frac{4\rho|\lambda-1||\mu-\rho|}{(\rho^2+1)\lambda\mu}\right\}$.
\end{proof}

  \section{Correlation Decay}

In this section, we are going to prove Theorem~\ref{thm:1}, Theorem~\ref{thm:2} and Theorem~\ref{thm:spin} by analyzing the correlation decay property stated in Lemma~\ref{lem:algo}. To this end, we shall study the recursions discussed in Section~\ref{sec:recursion}.

\subsection{Proof of Theorem~\ref{thm:1}}

It follows from Corollary~\ref{cor:conclusion} that for every $\eta>0$, there exists $\beta(\eta)>0$ such that $k=\max\{|\lambda_1-1|,|\lambda_2-1|\}\le\beta(\eta)$ implies $\Delta=\max\{|\rho-R_1|,|\rho-R_2|\}<\eta$ by choosing $L$ sufficiently large.

We use the trivial potential function $\Phi(x)=1$ and as discussed in Section~\ref{sec:algo}, it is sufficient to bound
\[
  \alpha_1(x) = \left|\deriv{h}{x}\right|;\quad
  \alpha_2(x,y,z) =\left|\pderiv{g}{x}\right|+ \left|\pderiv{g}{y}\right|+\left|\pderiv{g}{z}\right|;\quad
  \alpha_3(x,y) =\left|\pderiv{\hat{g}}{x}\right|;\quad
  \alpha_4(x,y) =\left|\pderiv{\hat{g}}{y}\right|.
\]
where
\begin{align*}
  h(x) &=\frac{\mu+(c\mu+1)\lambda x}{1+\lambda\mu x}\\
  g(x,y,z) &=\frac{\lambda cxz+\lambda y+x}{\lambda xz+1}\\
  \hat{g}(x,y) &=\frac{\lambda c xy+\lambda y+x}{\lambda xy+1}\\
\end{align*}
for some $\mu\ge p$ and $\lambda>0$.

We shall frequently use the following equality:

\begin{fact}
  Assume $a,b,A,B,x$ are all positive numbers. If $b-Bx>0$, then $\frac{a+Ax}{b-Bx}=\frac{a}{b}+\frac{Ab+aB}{b(b-Bx)}x$.
\end{fact}

\begin{lemma}\label{lem:dx}
Let $\frac{1}{2}<\lambda<2$ and $\varepsilon<\frac{1}{4}$. If $\rho-\varepsilon<x,y,z<\rho+\varepsilon$, then $\pderiv{g}{x}\leq\frac{|\lambda-1|}{\lambda \rho^2+1}+15\varepsilon$.
\end{lemma}
\begin{proof}
  \begin{align*}
    \pderiv{g}{x}
    &= \frac{\lambda cz+1-\lambda^2yz}{(\lambda xz+1)^2} \\
    &\le \frac{\lambda(\rho-\frac{1}{\rho})(\rho+\varepsilon)+1-\lambda^2(\rho-\varepsilon)^2}
    {(\lambda(\rho-\varepsilon)^2+1)^2} \\
    &\leq \frac{(\lambda\rho^2+1)|\lambda-1|+(\lambda\rho+2\lambda^2\rho)\varepsilon}{(\lambda\rho^2+1)^2-4\rho(\lambda\rho^2+1)\lambda\varepsilon}\\
    &= \frac{|\lambda-1|}{\lambda\rho^2+1}+\frac{4\rho|\lambda-1|\lambda+(2\lambda^2\rho+\lambda\rho)}{(\lambda\rho^2+1)^2-4\rho(\lambda\rho^2+1)\lambda\varepsilon}\varepsilon \\
    &\le \frac{|\lambda-1|}{\lambda\rho^2+1}+\frac{8\rho+8\rho+2\rho}{(\lambda\rho^2+1)\rho(\lambda\rho+\frac{1}{\rho}-4\lambda\varepsilon)}\varepsilon \\
    &\le \frac{|\lambda-1|}{\lambda\rho^2+1}+\frac{18\rho}{(\lambda\rho^2+1)\rho(2\sqrt{\lambda}-\lambda)}\varepsilon \\
    &\le \frac{|\lambda-1|}{\lambda\rho^2+1}+\frac{18}{\frac{4}{5}(\frac{1}{2}+1)}\varepsilon<\frac{|\lambda-1|}{\lambda\rho^2+1}+15\varepsilon.
    \end{align*}
\end{proof}

\begin{lemma}\label{lem:mdx}
Let $\frac{1}{2}<\lambda<2$ and $\varepsilon<\frac{1}{4}$. If $\rho-\varepsilon<x,y,z<\rho+\varepsilon$, then $-\pderiv{g}{x}\leq\frac{|\lambda-1|}{\lambda \rho^2+1}+19\varepsilon$.
\end{lemma}
\begin{proof}
    \begin{align*}
        -\pderiv{g}{x}
        &= \frac{\lambda^2yz-\lambda cz-1}{(\lambda xz+1)^2} \\
        &\le \frac{\lambda^2 (\rho+\varepsilon)^2-\lambda(\rho-\frac{1}{\rho})(\rho-\varepsilon)-1}
        {(\lambda(\rho-\varepsilon)^2+1)^2} \\
        &\le \frac{(\lambda\rho^2+1)|\lambda-1|+(2\lambda^2\rho+\lambda^2+\lambda\rho)\varepsilon}{(\lambda\rho^2+1)^2-4\rho(\lambda\rho^2+1)\lambda\varepsilon}\\
        &= \frac{|\lambda-1|}{\lambda\rho^2+1}+\frac{4\rho|\lambda-1|\lambda+(2\lambda^2\rho+\lambda^2+\lambda\rho)}{(\lambda\rho^2+1)^2-4\rho(\lambda\rho^2+1)\lambda\varepsilon}\varepsilon \\
        &\le \frac{|\lambda-1|}{\lambda\rho^2+1}+\frac{8\rho+8\rho+4+2\rho}{(\lambda\rho^2+1)\rho(\lambda\rho+\frac{1}{\rho}-4\lambda\varepsilon)}\varepsilon \\
        &\le \frac{|\lambda-1|}{\lambda\rho^2+1}+\frac{18\rho+4}{(\lambda\rho^2+1)\rho(2\sqrt{\lambda}-\lambda)}\varepsilon \\
        &\le \frac{|\lambda-1|}{\lambda\rho^2+1}+\frac{22}{\frac{4}{5}(\frac{1}{2}+1)}\varepsilon<\frac{|\lambda-1|}{\lambda\rho^2+1}+19\varepsilon
      \end{align*}
\end{proof}

\begin{lemma}\label{lem:dy}
Let $\lambda<2$ and $\varepsilon<\frac{1}{2}$. If $\rho-\varepsilon<x,y,z<\rho+\varepsilon$, then $\left|\pderiv{g}{y}\right|\leq\frac{\lambda}{\lambda \rho^2+1}+3\varepsilon$.
\end{lemma}
\begin{proof}
    \begin{align*}
      \left|\pderiv{g}{y}\right|
      &= \frac{\lambda}{\lambda xz+1} \\
      &\le \frac{\lambda}{\lambda(\rho-\varepsilon)^2+1} \\
      &\le \frac{\lambda}{\lambda\rho^2+1-2\lambda\rho\varepsilon}\\
      &= \frac{\lambda}{\lambda\rho^2+1}+\frac{2\lambda^2\rho}{(\lambda\rho^2+1)(\lambda\rho^2+1-2\lambda\rho\varepsilon)}\varepsilon \\
      &\le \frac{\lambda}{\lambda\rho^2+1}+\frac{2\lambda^2\rho}{\lambda\rho^2+1}\varepsilon \\
      &= \frac{\lambda}{\lambda\rho^2+1}+\frac{2\lambda^2}{\lambda\rho+\frac{1}{\rho}}\varepsilon \\
      &\le \frac{\lambda}{\lambda\rho^2+1}+\frac{2\lambda^2}{2\sqrt{\lambda}}\varepsilon \\
      &\le \frac{\lambda}{\lambda\rho^2+1}+3\varepsilon.
    \end{align*}
\end{proof}

\begin{lemma}\label{lem:dz}
Let $\frac{1}{2}<\lambda<2$ and $\varepsilon<\frac{1}{4}$. If $\rho-\varepsilon<x,y,z<\rho+\varepsilon$, then $|\pderiv{g}{z}|\leq\frac{\lambda}{\lambda \rho^2+1}+30\varepsilon$.
\end{lemma}
\begin{proof}
    Since $\left|\pderiv{g}{z}\right|=\frac{|\lambda^2 xy+\lambda x^2-\lambda cx|}{(\lambda xz+1)^2}$ and
    \begin{align*}
      \lambda^2 xy+\lambda x^2-\lambda cx
      &=\lambda x(\lambda y+x-c)\\
      &\ge \lambda x(\frac{1}{2}(\rho-\varepsilon)+(\rho-\varepsilon)-\rho+\frac{1}{\rho}))\\
      &=\lambda x(\frac{\rho}{2}+\frac{1}{\rho}-\frac{3}{2}\varepsilon)>0,
    \end{align*}
    we have
    \begin{align*}
      \left|\pderiv{g}{z}\right| &= \frac{\lambda^2 xy+\lambda x^2-\lambda cx}{(\lambda xz+1)^2} \\
      &\le \frac{(\lambda^2+\lambda)(\rho+\varepsilon)^2-\lambda(\rho-\frac{1}{\rho})(\rho-\varepsilon)}
      {(\lambda(\rho-\varepsilon)^2+1)^2} \\
      &\le \frac{\lambda(\lambda\rho^2+1)+((\lambda^2+\lambda)(2\rho+\varepsilon)+\lambda\rho)\varepsilon}
      {(\lambda\rho^2+1)^2-4\rho(\lambda\rho^2+1)\lambda\varepsilon}\\
      &= \frac{\lambda}{\lambda\rho^2+1}+\frac{4\rho\lambda^2+(\lambda^2+\lambda)(2\rho+\varepsilon)+\lambda\rho}{(\lambda\rho^2+1)^2-4\rho(\lambda\rho^2+1)\lambda\varepsilon}\varepsilon \\
      &\le \frac{\lambda}{\lambda\rho^2+1}+\frac{16\rho+6(2\rho+1)+2\rho}{(\lambda\rho^2+1)\rho(\lambda\rho+\frac{1}{\rho}-4\lambda\varepsilon)}\varepsilon \\
      &\le \frac{\lambda}{\lambda\rho^2+1}+\frac{16\rho+6(2\rho+1)+2\rho}{(\lambda\rho^2+1)\rho(2\sqrt{\lambda}-\lambda)}\varepsilon \\
      &\le \frac{\lambda}{\lambda\rho^2+1}+\frac{30+\frac{6}{\rho}}{\frac{4}{5}(\frac{1}{2}\rho^2+1)}\varepsilon \\
      &\le \frac{\lambda}{\lambda\rho^2+1}+\frac{36}{\frac{4}{5}(\frac{1}{2}+1)}=\frac{\lambda}{\lambda\rho^2+1}+30\varepsilon.
    \end{align*}
\end{proof}

\begin{lemma}
\label{lem:gbound}
For $\frac{1}{2}<\lambda<2$ and $\varepsilon<\frac{1}{4}$, if $\rho-\varepsilon<x,y,z<\rho+\varepsilon$, then $\left|\pderiv{g}{x}\right|+\left|\pderiv{g}{y}\right|+\left|\pderiv{g}{z}\right| \leq \frac{|\lambda-1|+2\lambda}{\lambda\rho^2+1}+52\varepsilon.$
\end{lemma}
\begin{proof}
  It holds from previous lemmas that
  \[
  \left|\pderiv{g}{x}\right|+\left|\pderiv{g}{y}\right|+\left|\pderiv{g}{z}\right| = \max\left\{\pderiv{g}{x},-\pderiv{g}{x}\right\}+\left|\pderiv{g}{y}\right|+\left|\pderiv{g}{z}\right| \leq \frac{|\lambda-1|+2\lambda}{\lambda\rho^2+1}+52\varepsilon.
  \]
\end{proof}

\begin{lemma}
\label{lem:thm1g1}
Assume $\rho<2$.
If $\max\left\{\frac{1}{2},\frac{1}{2\rho-\rho^2+2},1-\beta\left(\frac{\rho-1}{208}\right)\right\}<\lambda<\min\left\{\frac{5-\rho}{6+\rho^3-3\rho^2},1+\beta\left(\frac{\rho-1}{208}\right)\right\}$, then $\left|\pderiv{g}{x}\right|+\left|\pderiv{g}{y}\right|+\left|\pderiv{g}{z}\right|\leq\frac{5-\rho}{4}<1$.
\end{lemma}
\begin{proof}
    Choose $\varepsilon$ in Lemma~\ref{lem:gbound} such that $52\varepsilon<\frac{\rho-1}{4}$.
    According to Lemma~\ref{cor:conclusion}, this can be done by setting $|\lambda-1|<\beta\left(\frac{\rho-1}{208}\right)$.
\begin{enumerate}[(1)]
    \item When $\lambda>1$, according to Lemma~\ref{lem:gbound}, we have
        \begin{align*}
            \left|\pderiv{g}{x}\right|+\left|\pderiv{g}{y}\right|+\left|\pderiv{g}{z}\right| &\leq
                \frac{3\lambda-1}{\lambda\rho^2+1}+\frac{\rho-1}{4} \\
                &\leq \frac{3\frac{5-\rho}{6+\rho^3-3\rho^2}-1}{\frac{5-\rho}{6+\rho^3-3\rho^2}\rho^2+1}+\frac{\rho-1}{4} \\
                &= \frac{5-\rho}{4}.
        \end{align*}
    \item When $\lambda\leq 1$, according to Lemma~\ref{lem:gbound}, we have
        \begin{align*}
            \left|\pderiv{g}{x}\right|+\left|\pderiv{g}{y}\right|+\left|\pderiv{g}{z}\right| &\leq
                \frac{\lambda+1}{\lambda\rho^2+1}+\frac{\rho-1}{4} \\
                &\leq \frac{\frac{1}{2\rho-\rho^2+2}+1}{\frac{1}{2\rho-\rho^2+1}\rho^2+1}+\frac{\rho-1}{4} \\
                &= \frac{5-\rho}{4}.
        \end{align*}
\end{enumerate}
\end{proof}

\begin{lemma}
\label{lem:thm1g2}
Assume $\rho\geq 2$.
If $\max\left\{\frac{1}{2},1-\beta\left(\frac{1}{416}\right)\right\}<\lambda<\min\left\{2,1+\beta\left(\frac{1}{416}\right)\right\}$, then $\left|\pderiv{g}{x}\right|+\left|\pderiv{g}{y}\right|+\left|\pderiv{g}{z}\right|\leq\frac{7}{8}$.
\end{lemma}
\begin{proof}
    Choose $\varepsilon$ in Lemma~\ref{lem:gbound} such that $52\varepsilon<\frac{1}{8}$.
    According to Lemma~\ref{cor:conclusion}, this can be done by setting $|\lambda-1|<\beta\left(\frac{1}{416}\right)$.
\begin{enumerate}[(1)]
    \item When $\lambda>1$, according to Lemma~\ref{lem:gbound}, we have
        \begin{align*}
            \left|\pderiv{g}{x}\right|+\left|\pderiv{g}{y}\right|+\left|\pderiv{g}{z}\right| &\leq
                \frac{3\lambda-1}{4\lambda+1}+\frac{1}{8}
                \leq \frac{7}{8}.
        \end{align*}
    \item When $\lambda<1$, according to Lemma~\ref{lem:gbound}, we have
        \begin{align*}
            \left|\pderiv{g}{x}\right|+\left|\pderiv{g}{y}\right|+\left|\pderiv{g}{z}\right| &\leq
                \frac{\lambda+1}{4\lambda+1}+\frac{1}{8}
                \leq \frac{\frac{1}{2}+1}{4\frac{1}{2}+1}+\frac{1}{8}<\frac{7}{8}.
        \end{align*}
\end{enumerate}
\end{proof}

\begin{lemma}\label{lem:dh}
  Let $\max\left\{\frac{1}{2},\frac{1}{\rho}\right\}<\lambda<2$ and $\varepsilon<\frac{1}{4}$. If $\rho-\varepsilon<x<\rho+\varepsilon$, then $\left|\deriv{h}{x}\right|\leq\frac{\lambda\mu\rho+\lambda+\lambda\mu^2}{1+2\lambda\mu\rho+\lambda\mu^2\rho}+64\varepsilon$.
\end{lemma}
\begin{proof}
  \begin{align*}
    \left|\deriv{h}{x}\right|
    &= \left|\frac{\lambda(c\mu+1)-\lambda\mu^2}{(1+\lambda\mu x)^2}\right| \\
    &\le \frac{\lambda\rho\mu+\lambda+\lambda\mu^2}{(1+\lambda\mu(\rho-\varepsilon))^2} \\
    &\le \frac{\lambda\mu\rho+\lambda+\lambda\mu^2}{(1+\lambda\mu\rho)^2-2\lambda\mu(1+\lambda\mu\rho)\varepsilon}\\
    &\le \frac{\lambda\mu\rho+\lambda+\lambda\mu^2}{(1+\lambda\mu\rho)^2}+\frac{(\mu\rho+1+\mu^2)2\lambda^2\mu}{(1+\lambda\mu\rho)\left(\left(1+\lambda\mu\rho\right)^2-2\lambda\mu\left(1+\lambda\mu\rho\right)\varepsilon\right)}\varepsilon \\
    &\le \frac{\lambda\mu\rho+\lambda+\lambda\mu^2}{1+2\lambda\mu\rho+\lambda\mu^2\rho}+\frac{(\mu\rho+1+\mu^2)2\lambda^2\mu}{(1+\lambda\mu\rho)(1+2\lambda\mu\rho+\lambda^2\mu^2\rho^2-\frac{1}{2}\lambda\mu-\frac{1}{2}\lambda^2\mu^2\rho)}\varepsilon \\
    &\le \frac{\lambda\mu\rho+\lambda+\lambda\mu^2}{1+2\lambda\mu\rho+\lambda\mu^2\rho}+\frac{(\mu\rho+1+\mu^2)2\lambda^2\mu}{(1+\lambda\mu\rho)\left(1+\frac{1}{2}\lambda^2\mu^2\rho\right)}\varepsilon \\
    &= \frac{\lambda\mu\rho+\lambda+\lambda\mu^2}{1+2\lambda\mu\rho+\lambda\mu^2\rho}+\frac{2\lambda^2\mu^2\rho+2\lambda^2\mu+2\lambda^2\mu^3}{\frac{1}{2}\lambda^3\mu^3\rho^2+\frac{1}{2}\lambda^2\mu^2\rho+\lambda\mu\rho+1}\varepsilon \\
    &\le \frac{\lambda\mu\rho+\lambda+\lambda\mu^2}{1+2\lambda\mu\rho+\lambda\mu^2\rho}+\frac{8\mu^2\rho+8\mu+8\mu^3}{\frac{1}{8}\mu^2\rho+\mu+\frac{1}{4}\mu^3}\varepsilon \\
    &\le \frac{\lambda\mu\rho+\lambda+\lambda\mu^2}{1+2\lambda\mu\rho+\lambda\mu^2\rho}+64\varepsilon
    \le \max\left\{\frac{1}{2},\frac{\lambda\mu+\lambda+\lambda\mu^2}{1+2\lambda\mu+\lambda\mu^2}\right\}+64\varepsilon
  \end{align*}
\end{proof}

\begin{lemma}\label{lem:dht1}
Assume $0<\mu<1$.%
If $\max\left\{\frac{1}{2},\frac{1}{\rho},1-\beta\left(\frac{\mu}{1280}\right)\right\}<\lambda<\min\left\{2,\frac{10-\mu}{\mu^3+2\mu^2-10\mu+10},1+\beta\left(\frac{\mu}{1280}\right)\right\}$, then $\left|\deriv{h}{x}\right|\leq 1-\frac{\mu}{20}<1$.
\end{lemma}
\begin{proof}
  Choose $\varepsilon$ in Lemma~\ref{lem:dh} such that $64\varepsilon<\frac{\mu}{20}$, according to Lemma~\ref{cor:conclusion}, this can be done by setting $|\lambda-1|<\beta\left(\frac{\mu}{1280}\right)$.
    \begin{align*}
        \left|\deriv{h}{x}\right| &\leq \max\left\{\frac{1}{2},\frac{\lambda\mu+\lambda+\lambda\mu^2}{1+2\lambda\mu+\lambda\mu^2}\right\}+\frac{\mu}{20} \\
            &\leq \max\left\{\frac{1}{2},\frac{\frac{10-\mu}{\mu^3+2\mu^2-10\mu+10}\left(\mu+1+\mu^2\right)}{1+2\frac{10-\mu}{\mu^3+2\mu^2-10\mu+10}\mu+\frac{10-\mu}{\mu^3+2\mu^2-10\mu+10}\mu^2}\right\}+\frac{\mu}{20} \\
            &= 1-\frac{\mu}{20}.
    \end{align*}
\end{proof}

\begin{lemma}\label{lem:dht2}
Assume $\mu\ge 1$.
If $\max\left\{\frac{1}{2},\frac{1}{\rho},1-\beta\left(\frac{\rho-1}{256\rho}\right)\right\}<\lambda<\min\left\{2,1+\beta\left(\frac{\rho-1}{512\rho}\right)\right\}$, then $\left|\deriv{h}{x}\right|\leq \frac{3\rho+1}{4\rho}<1$.
\end{lemma}
\begin{proof}
  According to Lemma~\ref{lem:dh},
  \begin{align*}
    \left|\deriv{h}{x}\right|\leq \frac{\mu\rho+1+\mu^2}{2\mu\rho+\mu^2\rho}+64x\varepsilon.
  \end{align*}
  Set $\alpha=\frac{\rho+1}{2\rho}$. Since
  \begin{align*}
    \frac{\mu\rho+1+\mu^2}{2\mu\rho+\mu^2\rho}-\alpha &= \frac{(\mu\rho+1-2\alpha\mu\rho)+(\mu^2-\alpha\mu^2\rho)}{2\mu\rho+\mu^2\rho}<0,
  \end{align*}
  if we choose $\varepsilon$ in Lemma~\ref{lem:dh} such that $128\varepsilon<\frac{1-\alpha}{2}$, then
  \begin{align*}
    \left|\deriv{h}{x}\right|\leq \alpha+\frac{1-\alpha}{2}=\frac{3\rho+1}{4\rho}.
  \end{align*}
\end{proof}

\begin{lemma}\label{lem:dghx}
  Let $\max\left\{\frac{1}{2},\frac{1}{\rho}\right\}<\lambda<2$ and $\varepsilon<\frac{1}{4}$. If $\rho-\varepsilon<x<\rho+\varepsilon$, then $\left|\pderiv{\hat{g}}{x}\right|\leq\frac{\lambda\rho y+\lambda^2 y^2+1}{1+2\lambda\rho y+\lambda^2\rho^2 y^2}+256\varepsilon$.
\end{lemma}
\begin{proof}
    \begin{align*}
      \left|\pderiv{\hat{g}}{x}\right|
      &\le \frac{\lambda\rho y+\lambda^2 y^2+1}{\left(\lambda\left(\rho-\varepsilon\right)y+1\right)^2} \\
      &\le \frac{\lambda\rho y+\lambda^2 y^2+1}{\lambda^2\rho^2y^2+2\lambda\rho y+1-2\lambda y(\lambda\rho y+1)\varepsilon}\\
      &\leq \frac{\lambda^2\rho^2 y+\lambda^2 y^2+1}{\lambda^2\rho^2 y^2+2\lambda\rho y+1}
      +\frac{(\lambda\rho y+\lambda^2 y^2+1)2\lambda y}{\left(\lambda\rho y+1\right)\left(\lambda^2\rho^2 y^2+2\lambda\rho y+1-2\lambda y\left(\lambda\rho y+1\right)\varepsilon\right)}\varepsilon \\
      &\le \frac{\lambda\rho y+\lambda^2 y^2+1}{\lambda^2\rho^2 y^2+2\lambda\rho y+1}
      +\frac{(\lambda\rho y+\lambda^2 y^2+1)2\lambda y}{\left(\lambda\rho y+1\right)\left(\lambda^2\rho^2 y^2+2\lambda\rho y+1-\frac{1}{2}\lambda^2\rho y^2-\frac{1}{2}\lambda y\right)}\varepsilon \\
      &\le \frac{\lambda\rho y+\lambda^2 y^2+1}{\lambda^2\rho^2 y^2+2\lambda\rho y+1}
      +\frac{(\lambda\rho y+\lambda^2 y^2+1)2\lambda y}{\left(\lambda\rho y+1\right)\left(1+\lambda\rho y+\lambda^2 \rho y^2\left(\rho-\frac{1}{2}\right)\right)}\varepsilon \\
      &\le \frac{\lambda\rho y+\lambda^2 y^2+1}{\lambda^2\rho^2 y^2+2\lambda\rho y+1}
      +\frac{(\lambda\rho y+\lambda^2 y^2+1)2\lambda y}{\left(\lambda\rho y+1\right)\left(1+\frac{1}{2}\lambda^2\rho y^2\right)}\varepsilon \\
      &= \frac{\lambda\rho y+\lambda^2 y^2+1}{\lambda^2\rho^2 y^2+2\lambda\rho y+1}
      +\frac{2\lambda^2\rho y^2+2\lambda^3 y^3+2\lambda y}{\lambda\rho y+1+\frac{1}{2}\lambda^3\rho^2 y^3+\frac{1}{2}\lambda^2\rho y^2}\varepsilon \\
      &= \frac{\lambda\rho y+\lambda^2 y^2+1}{\lambda^2\rho^2 y^2+2\lambda\rho y+1}
      +\frac{8\rho y^2+16y^3+4y}{\frac{1}{8}\rho y^2+\frac{1}{16}y^3+\frac{1}{2}y}\varepsilon \\
      &= \frac{\lambda\rho y+\lambda^2 y^2+1}{\lambda^2\rho^2 y^2+2\lambda\rho y+1}+256\varepsilon.
    \end{align*}
\end{proof}

In the following, we fix $p$ as a nonnegative number, then we have:

\begin{lemma}\label{lem:thm1ghx}
  If $\max\left\{\frac{1}{2},\frac{1}{\rho},1-\beta\left(\frac{1}{512}\min\left\{\frac{\rho^2-1}{2\rho^2},\frac{p}{p+1}\right\}\right)\right\}<\lambda<\min\left\{2,1+\beta\left(\frac{1}{512}\min\left\{\frac{\rho^2-1}{2\rho^2},\frac{p}{p+1}\right\}\right)\right\}$ and $y\geq p$, then $\left|\pderiv{\hat{g}}{x}\right|<\max\left\{\frac{\rho^2+1}{2\rho^2},\frac{3p+4}{4p+4}\right\}<1$.
\end{lemma}
\begin{proof}
  Set $\alpha=\max\left\{\frac{1}{\rho^2},\frac{p+2}{2p+2}\right\}$. Since
  $$\frac{\lambda\rho y+\lambda^2 y^2+1}{\lambda^2\rho^2 y^2+2\lambda \rho y+1}-\alpha
  =\frac{(\lambda\rho y+1-2\alpha\lambda\rho y-\alpha)+(\lambda^2 y^2-\alpha\lambda^2\rho^2 y^2)}{\lambda\rho y^2+2\lambda\rho y+1}<0,$$
  if we choose $\varepsilon$ such that $256\varepsilon<\frac{1-\alpha}{2}$, then according to Lemma~\ref{lem:dghx}, we have
  $$\left|\pderiv{\hat{g}}{x}\right|\leq \alpha+\frac{1-\alpha}{2}=\frac{1+\alpha}{2}=\max\left\{\frac{\rho^2+1}{2\rho^2+1},\frac{3p+4}{4p+4}\right\}<1.$$
\end{proof}

It follows from Lemma~\ref{lem:dht1} and Lemma~\ref{lem:dht2} that

\begin{lemma}\label{lem:dh1}
Assume $p<1$.
If $\max\left\{\frac{1}{2},\frac{1}{\rho},1-\beta\left(\frac{p}{1280}\right)\right\}<\lambda<\min\left\{2,\frac{10-p}{p^3+2p^2-10p+10},1+\beta\left(\frac{p}{1280}\right)\right\}$, then $\left|\deriv{h}{x}\right|\leq \max\left\{\frac{3\rho+1}{4\rho},1-\frac{p}{20}\right\}<1$.
\end{lemma}

\begin{lemma}\label{lem:dh2}
Assume $p\geq 1$.
If $\max\left\{\frac{1}{2},\frac{1}{\rho},1-\beta\left(\frac{\rho-1}{512\rho}\right)\right\}<\lambda<\min\left\{2,1+\beta\left(\frac{\rho-1}{512\rho}\right)\right\}$, then $\left|\deriv{h}{x}\right|\leq \frac{3\rho+1}{4\rho}<1$.
\end{lemma}

If we define $h(\mu,x):=\frac{\mu+(c\mu+1)\lambda x}{1+\lambda\mu x}=h(x)$, then $\tilde g(x,y)=h(x,y)$. Thus $\pderiv{\tilde h(x,y)}{y}$. Then according to Lemma~\ref{lem:dh1} and Lemma~\ref{lem:dh2}, we have

\begin{lemma}\label{lem:thm1ghy1}
Assume $0<p<1$.%
If $\max\left\{\frac{1}{2},\frac{1}{\rho},1-\beta\left(\frac{p}{1280}\right)\right\}<\lambda<\min\left\{2,\frac{10-p}{p^3+2p^2-10p+10},1+\beta\left(\frac{p}{1280}\right)\right\}$ and $x\ge p$, then $\left|\deriv{\hat{g}}{y}\right|\leq \max\left\{\frac{3\rho+1}{4\rho},1-\frac{p}{20}\right\}<1$.
\end{lemma}

\begin{lemma}\label{lem:thm1ghy2}
Assume $p\geq 1$.%
If $\max\left\{\frac{1}{2},\frac{1}{\rho},1-\beta\left(\frac{\rho-1}{512\rho}\right)\right\}<\lambda<\min\left\{2,1+\beta\left(\frac{\rho-1}{512\rho}\right)\right\}$ and $\ge p$, then $\left|\pderiv{\hat{g}}{y}\right|\leq \frac{3\rho+1}{4\rho}<1$.
\end{lemma}

\bigskip
%%% Proof of Theorem 1

\begin{proof}[Proof of Theorem~\ref{thm:1}]
According to Lemma~\ref{lem:algo} and Lemma~\ref{lem:thm1g1},~\ref{lem:thm1g2},~\ref{lem:dh1},~\ref{lem:dh2},~\ref{lem:thm1ghx},~\ref{lem:thm1ghy1},~\ref{lem:thm1ghy2}, we have the following results.
\begin{enumerate}[C{a}se I.]
    \item If $0<p<1$ and $1<\rho<2$, then
        \begin{align*}
            \lambda_1=\max\left\{\frac{1}{2},\frac{1}{2\rho-\rho^2+2},1-\beta\left(\frac{\rho-1}{208}\right),
                \frac{1}{\rho},1-\beta\left(\frac{p}{1280}\right), \right.\\
                \left. 1-\beta\left(\frac{1}{512}\min\left\{\frac{\rho^2-1}{2\rho^2},\frac{p}{p+1}\right\}\right)\right\},
        \end{align*}
        \begin{align*}
            \lambda_2=\min\left\{\frac{5-\rho}{6+\rho^3-3\rho^2},1+\beta\left(\frac{\rho-1}{208}\right),
                \frac{10-p}{p^3+2p^2-10p+10}, \right. \\ \left. 1+\beta\left(\frac{p}{1280}\right),
                1+\beta\left(\frac{1}{512}\min\left\{\frac{\rho^2-1}{2\rho^2},\frac{p}{p+1}\right\}\right)\right\},
        \end{align*}
        \begin{align*}
            \alpha=\max\left\{\frac{5-\rho}{4},
                \frac{3\rho+1}{4\rho},1-\frac{p}{20},
                \frac{\rho^2+1}{2\rho^2},\frac{3p+4}{4p+4}\right\}.
        \end{align*}
    \item If $0<p<1$ and $\rho \geq 2$, then
        \begin{align*}
            \lambda_1=\max\left\{\frac{1}{2},1-\beta\left(\frac{1}{416}\right),
                \frac{1}{\rho},1-\beta\left(\frac{p}{1280}\right),
                1-\beta\left(\frac{1}{512}\min\left\{\frac{\rho^2-1}{2\rho^2},\frac{p}{p+1}\right\}\right)\right\},
        \end{align*}
        \begin{align*}
            \lambda_2=\min\left\{2,1+\beta\left(\frac{1}{416}\right),
                \frac{10-p}{p^3+2p^2-10p+10}, 1+\beta\left(\frac{p}{1280}\right), \right. \\ \left.
                1+\beta\left(\frac{1}{512}\min\left\{\frac{\rho^2-1}{2\rho^2},\frac{p}{p+1}\right\}\right)\right\},
        \end{align*}
        \begin{align*}
            \alpha=\max\left\{\frac{7}{8},
                \frac{3\rho+1}{4\rho},1-\frac{p}{20},
                \frac{\rho^2+1}{2\rho^2},\frac{3p+4}{4p+4}\right\}.
        \end{align*}
    \item If $p \geq 1$ and $1<\rho<2$, then
        \begin{align*}
            \lambda_1=\max\left\{\frac{1}{2},\frac{1}{2\rho-\rho^2+2},1-\beta\left(\frac{\rho-1}{208}\right),
                \frac{1}{\rho},1-\beta\left(\frac{\rho-1}{512\rho}\right), \right. \\
                \left. 1-\beta\left(\frac{1}{512}\min\left\{\frac{\rho^2-1}{2\rho^2},\frac{p}{p+1}\right\}\right)\right\},
        \end{align*}
        \begin{align*}
            \lambda_2=\min\left\{\frac{5-\rho}{6+\rho^3-3\rho^2},1+\beta\left(\frac{\rho-1}{208}\right),
                1+\beta\left(\frac{\rho-1}{512\rho}\right), \right. \\ \left.
                1+\beta\left(\frac{1}{512}\min\left\{\frac{\rho^2-1}{2\rho^2},\frac{p}{p+1}\right\}\right)\right\},
        \end{align*}
        \begin{align*}
            \alpha=\max\left\{\frac{5-\rho}{4},
                \frac{3\rho+1}{4\rho},
                \frac{\rho^2+1}{2\rho^2},\frac{3p+4}{4p+4}\right\}.
        \end{align*}
    \item If $p \geq 1$ and $\rho \geq 2$, then
        \begin{align*}
            \lambda_1=\max\left\{\frac{1}{2},1-\beta\left(\frac{1}{416}\right),
                \frac{1}{\rho},1-\beta\left(\frac{\rho-1}{512\rho}\right), \right. \\
                \left. 1-\beta\left(\frac{1}{512}\min\left\{\frac{\rho^2-1}{2\rho^2},\frac{p}{p+1}\right\}\right)\right\},
        \end{align*}
        \begin{align*}
            \lambda_2=\min\left\{2,1+\beta\left(\frac{1}{416}\right),
                1+\beta\left(\frac{\rho-1}{512\rho}\right), \right. \\ \left.
                1+\beta\left(\frac{1}{512}\min\left\{\frac{\rho^2-1}{2\rho^2},\frac{p}{p+1}\right\}\right)\right\},
        \end{align*}
        \begin{align*}
            \alpha=\max\left\{\frac{7}{8},
                \frac{3\rho+1}{4\rho},
                \frac{\rho^2+1}{2\rho^2},\frac{3p+4}{4p+4}\right\}.
        \end{align*}
\end{enumerate}
Here the $\beta\left(\cdot\right)$ is defined in the beginning of the section.
\end{proof}

  \subsection{Proof of Theorem~\ref{thm:2}}

In order to use Lemma~\ref{lem:algo}, we need to establish four inequalities of the form $\frac{A}{B}\leq \alpha<1$. It turns out that each multivariable polynomial $A-\alpha B$ enjoys the property that the highest degrees of variables $x,y,z,\lambda$ are no greater than two. Therefore it is possible to determine the monotonicity of each variable within the given range. 
We shall show that it is decreasing with respect to variables $x,y,z,\lambda$ respectively and verify the fact that $A-\alpha B |_{x=y=z=c,\lambda=1} < 0$.

In the following proof, we use $\Phi(x)=x$ as potential function and a different set of recursions from those used in the proof of Theorem~\ref{thm:1}.
These recursions can be obtained by the same methods proposed in Section~\ref{sec:recursion} except for one step:
when converting an instance with two dangling edges to one with single dangling edge, we define three different sub-instances.
Let $\Omega^{e'}$ denote the sub-instance of $\Omega^{e',e_1}$ achieved by leaving $e_1$ free (which is equivalent to attaching a vertex with signature $[1,1]$ on the dangling end of $e_1$), and 
set $\Omega^{e_1}=\Pin{\Omega^{e', e_1}}{e'}{0}$ and $\widetilde{\Omega}^{e_1}=\Pin{\Omega^{e', e_1}}{e'}{1}$.
Now we have
$$\mgn_{\Omega^{e'}}(\sigma(e')=0)=\mgn_{\Omega^{e', e_1}}(\sigma(e')=0),$$
$$\mgn_{\Omega^{e_1}}(\sigma(e_1)=0)=\mgn_{\Omega^{e', e_1}}(\sigma(e_1)=0|\sigma(e')=0),$$
$$\mgn_{\widetilde{\Omega}^{e_1}}(\sigma(e_1)=0)=\mgn_{\Omega^{e', e_1}}(\sigma(e_1)=0|\sigma(e')=1).$$
Applying the remaining steps in Section~\ref{sec:recursion} gives the following recursions:
\begin{align*}
  g(x,y,z)&=\frac{x(1+y)+\lambda y(1+z)+\lambda c x(1+y)z}{1+z+\lambda x(1+y)z} \\
  \frac{\partial g}{\partial x}&=-\frac{(y+1)(z+1)(-c\lambda z+\lambda^2 yz-1)}{(1+z+\lambda x(1+y)z)^2} \\
  \frac{\partial g}{\partial y}&=\frac{(z+1)(\lambda(cxz+z+1)+\lambda^2 xz+x)}{(1+z+\lambda x(1+y)z)^2} \\
  \frac{\partial g}{\partial z}&=-\frac{x(y+1)(\lambda(x-c+y(\lambda+x))+1)}{(1+z+\lambda x(1+y)z)^2}
\end{align*}
$h(x)$ and $\hat{g}(x,y)$ are the same as they were in the last part.

In the following, let $c_0=1.17$ be a constant.

Note that $\frac{\partial g}{\partial y}(x,y,z) \geq 0$ and $\frac{\partial g}{\partial z}(x,y,z) \leq 0$ for $x,y,z\geq c$ and $\lambda \geq 1$. Let
$$g_1(x,y,z)=\frac{\frac{\partial g}{\partial x}(x,y,z)x+\frac{\partial g}{\partial y}(x,y,z)y-\frac{\partial g}{\partial z}(x,y,z)z}{g(x,y,z)},$$
$$g_2(x,y,z)=\frac{-\frac{\partial g}{\partial x}(x,y,z)x+\frac{\partial g}{\partial y}(x,y,z)y-\frac{\partial g}{\partial z}(x,y,z)z}{g(x,y,z)},$$
it is clear that
$$\frac{\left|\frac{\partial g}{\partial x}(x,y,z)\right|\Phi(x)+\left|\frac{\partial g}{\partial y}(x,y,z)\right|\Phi(y)+\left|\frac{\partial g}{\partial z}(x,y,z)\right|\Phi(z)}{\Phi(g(x,y,z))}=\max\left\{g_1(x,y,z),g_2(x,y,z)\right\}.$$
We shall bound $g_1$ and $g_2$ separately. 
\begin{lemma}
\label{lem:22}
$2 (22 c^3+c^2+c+1)+22(-94 c^3-54 c^2+c) \leq 0$ for $c \geq c_0$.
\end{lemma}

\begin{lemma}
\label{lem:17}
$-36 c^2 \lambda^2+2 c (-94 c^2 \lambda^2+22 c^2 \lambda-36 c \lambda^2)+2 c^2 \lambda+2 c (\lambda^2+\lambda+21)+2 \lambda \leq 0$ for $\lambda \geq 1$ and $c \geq c_0$.
\end{lemma}
\begin{proof}
\begin{align*}
    & -36 c^2 \lambda^2+2 c (-94 c^2 \lambda^2+22 c^2 \lambda-36 c \lambda^2)+2 c^2 \lambda+2 c (\lambda^2+\lambda+21)+2 \lambda \\
    =& 2 (-94 c^3-54 c^2+c) \lambda^2+2 (22 c^3+c^2+c+1) \lambda+42 c
\end{align*}
which is a parabola of $\lambda$ and according to Lemma~\ref{lem:22}, its center is to the left of $1$.
Therefore it is decreasing with $\lambda$, and if we set $\lambda=1$ we have
\begin{align*}
    & -36 c^2 \lambda^2+2 c (-94 c^2 \lambda^2+22 c^2 \lambda-36 c \lambda^2)+2 c^2 \lambda+2 c (\lambda^2+\lambda+21)+2 \lambda \\
    \leq& 2 (22 c^3+c^2+c+1)+2 (-94 c^3-54 c^2+c)+42 c \\
    =& -144 c^3-106 c^2+46 c+2 \leq 0.
\end{align*}
The last less-than clause derives from the condition that $c \geq c_0$.
\end{proof}

\begin{lemma}
\label{lem:36}
$2 c (11 c^3+c^2+c+1)+4 c (-47 c^3-36 c^2+c) \leq 0$ for $c \geq c_0$.
\end{lemma}

\begin{lemma}
\label{lem:15}
$2 c (z (c \lambda+\lambda^2+21)+\lambda z^2 (11 c-9 \lambda)+11)+2 c (-2 c \lambda^2 z (19 z+9)-\lambda z (9 c \lambda z-1))-2 \lambda z (9 c \lambda z-1) \leq 0$ for $z \geq c$, $\lambda \geq 1$, $c \geq c_0$.
\end{lemma}
\begin{proof}
\begin{align*}
    & 2 c (z (c \lambda+\lambda^2+21)+\lambda z^2 (11 c-9 \lambda)+11)+2 c (-2 c \lambda^2 z (19 z+9) \\
            & \quad -\lambda z (9 c \lambda z-1))-2 \lambda z (9 c \lambda z-1) \\
    =& z^2 (-94 c^2 \lambda^2+22 c^2 \lambda-36 c \lambda^2)+z (-36 c^2 \lambda^2+2 c^2 \lambda+2 c (\lambda^2+\lambda+21)+2 \lambda)+22 c.
\end{align*}
This is a parabola of $z$ and according to Lemma~\ref{lem:17}, its center is to the left of $c$ so it is decreasing with $z$.
If we set $z=c$ we have
\begin{align*}
    & 2 c (z (c \lambda+\lambda^2+21)+\lambda z^2 (11 c-9 \lambda)+11)+2 c (-2 c \lambda^2 z (19 z+9) \\
            & \quad -\lambda z (9 c \lambda z-1))-2 \lambda z (9 c \lambda z-1) \\
    \leq & c^2 (-94 c^2 \lambda^2+22 c^2 \lambda-36 c \lambda^2)+c (-36 c^2 \lambda^2+2 c^2 \lambda+2 c (\lambda^2+\lambda+21)+2 \lambda)+22 c \\
    =& 2 c (-47 c^3-36 c^2+c) \lambda^2+2 c (11 c^3+c^2+c+1) \lambda+2 c (21 c+11).
\end{align*}
This is a parabola of $\lambda$ and according to Lemma~\ref{lem:36}, its center is to the left of $1$ so it is decreasing with $\lambda$.
If we set $\lambda=1$ we have
\begin{align*}
    & 2 c (z (c \lambda+\lambda^2+21)+\lambda z^2 (11 c-9 \lambda)+11)+2 c (-2 c \lambda^2 z (19 z+9) \\
            & \quad -\lambda z (9 c \lambda z-1))-2 \lambda z (9 c \lambda z-1) \\
    \leq & 2 c (11 c^3+c^2+c+1)+2 c (-47 c^3-36 c^2+c)+2 c (21 c+11) \\
    =& 2 c (-36 c^3-35 c^2+23 c+12) \leq 0.
\end{align*}
The last less-than derives from the condition that $c \geq c_0$.
\end{proof}

\begin{lemma}
\label{lem:46}
$44 c^4+6 c^3-17 c^2+2 (-94 c^4-90 c^3-16 c^2)+2 c+1 \leq 0$ for $c \geq c_0$.
\end{lemma}
\begin{proof}
This is a numerical result.
\end{proof}

\begin{lemma}
\label{lem:44}
$-18 c^3 \lambda^2+2 c^3 \lambda+c^2 (2 \lambda^2-17 \lambda+42)+2 c (-47 c^3 \lambda^2+22 c^3 \lambda+2 c^2 \lambda (1-18 \lambda)-9 c \lambda^2)+2 c (\lambda+11)+\lambda \leq 0$ for $\lambda \geq 1$ and $c \geq c_0$.
\end{lemma}
\begin{proof}
\begin{align*}
    & -18 c^3 \lambda^2+ 2 c^3 \lambda+c^2 (2 \lambda^2-17 \lambda+42)+ 2 c (-47 c^3 \lambda^2+22 c^3 \lambda +2 c^2 \lambda (1-18 \lambda)-9 c \lambda^2) \\
        & \quad +2 c (\lambda+11)+\lambda \\
    =& 42 c^2+(-94 c^4-90 c^3-16 c^2) \lambda^2+(44 c^4+6 c^3-17 c^2+2 c+1) \lambda+22 c.
\end{align*}
This is a parabola of $\lambda$ and according to Lemma~\ref{lem:46} its center is to the left of $1$, so it is decreasing with $\lambda$.
If we set $\lambda=1$ we have
\begin{align*}
    & -18 c^3 \lambda^2+2 c^3 \lambda+c^2 (2 \lambda^2-17 \lambda+42)+2 c (-47 c^3 \lambda^2+22 c^3 \lambda+2 c^2 \lambda (1-18 \lambda)-9 c \lambda^2) \\
        & \quad +2 c (\lambda+11)+\lambda \\
    \leq & -50 c^4-84 c^3+9 c^2+24 c+1 \leq 0.
\end{align*}
The last less-than derives from the condition that $c \geq c_0$.
\end{proof}

\begin{lemma}
\label{lem:53}
$22 c^5+4 c^4-17 c^3+2 c^2+2 (-47 c^5-54 c^4-7 c^3)+c \leq 0$ for $c \geq c_0$.
\end{lemma}
\begin{proof}
This is a numerical result.
\end{proof}

\begin{lemma}
\label{lem:13}
$2 c (y (z (c \lambda+\lambda^2+21)+\lambda z^2 (11 c-9 \lambda)+11)+c \lambda z^2-9 c \lambda z-\lambda^2 y^2 z (19 z+9)+11 z+1)-\lambda (y+1)^2 z (9 c \lambda z-1) \leq 0$
for $y,z \geq c$, $\lambda \geq 1$, $c \geq c_0$.
\end{lemma}
\begin{proof}
\begin{align*}
    & 2 c (y (z (c \lambda+\lambda^2+21)+\lambda z^2 (11 c-9 \lambda)+11)+c \lambda z^2-9 c \lambda z-\lambda^2 y^2 z (19 z+9)+11 z+1) \\
        & \quad -\lambda (y+1)^2 z (9 c \lambda z-1) \\
    =& 2 c^2 \lambda z^2-18 c^2 \lambda z+y^2 (-2 c \lambda^2 z (19 z+9)-\lambda z (9 c \lambda z-1)) \\
        & \quad +y (2 c (z (c \lambda+\lambda^2+21)+\lambda z^2 (11 c-9 \lambda)+11)-2 \lambda z (9 c \lambda z-1)) \\
        & \quad -\lambda z (9 c \lambda z-1)+22 c z+2 c.
\end{align*}
This is a parabola of $y$ and according to Lemma~\ref{lem:15} its center is to the left of $c$, so it is decreasing with $y$.
If we set $y=c$ we have
\begin{align*}
    & 2 c (y (z (c \lambda+\lambda^2+21)+\lambda z^2 (11 c-9 \lambda)+11)+c \lambda z^2-9 c \lambda z-\lambda^2 y^2 z (19 z+9)+11 z+1) \\
        & \quad -\lambda (y+1)^2 z (9 c \lambda z-1) \\
    \leq & c^2 (-2 c \lambda^2 z (19 z+9)-\lambda z (9 c \lambda z-1))+2 c^2 \lambda z^2-18 c^2 \lambda z \\
        & \quad +c (2 c (z (c \lambda+\lambda^2+21)+\lambda z^2 (11 c-9 \lambda)+11)-2 \lambda z (9 c \lambda z-1)) \\
        & \quad -\lambda z (9 c \lambda z-1)+22 c z+2 c \\
    =& 22 c^2+z^2 (-47 c^3 \lambda^2+22 c^3 \lambda+2 c^2 \lambda (1-18 \lambda)-9 c \lambda^2) \\
        & \quad +z (-18 c^3 \lambda^2+2 c^3 \lambda+c^2 (2 \lambda^2-17 \lambda+42)+2 c (\lambda+11)+\lambda)+2 c.
\end{align*}
This is a parabola of $z$ and according to Lemma~\ref{lem:44} its center is to the left of $c$, so it is decreasing with $z$.
If we set $z=c$ we have
\begin{align*}
    & 2 c (y (z (c \lambda+\lambda^2+21)+\lambda z^2 (11 c-9 \lambda)+11)+c \lambda z^2-9 c \lambda z-\lambda^2 y^2 z (19 z+9)+11 z+1) \\
        & \quad -\lambda (y+1)^2 z (9 c \lambda z-1) \\
    \leq & 22 c^2+c^2 (-47 c^3 \lambda^2+22 c^3 \lambda+2 c^2 \lambda (1-18 \lambda)-9 c \lambda^2) \\
        & \quad +c (-18 c^3 \lambda^2+2 c^3 \lambda+c^2 (2 \lambda^2-17 \lambda+42)+2 c (\lambda+11)+\lambda)+2 c \\
    =& 42 c^3+44 c^2+(-47 c^5-54 c^4-7 c^3) \lambda^2+(22 c^5+4 c^4-17 c^3+2 c^2+c) \lambda+2 c.
\end{align*}
This is a parabola of $\lambda$ and according to Lemma~\ref{lem:53} its center is to the left of $1$, so it is decreasing with $\lambda$.
If we set $\lambda=1$ we have
\begin{align*}
    & 2 c (y (z (c \lambda+\lambda^2+21)+\lambda z^2 (11 c-9 \lambda)+11)+c \lambda z^2-9 c \lambda z-\lambda^2 y^2 z (19 z+9)+11 z+1) \\
        & \quad -\lambda (y+1)^2 z (9 c \lambda z-1) \\
    \leq & -25 c^5-50 c^4+18 c^3+46 c^2+3 c \leq 0.
\end{align*}
The last less-than derives from the condition that $c \geq c_0$.
\end{proof}

\begin{lemma}
\label{lem:66}
$24 c^3+3 c^2+2 (-36 c^5-36 c^4-76 c^3-36 c^2+c)+2 c+2 \leq 0$ for $c \geq c_0$.
\end{lemma}
\begin{proof}
This is a numerical result.
\end{proof}

\begin{lemma}
\label{lem:64}
$2 c^3 \lambda-18 c^2 \lambda^2+3 c^2 \lambda+2 c (-18 c^4 \lambda^2-18 c^3 \lambda^2-38 c^2 \lambda^2+11 c^2 \lambda-9 c \lambda^2+\lambda)+c (\lambda^2+21)+2 \lambda \leq 0$ for $\lambda \geq 1$ and $c \geq c_0$.
\end{lemma}
\begin{proof}
\begin{align*}
    & 2 c^3 \lambda-18 c^2 \lambda^2+3 c^2 \lambda+2 c (-18 c^4 \lambda^2-18 c^3 \lambda^2-38 c^2 \lambda^2+11 c^2 \lambda-9 c \lambda^2+\lambda)+c (\lambda^2+21)+2 \lambda \\
    =& (24 c^3+3 c^2+2 c+2) \lambda+(-36 c^5-36 c^4-76 c^3-36 c^2+c) \lambda^2+21 c.
\end{align*}
This is a parabola of $\lambda$ and according to Lemma~\ref{lem:66} its center is to the left of $1$, so it is decreasing with $\lambda$.
If we set $\lambda=1$ we have
\begin{align*}
    & 2 c^3 \lambda-18 c^2 \lambda^2+3 c^2 \lambda+2 c (-18 c^4 \lambda^2-18 c^3 \lambda^2-38 c^2 \lambda^2+11 c^2 \lambda-9 c \lambda^2+\lambda)+c (\lambda^2+21)+2 \lambda \\
    \leq & -36 c^5-36 c^4-52 c^3-33 c^2+24 c+2 \leq 0.
\end{align*}
The last less-than derives from the condition that $c \geq c_0$.
\end{proof}

\begin{lemma}
\label{lem:74}
$13 c^4+3 c^3+c^2+2 (-18 c^6-18 c^5-38 c^4-27 c^3+c^2)+2 c+1 \leq 0$ for $c \geq c_0$.
\end{lemma}
\begin{proof}
This is a numerical result.
\end{proof}

\begin{lemma}
\label{lem:60}
$-18 c^3 \lambda^2 z^2+c^2 \lambda z (11 z+3)+2 c (-9 c^3 \lambda^2 z^2+c^2 \lambda z-c \lambda^2 z (19 z+9))+c (-9 \lambda^2 z^2+(\lambda^2+21) z+11)+\lambda (z+1)^2 \leq 0$ for $z \geq c$, $\lambda \geq 1$, $c \geq c_0$.
\end{lemma}
\begin{proof}
\begin{align*}
    & -18 c^3 \lambda^2 z^2+c^2 \lambda z (11 z+3)+2 c (-9 c^3 \lambda^2 z^2+c^2 \lambda z-c \lambda^2 z (19 z+9)) \\
        & \quad +c (-9 \lambda^2 z^2+(\lambda^2+21) z+11)+\lambda (z+1)^2 \\
    =& z (2 c^3 \lambda-18 c^2 \lambda^2+3 c^2 \lambda+c (\lambda^2+21)+2 \lambda) \\
        & \quad +z^2 (-18 c^4 \lambda^2-18 c^3 \lambda^2-38 c^2 \lambda^2+11 c^2 \lambda-9 c \lambda^2+\lambda)+11 c+\lambda.
\end{align*}
This is a parabola of $z$ and according to Lemma~\ref{lem:64} its center is to the left of $c$, so it is decreasing with $z$.
If we set $z=c$ we have
\begin{align*}
    & -18 c^3 \lambda^2 z^2+c^2 \lambda z (11 z+3)+2 c (-9 c^3 \lambda^2 z^2+c^2 \lambda z-c \lambda^2 z (19 z+9)) \\
        & \quad +c (-9 \lambda^2 z^2+(\lambda^2+21) z+11)+\lambda (z+1)^2 \\
    \leq & c (2 c^3 \lambda-18 c^2 \lambda^2+3 c^2 \lambda+c (\lambda^2+21)+2 \lambda) \\
        & \quad +c^2 (-18 c^4 \lambda^2-18 c^3 \lambda^2-38 c^2 \lambda^2+11 c^2 \lambda-9 c \lambda^2+\lambda)+11 c+\lambda \\
    =& 21 c^2+(13 c^4+3 c^3+c^2+2 c+1) \lambda+(-18 c^6-18 c^5-38 c^4-27 c^3+c^2) \lambda^2+11 c.
\end{align*}
This is a parabola of $\lambda$ and according to Lemma~\ref{lem:74} its center is to the left of $1$, so it is decreasing with $\lambda$.
If we set $\lambda=1$ we have
\begin{align*}
    & -18 c^3 \lambda^2 z^2+c^2 \lambda z (11 z+3)+2 c (-9 c^3 \lambda^2 z^2+c^2 \lambda z-c \lambda^2 z (19 z+9)) \\
        & \quad +c (-9 \lambda^2 z^2+(\lambda^2+21) z+11)+\lambda (z+1)^2 \\
    \leq & -18 c^6-18 c^5-25 c^4-24 c^3+23 c^2+13 c+1 \leq 0.
\end{align*}
The last less-than derives from the condition that $c \geq c_0$.
\end{proof}

\begin{lemma}
\label{lem:85}
$c (23 c^3+5 c^2-6 c+2)+2 c (-18 c^5-36 c^4-56 c^3-27 c^2+c) \leq 0$ for $c \geq c_0$.
\end{lemma}
\begin{proof}
This is a numerical result.
\end{proof}

\begin{lemma}
\label{lem:82}
$c (c^3 \lambda-9 c^2 \lambda^2+3 c^2 \lambda+c (\lambda^2-8 \lambda+21)+2 \lambda+11)+2 c^2 (-9 c^4 \lambda^2-18 c^3 \lambda^2-28 c^2 \lambda^2+11 c^2 \lambda+c (\lambda-9 \lambda^2)+\lambda) \leq 0$ for $\lambda \geq 1$ and $c \geq c_0$.
\end{lemma}
\begin{proof}
\begin{align*}
    & c (c^3 \lambda-9 c^2 \lambda^2+3 c^2 \lambda+c (\lambda^2-8 \lambda+21)+2 \lambda+11) \\
        & \quad +2 c^2 (-9 c^4 \lambda^2-18 c^3 \lambda^2-28 c^2 \lambda^2+11 c^2 \lambda+c (\lambda-9 \lambda^2)+\lambda) \\
    =& c (23 c^3+5 c^2-6 c+2) \lambda+c (-18 c^5-36 c^4-56 c^3-27 c^2+c) \lambda^2+c (21 c+11).
\end{align*}
This is a parabola of $\lambda$ and according to Lemma~\ref{lem:85} its center is to the left of $1$, so it is decreasing with $\lambda$.
If we set $\lambda=1$ we have
\begin{align*}
    & c (c^3 \lambda-9 c^2 \lambda^2+3 c^2 \lambda+c (\lambda^2-8 \lambda+21)+2 \lambda+11) \\
        & \quad +2 c^2 (-9 c^4 \lambda^2-18 c^3 \lambda^2-28 c^2 \lambda^2+11 c^2 \lambda+c (\lambda-9 \lambda^2)+\lambda) \\
    \leq & c (23 c^3+5 c^2-6 c+2)+c (-18 c^5-36 c^4-56 c^3-27 c^2+c)+c (21 c+11) \\
    =& -c (18 c^5+36 c^4+33 c^3+22 c^2-16 c-13) \leq 0.
\end{align*}
The last less-than derives from the condition that $c \geq c_0$.
\end{proof}

\begin{lemma}
\label{lem:93}
$12 c^5+4 c^4-7 c^3+2 c^2+2 (-9 c^7-18 c^6-28 c^5-18 c^4+c^3)+c \leq 0$ for $c \geq c_0$.
\end{lemma}
\begin{proof}
This is a numerical result.
\end{proof}

\begin{lemma}
\label{lem:6}
$g_1(x,y,z) \leq \frac{9}{10}$ for $x,y,z \geq c$, $\lambda \geq 1$, $c \geq c_0$.
\end{lemma}
\begin{proof}
$g_1(x,y,z) \leq \frac{9}{10}$ is equivalent to
\begin{align*}
x (y (z (c \lambda+\lambda^2+21)+\lambda z^2 (11 c-9 \lambda)+11)+c \lambda z^2-9 c \lambda z-\lambda^2 y^2 z (19 z+9)+11 z+1) \\
        -\lambda x^2 (y+1)^2 z (9 c \lambda z-1)+\lambda y (z+1)^2 \leq 0.
\end{align*}
This is a parabola of $x$ and according to Lemma~\ref{lem:13} its center is to the left of $c$, so it is decreasing with $x$.
If we set $x=c$ we have
\begin{align*}
    & x (y (z (c \lambda+\lambda^2+21)+\lambda z^2 (11 c-9 \lambda)+11)+c \lambda z^2-9 c \lambda z-\lambda^2 y^2 z (19 z+9)+11 z+1) \\
        & \quad -\lambda x^2 (y+1)^2 z (9 c \lambda z-1)+\lambda y (z+1)^2 \\
    \leq & -c^2 \lambda (y+1)^2 z (9 c \lambda z-1)+c (y (z (c \lambda+\lambda^2+21)+\lambda z^2 (11 c-9 \lambda)+11)+c \lambda z^2-9 c \lambda z \\
        & \quad -\lambda^2 y^2 z (19 z+9)+11 z+1)+\lambda y (z+1)^2 \\
    =& -9 c^3 \lambda^2 z^2+c^2 \lambda z^2-8 c^2 \lambda z+y^2 (-9 c^3 \lambda^2 z^2+c^2 \lambda z-c \lambda^2 z (19 z+9)) \\
        & \quad +y (-18 c^3 \lambda^2 z^2+c^2 \lambda z (11 z+3)+c (-9 \lambda^2 z^2+(\lambda^2+21) z+11)+\lambda (z+1)^2)+11 c z+c.
\end{align*}
This is a parabola of $y$ and according to Lemma~\ref{lem:60} its center is to the left of $c$, so it is decreasing with $y$.
If we set $y=c$ we have
\begin{align*}
    & x (y (z (c \lambda+\lambda^2+21)+\lambda z^2 (11 c-9 \lambda)+11)+c \lambda z^2-9 c \lambda z \\
            & \quad -\lambda^2 y^2 z (19 z+9)+11 z+1)-\lambda x^2 (y+1)^2 z (9 c \lambda z-1)+\lambda y (z+1)^2 \\
    \leq & -9 c^3 \lambda^2 z^2+c^2 \lambda z^2-8 c^2 \lambda z+c^2 (-9 c^3 \lambda^2 z^2+c^2 \lambda z-c \lambda^2 z (19 z+9)) \\
        & \quad +c (-18 c^3 \lambda^2 z^2+c^2 \lambda z (11 z+3)+c (-9 \lambda^2 z^2+(\lambda^2+21) z+11)+\lambda (z+1)^2)+11 c z+c \\
    =& c z (c^3 \lambda-9 c^2 \lambda^2+3 c^2 \lambda+c (\lambda^2-8 \lambda+21)+2 \lambda+11) \\
        & \quad +c z^2 (-9 c^4 \lambda^2-18 c^3 \lambda^2-28 c^2 \lambda^2+11 c^2 \lambda+c (\lambda-9 \lambda^2)+\lambda)+c (11 c+\lambda+1).
\end{align*}
This is a parabola of $z$ and according to Lemma~\ref{lem:82} its center is to the left of $c$, so it is decreasing with $z$.
If we set $z=c$ we have
\begin{align*}
    & x (y (z (c \lambda+\lambda^2+21)+\lambda z^2 (11 c-9 \lambda)+11)+c \lambda z^2-9 c \lambda z-\lambda^2 y^2 z (19 z+9)+11 z+1) \\
        & \quad -\lambda x^2 (y+1)^2 z (9 c \lambda z-1)+\lambda y (z+1)^2 \\
    \leq & c^2 (c^3 \lambda-9 c^2 \lambda^2+3 c^2 \lambda+c (\lambda^2-8 \lambda+21)+2 \lambda+11) \\
        & \quad +c^3 (-9 c^4 \lambda^2-18 c^3 \lambda^2-28 c^2 \lambda^2+11 c^2 \lambda+c (\lambda-9 \lambda^2)+\lambda)+c (11 c+\lambda+1) \\
    =& 21 c^3+22 c^2+(12 c^5+4 c^4-7 c^3+2 c^2+c) \lambda+(-9 c^7-18 c^6-28 c^5-18 c^4+c^3) \lambda^2+c.
\end{align*}
This is a parabola of $\lambda$ and according to Lemma~\ref{lem:93} its center is to the left of $1$, so it is decreasing with $\lambda$.
If we set $\lambda=1$ we have
\begin{align*}
    & x (y (z (c \lambda+\lambda^2+21)+\lambda z^2 (11 c-9 \lambda)+11)+c \lambda z^2-9 c \lambda z-\lambda^2 y^2 z (19 z+9)+11 z+1) \\
        & \quad -\lambda x^2 (y+1)^2 z (9 c \lambda z-1)+\lambda y (z+1)^2 \\
    \leq & -9 c^7-18 c^6-16 c^5-14 c^4+15 c^3+24 c^2+2 c \leq 0.
\end{align*}
The last less-than derives from the condition that $c \geq c_0$.
\end{proof}

  \begin{lemma}
\label{lem:2:19}
$1+4 c^2 \lambda+21 \lambda^2+c \lambda (-15+22 \lambda)+2 c \lambda (11 \lambda-36 c^2 \lambda-36 c^3 \lambda+c (-9+2 \lambda)) \leq 0$ for $\lambda \geq 1$ and $c \geq c_0$.
\end{lemma}
\begin{proof}
\begin{align*}
    & 1+4 c^2 \lambda+21 \lambda^2+c \lambda (-15+22 \lambda)+2 c \lambda (11 \lambda-36 c^2 \lambda-36 c^3 \lambda+c (-9+2 \lambda)) \\
    =& 1+(-15 c-14 c^2) \lambda+(21+44 c+4 c^2-72 c^3-72 c^4) \lambda^2.
\end{align*}
This is a parabola of $\lambda$ and it is decreasing when $\lambda \geq 1$ and $c \geq c_0$.
Therefore if we set $\lambda=1$ we have
\begin{align*}
    & 1+4 c^2 \lambda+21 \lambda^2+c \lambda (-15+22 \lambda)+2 c \lambda (11 \lambda-36 c^2 \lambda-36 c^3 \lambda+c (-9+2 \lambda)) \\
    \leq & 22+29 c-10 c^2-72 c^3-72 c^4 \leq 0.
\end{align*}
\end{proof}

\begin{lemma}
\label{lem:2:17}
$-9+(1-15 c \lambda+21 \lambda^2) z+\lambda (-9 c+11 \lambda-36 c^2 \lambda) z^2+2 c \lambda z (2 c-18 c^2 \lambda z+\lambda (11+z)) \leq 0$
for $\lambda \geq 1$, $z \geq c \geq c_0$.
\end{lemma}
\begin{proof}
\begin{align*}
    & -9+(1-15 c \lambda+21 \lambda^2) z+\lambda (-9 c+11 \lambda-36 c^2 \lambda) z^2+2 c \lambda z (2 c-18 c^2 \lambda z+\lambda (11+z)) \\
    =& -9+(1+4 c^2 \lambda+21 \lambda^2+c \lambda (-15+22 \lambda)) z+\lambda (11 \lambda-36 c^2 \lambda-36 c^3 \lambda+c (-9+2 \lambda)) z^2.
\end{align*}
This is a parabola of $z$ and according to Lemma~\ref{lem:2:19} it is decreasing when $z \geq c$.
Therefore if we set $z=c$ we have
\begin{align*}
    & -9+(1-15 c \lambda+21 \lambda^2) z+\lambda (-9 c+11 \lambda-36 c^2 \lambda) z^2+2 c \lambda z (2 c-18 c^2 \lambda z+\lambda (11+z)) \\
    \leq & -9+c^2 \lambda (11 \lambda-36 c^2 \lambda-36 c^3 \lambda+c (-9+2 \lambda))+c (1+4 c^2 \lambda+21 \lambda^2+c \lambda (-15+22 \lambda)) \\
    =& -9+c+(-15 c^2-5 c^3) \lambda+(21 c+33 c^2+2 c^3-36 c^4-36 c^5) \lambda^2.
\end{align*}
This is a parabola of $\lambda$ and it is decreasing when $\lambda \geq 1$ and $c \geq c_0$.
Therefore if we set $\lambda=1$ we have
\begin{align*}
    & -9+(1-15 c \lambda+21 \lambda^2) z+\lambda (-9 c+11 \lambda-36 c^2 \lambda) z^2+2 c \lambda z (2 c-18 c^2 \lambda z+\lambda (11+z)) \\
    \leq & -9+22 c+18 c^2-3 c^3-36 c^4-36 c^5 \leq 0.
\end{align*}
\end{proof}

\begin{lemma}
\label{lem:2:32}
$-9-15 c^2 \lambda+2 c^3 \lambda+11 c^2 \lambda^2+c (1-27 \lambda+21 \lambda^2)+2 c (-9 c^2 \lambda-17 c^2 \lambda^2-36 c^3 \lambda^2-18 c^4 \lambda^2+c \lambda (-19+11 \lambda)) \leq 0$
for $\lambda \geq 1$ and $c \geq c_0$.
\end{lemma}
\begin{proof}
\begin{align*}
    & -9-15 c^2 \lambda+2 c^3 \lambda+11 c^2 \lambda^2+c (1-27 \lambda+21 \lambda^2) \\
        & \quad +2 c (-9 c^2 \lambda-17 c^2 \lambda^2-36 c^3 \lambda^2-18 c^4 \lambda^2+c \lambda (-19+11 \lambda)) \\
    =& -9+c+(-27 c-53 c^2-16 c^3) \lambda+(21 c+33 c^2-34 c^3-72 c^4-36 c^5) \lambda^2.
\end{align*}
This is a parabola of $\lambda$ and it is decreasing when $\lambda \geq 1$ and $c \geq c_0$.
Therefore if we set $\lambda=1$ we have
\begin{align*}
    & -9-15 c^2 \lambda+2 c^3 \lambda+11 c^2 \lambda^2+c (1-27 \lambda+21 \lambda^2) \\
        & \quad +2 c (-9 c^2 \lambda-17 c^2 \lambda^2-36 c^3 \lambda^2-18 c^4 \lambda^2+c \lambda (-19+11 \lambda)) \\
    \leq & -9-5 c-20 c^2-50 c^3-72 c^4-36 c^5 \leq 0.
\end{align*}
\end{proof}

\begin{lemma}
\label{lem:2:14}
$-19-9 z-29 c \lambda z-19 c \lambda z^2+\lambda^2 y^2 z (11+z)-2 c \lambda (1+y)^2 z (-1+9 c \lambda z)+y (-9+(1-19 c \lambda+21 \lambda^2) z+\lambda (-9 c+11 \lambda) z^2) \leq 0$
for $\lambda \geq 1$, $y,z \geq c \geq c_0$.
\end{lemma}
\begin{proof}
\begin{align*}
    & -19-9 z-29 c \lambda z-19 c \lambda z^2+\lambda^2 y^2 z (11+z)-2 c \lambda (1+y)^2 z (-1+9 c \lambda z) \\
        & \quad +y (-9+(1-19 c \lambda+21 \lambda^2) z+\lambda (-9 c+11 \lambda) z^2) \\
    =& -19-9 (1+3 c \lambda) z-c \lambda (19+18 c \lambda) z^2 \\
        & \quad +y (-9+(1-15 c \lambda+21 \lambda^2) z+\lambda (-9 c+11 \lambda-36 c^2 \lambda) z^2)+\lambda y^2 z (2 c-18 c^2 \lambda z+\lambda (11+z)).
\end{align*}
This is a parabola of $y$ and according to Lemma~\ref{lem:2:17} it is decreasing when $y \geq c$.
Therefore if we set $y=c$ we have
\begin{align*}
    & -19-9 z-29 c \lambda z-19 c \lambda z^2+\lambda^2 y^2 z (11+z)-2 c \lambda (1+y)^2 z (-1+9 c \lambda z) \\
        & \quad +y (-9+(1-19 c \lambda+21 \lambda^2) z+\lambda (-9 c+11 \lambda) z^2) \\
    \leq & -19-9 (1+3 c \lambda) z-c \lambda (19+18 c \lambda) z^2 \\
        & \quad +c (-9+(1-15 c \lambda+21 \lambda^2) z+\lambda (-9 c+11 \lambda-36 c^2 \lambda) z^2)+c^2 \lambda z (2 c-18 c^2 \lambda z+\lambda (11+z)) \\
    =& -19-9 c+(-9-15 c^2 \lambda+2 c^3 \lambda+11 c^2 \lambda^2+c (1-27 \lambda+21 \lambda^2)) z \\
        & \quad +(-9 c^2 \lambda-17 c^2 \lambda^2-36 c^3 \lambda^2-18 c^4 \lambda^2+c \lambda (-19+11 \lambda)) z^2.
\end{align*}
This is a parabola of $z$ and according to Lemma~\ref{lem:2:32} it is decreasing when $z \geq c$.
Therefore if we set $z=c$ we have
\begin{align*}
    & -19-9 z-29 c \lambda z-19 c \lambda z^2+\lambda^2 y^2 z (11+z)-2 c \lambda (1+y)^2 z (-1+9 c \lambda z) \\
        & \quad +y (-9+(1-19 c \lambda+21 \lambda^2) z+\lambda (-9 c+11 \lambda) z^2) \\
    \leq & -19-9 c+c^2 (-9 c^2 \lambda-17 c^2 \lambda^2-36 c^3 \lambda^2-18 c^4 \lambda^2+c \lambda (-19+11 \lambda)) \\
        & \quad +c (-9-15 c^2 \lambda+2 c^3 \lambda+11 c^2 \lambda^2+c (1-27 \lambda+21 \lambda^2)) \\
    =& -(19-c) (1+c)-(1+c) (27 c^2+7 c^3) \lambda-(1+c) (-21 c^2-c^3+18 c^4+18 c^5) \lambda^2.
\end{align*}
This is a parabola of $\lambda$ and it is decreasing when $\lambda \geq 1$ and $c \geq c_0$.
Therefore if we set $\lambda=1$ we have
\begin{align*}
    & -19-9 z-29 c \lambda z-19 c \lambda z^2+\lambda^2 y^2 z (11+z)-2 c \lambda (1+y)^2 z (-1+9 c \lambda z) \\
        & \quad +y (-9+(1-19 c \lambda+21 \lambda^2) z+\lambda (-9 c+11 \lambda) z^2) \\
    \leq & -19-18 c-5 c^2-12 c^3-24 c^4-36 c^5-18 c^6 \leq 0.
\end{align*}
\end{proof}

\begin{lemma}
\label{lem:2:64}
$c+2 \lambda-17 c^2 \lambda+2 c^3 \lambda+21 c \lambda^2+22 c^2 \lambda^2+2 c (\lambda-9 c^2 \lambda+11 c \lambda^2+2 c^2 \lambda^2-18 c^3 \lambda^2-18 c^4 \lambda^2) \leq 0$
for $\lambda \geq 1$ and $c \geq c_0$.
\end{lemma}
\begin{proof}
\begin{align*}
    & c+2 \lambda-17 c^2 \lambda+2 c^3 \lambda+21 c \lambda^2+22 c^2 \lambda^2+2 c (\lambda-9 c^2 \lambda+11 c \lambda^2+2 c^2 \lambda^2-18 c^3 \lambda^2-18 c^4 \lambda^2) \\
    =& c+(2+2 c-17 c^2-16 c^3) \lambda+(21 c+44 c^2+4 c^3-36 c^4-36 c^5) \lambda^2.
\end{align*}
This is a parabola of $\lambda$ and it is decreasing when $\lambda \geq 1$ and $c \geq c_0$.
Therefore if we set $\lambda=1$ we have
\begin{align*}
    & c+2 \lambda-17 c^2 \lambda+2 c^3 \lambda+21 c \lambda^2+22 c^2 \lambda^2+2 c (\lambda-9 c^2 \lambda+11 c \lambda^2+2 c^2 \lambda^2-18 c^3 \lambda^2-18 c^4 \lambda^2) \\
    \leq & 2+24 c+27 c^2-12 c^3-36 c^4-36 c^5 \leq 0.
\end{align*}
\end{proof}

\begin{lemma}
\label{lem:2:62}
$-18 c^3 \lambda^2 z^2+\lambda (1+z)^2-c^2 \lambda z (17+9 z)+c (-9+z+21 \lambda^2 z+11 \lambda^2 z^2)+2 c (c^2 \lambda z-9 c^3 \lambda^2 z^2+c \lambda^2 z (11+z)) \leq 0$
for $\lambda \geq 1$, $z \geq c \geq c_0$.
\end{lemma}
\begin{proof}
\begin{align*}
    & -18 c^3 \lambda^2 z^2+\lambda (1+z)^2-c^2 \lambda z (17+9 z)+c (-9+z+21 \lambda^2 z+11 \lambda^2 z^2) \\
        & \quad +2 c (c^2 \lambda z-9 c^3 \lambda^2 z^2+c \lambda^2 z (11+z)) \\
    =& -9 c+\lambda+(c+2 \lambda-17 c^2 \lambda+2 c^3 \lambda+21 c \lambda^2+22 c^2 \lambda^2) z \\
        & \quad +(\lambda-9 c^2 \lambda+11 c \lambda^2+2 c^2 \lambda^2-18 c^3 \lambda^2-18 c^4 \lambda^2) z^2.
\end{align*}
This is a parabola of $z$ and according to Lemma~\ref{lem:2:64} it is decreasing when $z \geq c$.
Therefore if we set $z=c$ we have
\begin{align*}
    & -18 c^3 \lambda^2 z^2+\lambda (1+z)^2-c^2 \lambda z (17+9 z)+c (-9+z+21 \lambda^2 z+11 \lambda^2 z^2) \\
        & \quad +2 c (c^2 \lambda z-9 c^3 \lambda^2 z^2+c \lambda^2 z (11+z)) \\
    \leq & -9 c+\lambda+c (c+2 \lambda-17 c^2 \lambda+2 c^3 \lambda+21 c \lambda^2+22 c^2 \lambda^2) \\
        & \quad +c^2 (\lambda-9 c^2 \lambda+11 c \lambda^2+2 c^2 \lambda^2-18 c^3 \lambda^2-18 c^4 \lambda^2) \\
    =& -9 c+c^2+(1+2 c+c^2-17 c^3-7 c^4) \lambda+(21 c^2+33 c^3+2 c^4-18 c^5-18 c^6) \lambda^2.
\end{align*}
This is a parabola of $\lambda$ and it is decreasing when $\lambda \geq 1$ and $c \geq c_0$.
Therefore if we set $\lambda=1$ we have
\begin{align*}
    & -18 c^3 \lambda^2 z^2+\lambda (1+z)^2-c^2 \lambda z (17+9 z)+c (-9+z+21 \lambda^2 z+11 \lambda^2 z^2) \\
        & \quad +2 c (c^2 \lambda z-9 c^3 \lambda^2 z^2+c \lambda^2 z (11+z)) \\
    \leq & 1-7 c+23 c^2+16 c^3-5 c^4-18 c^5-18 c^6 \leq 0.
\end{align*}
\end{proof}

\begin{lemma}
\label{lem:2:79}
$2 c^2 (\lambda-9 c^2 \lambda-8 c^2 \lambda^2-18 c^3 \lambda^2-9 c^4 \lambda^2+c \lambda (-19+11 \lambda))+c (-9+2 \lambda-17 c^2 \lambda+c^3 \lambda+11 c^2 \lambda^2+c (1-28 \lambda+21 \lambda^2)) \leq 0$
for $\lambda \geq 1$ and $c \geq c_0$.
\end{lemma}
\begin{proof}
\begin{align*}
    & 2 c^2 (\lambda-9 c^2 \lambda-8 c^2 \lambda^2-18 c^3 \lambda^2-9 c^4 \lambda^2 +c \lambda (-19+11 \lambda)) \\
        & \quad +c (-9+2 \lambda-17 c^2 \lambda+c^3 \lambda+11 c^2 \lambda^2+c (1-28 \lambda+21 \lambda^2)) \\
    =& (-9+c) c+c (2-26 c-55 c^2-17 c^3) \lambda+c (21 c+33 c^2-16 c^3-36 c^4-18 c^5) \lambda^2.
\end{align*}
This is a parabola of $\lambda$ and it is decreasing when $\lambda \geq 1$ and $c \geq c_0$.
Therefore if we set $\lambda=1$ we have
\begin{align*}
    & 2 c^2 (\lambda-9 c^2 \lambda-8 c^2 \lambda^2-18 c^3 \lambda^2-9 c^4 \lambda^2+c \lambda (-19+11 \lambda)) \\
        & \quad +c (-9+2 \lambda-17 c^2 \lambda+c^3 \lambda+11 c^2 \lambda^2+c (1-28 \lambda+21 \lambda^2)) \\
    \leq & -c (7+4 c+22 c^2+33 c^3+36 c^4+18 c^5) \leq 0.
\end{align*}
\end{proof}

\begin{lemma}
\label{lem:2:9}
$g_2(x,y,z) \leq \frac{9}{10}$ for $\lambda \geq 1$ and $x,y,z \geq c \geq c_0$.
\end{lemma}
\begin{proof}
$g_2(x,y,z) \leq \frac{9}{10}$ is equivalent to
\begin{align*}
\lambda y (1+z)^2-\lambda x^2 (1+y)^2 z (-1+9 c \lambda z)+x (-19-9 z-29 c \lambda z-19 c \lambda z^2+\lambda^2 y^2 z (11+z) \\
        +y (-9+(1-19 c \lambda+21 \lambda^2) z+\lambda (-9 c+11 \lambda) z^2)) \leq 0.
\end{align*}
Denote the left hand-side of the above inequality by $A$.
$A$ is a parabola of $x$ and according to Lemma~\ref{lem:2:14} it is decreasing when $x \geq c$.
Therefore if we set $x=c$ we have
\begin{align*}
    A \leq & \lambda y (1+z)^2-c^2 \lambda (1+y)^2 z (-1+9 c \lambda z) \\
        & \quad +c (-19-9 z-29 c \lambda z-19 c \lambda z^2+\lambda^2 y^2 z (11+z) \\
                & \quad \quad +y (-9+(1-19 c \lambda+21 \lambda^2) z+\lambda (-9 c+11 \lambda) z^2)) \\
    =& -19 c-9 c z-28 c^2 \lambda z-19 c^2 \lambda z^2-9 c^3 \lambda^2 z^2 \\
        & \quad +y^2 (c^2 \lambda z-9 c^3 \lambda^2 z^2+c \lambda^2 z (11+z)) \\
        & \quad +y (-18 c^3 \lambda^2 z^2+\lambda (1+z)^2-c^2 \lambda z (17+9 z)+c (-9+z+21 \lambda^2 z+11 \lambda^2 z^2)).
\end{align*}
This is a parabola of $y$ and according to Lemma~\ref{lem:2:62} it is decreasing when $y \geq c$.
Therefore if we set $y=c$ we have
\begin{align*}
    A \leq & -19 c-9 c z-28 c^2 \lambda z-19 c^2 \lambda z^2-9 c^3 \lambda^2 z^2+c^2 (c^2 \lambda z-9 c^3 \lambda^2 z^2+c \lambda^2 z (11+z)) \\
        & \quad +c (-18 c^3 \lambda^2 z^2+\lambda (1+z)^2-c^2 \lambda z (17+9 z)+c (-9+z+21 \lambda^2 z+11 \lambda^2 z^2)) \\
    =& c (-19-9 c+\lambda)+c (-9+2 \lambda-17 c^2 \lambda+c^3 \lambda+11 c^2 \lambda^2+c (1-28 \lambda+21 \lambda^2)) z \\
        & \quad +c (\lambda-9 c^2 \lambda-8 c^2 \lambda^2-18 c^3 \lambda^2-9 c^4 \lambda^2+c \lambda (-19+11 \lambda)) z^2.
\end{align*}
This is a parabola of $z$ and according to Lemma~\ref{lem:2:79} it is decreasing when $z \geq c$.
Therefore if we set $z=c$ we have
\begin{align*}
    A \leq & c (-19-9 c+\lambda)+c^3 (\lambda-9 c^2 \lambda-8 c^2 \lambda^2-18 c^3 \lambda^2-9 c^4 \lambda^2+c \lambda (-19+11 \lambda)) \\
        & \quad +c^2 (-9+2 \lambda-17 c^2 \lambda+c^3 \lambda+11 c^2 \lambda^2+c (1-28 \lambda+21 \lambda^2)) \\
    =& -(19-c) c (1+c)-c (1+c) (-1-c+28 c^2+8 c^3) \lambda-c (1+c) (-21 c^2-c^3+9 c^4+9 c^5) \lambda^2.
\end{align*}
This is a parabola of $\lambda$ and it is decreasing when $\lambda \geq 1$ and $c \geq c_0$.
Therefore if we set $\lambda=1$ we have
\begin{align*}
    A \leq & -c (18+16 c+5 c^2+14 c^3+16 c^4+18 c^5+9 c^6) \leq 0.
\end{align*}
\end{proof}

We now combine Lemma~\ref{lem:6} and Lemma~\ref{lem:2:9} to give a bound for 
$\frac{\left|\frac{\partial g}{\partial x}(x,y,z)\right|\Phi(x)+\left|\frac{\partial g}{\partial y}(x,y,z)\right|\Phi(y)+\left|\frac{\partial g}{\partial z}(x,y,z)\right|\Phi(z)}{\Phi(g(x,y,z))}$.
\begin{lemma}
\label{lem:thm2g}
$\frac{\left|\frac{\partial g}{\partial x}(x,y,z)\right|\Phi(x)+\left|\frac{\partial g}{\partial y}(x,y,z)\right|\Phi(y)+\left|\frac{\partial g}{\partial z}(x,y,z)\right|\Phi(z)}{\Phi(g(x,y,z))} \leq \frac{9}{10}$ for $x,y,z \geq c$, $\lambda \geq 1$, $c \geq c_0$.
\end{lemma}

  \begin{lemma}
\label{lem:thm2:h}
$\frac{\left|\deriv{h}{x}\left(x\right)\right|\Phi\left(x\right)}{\Phi\left(h\left(x\right)\right)}
\leq \max\left\{\frac{1}{2},\frac{1+p+p^2}{(1+p)(1+2p)}\right\}<1$ when $x \geq c \geq 1$, $\lambda \geq 1$, $\mu \geq p>0$.
\end{lemma}
\begin{proof}
\begin{align*}
    \frac{\left|\deriv{h}{x}\left(x\right)\right|\Phi\left(x\right)}{\Phi\left(h\left(x\right)\right)} &=
        \frac{\left|\lambda(c\mu+1)-\lambda\mu^2\right|x}{(1+\lambda\mu x)(\mu+(c\mu+1)\lambda x)} \\
    &\leq \frac{(\lambda(c\mu+1)+\lambda\mu^2) x}{(1+\lambda\mu x)(\mu+(c\mu+1)\lambda x)} \\
    &= \frac{\lambda(c\mu+1)+\lambda\mu^2}{\lambda^2(c\mu+1)\mu x+\frac{\mu}{x}+\lambda(c\mu+1)+\lambda\mu^2} \\
    &\leq \frac{\lambda(c\mu+1)+\lambda\mu^2}{\lambda^2(c\mu+1)\mu c +\frac{\mu}{c}+\lambda(c\mu+1)+\lambda\mu^2} \\
    &= \frac{(c\mu+1)+\mu^2}{\lambda(c\mu+1)\mu c +\frac{\mu}{c \lambda}+(c\mu+1)+\mu^2} \\
    &\leq \frac{(c\mu+1)+\mu^2}{(c\mu+1)\mu c +\frac{\mu}{c}+(c\mu+1)+\mu^2} \\
    &= \frac{c(1+c \mu+\mu^2)}{(1+c \mu)(c+\mu+c^2 \mu)}.
\end{align*}
Since the derivative of $\frac{c(1+c \mu+\mu^2)}{(1+c \mu)(c+\mu+c^2 \mu)}$ with respect to $c$ is
$$-\frac{\mu((c^4-1)\mu^2+(c^2-1)+2c(c^2-1)\mu+c^2 \mu^2+2c^3 \mu^3)}{(1+c\mu)^2(c+\mu+c^2\mu)^2}\leq 0,$$
it is decreasing when $c \geq 1$.
Therefore we have
\begin{align*}
    \frac{\left|\deriv{h}{x}\left(x\right)\right|\Phi\left(x\right)}{\Phi\left(h\left(x\right)\right)} &\leq
        \frac{1+\mu+\mu^2}{(1+\mu)(1+2\mu)} \\
    &\leq \max\left\{\frac{1}{2},\frac{1+p+p^2}{(1+p)(1+2p)}\right\}.
\end{align*}
\end{proof}

\begin{lemma}
\label{lem:thm2:ghx}
$\frac{\left|\pderiv{\hat{g}}{x}\right|\Phi(x)}{\Phi\left(\hat{g}\left(x,y\right)\right)}
\leq \max\left\{\frac{1}{2},\frac{1+p+p^2}{(1+p)(1+2p)}\right\}<1$ when $x \geq c \geq 1$, $\lambda \geq 1$, $y \geq p>0$.
\end{lemma}
\begin{proof}
\begin{align*}
    \frac{\left|\pderiv{\hat{g}}{x}\right|\Phi(x)}{\Phi\left(\hat{g}\left(x,y\right)\right)} &=
        \frac{\left|1+c\lambda y-\lambda^2 y^2\right|x}{(1+\lambda xy)(x+\lambda y+\lambda cxy)} \\
    &\leq \frac{(1+c\lambda y+\lambda^2 y^2)x}{(1+\lambda xy)(x+\lambda y+\lambda cxy)} \\
    &= \frac{1+c\lambda y+\lambda^2 y^2}{(\lambda y+c\lambda^2 y^2)x+\frac{\lambda y}{x}+1+c\lambda y+\lambda^2 y^2} \\
    &\leq \frac{1+c\lambda y+\lambda^2 y^2}{(\lambda y+c\lambda^2 y^2)c+\frac{\lambda y}{c}+1+c\lambda y+\lambda^2 y^2} \\
    &= \frac{c(1+c\lambda y+\lambda^2 y^2)}{(1+c\lambda y)(c+\lambda y+c^2\lambda y)}.
\end{align*}
Since the derivative of $\frac{c(1+c\lambda y+\lambda^2 y^2)}{(1+c\lambda y)(c+\lambda y+c^2\lambda y)}$ with respect to $c$ is
$$-\frac{\lambda y((c^2-1)+2c\lambda y(c^2-1)+\lambda^2 y^2(c^4-1)+c^2\lambda^2 y^2+2c^3\lambda^3 y^3)}{(1+c\lambda y)^2(c+\lambda y+c^2\lambda y)^2}\leq 0,$$
it is decreasing when $c \geq 1$.
Therefore we have
\begin{align}
    \frac{\left|\pderiv{\hat{g}}{x}\right|\Phi(x)}{\Phi\left(\hat{g}\left(x,y\right)\right)} &\leq
        \frac{1+\lambda y+\lambda^2 y^2}{(1+\lambda y)(1+2\lambda y)} \\
    &\leq \max\left\{\frac{1}{2},\frac{1+p+p^2}{(1+p)(1+2p)}\right\}.
\end{align}
\end{proof}

\begin{lemma}
\label{lem:thm2:ghy}
$\frac{\left|\pderiv{\hat{g}}{y}\right|\Phi(y)}{\Phi\left(\hat{g}\left(x,y\right)\right)}
\leq \max\left\{\frac{1}{2},\frac{1+p+p^2}{(1+p)(1+2p)}\right\}<1$ when $y \geq c \geq 1$, $\lambda \geq 1$, $x \geq p>0$.
\end{lemma}
\begin{proof}
Observe that $\hat{g}(x,y)=h(y)\mid_{\mu=x}$, and thus $\pderiv{\hat{g}}{y}=\deriv{h}{x}(y)\mid_{\mu=x}$.
From Lemma~\ref{lem:thm2:h} we know that this bound also holds for
$\frac{\left|\pderiv{\hat{g}}{y}\right|\Phi(y)}{\Phi\left(\hat{g}\left(x,y\right)\right)}$.
\end{proof}

\begin{proof}[Proof of Theorem~\ref{thm:2}]
The proof follows immeadiately from Lemma~\ref{lem:algo} and Lemma~\ref{lem:thm2g},~\ref{lem:thm2:h},~\ref{lem:thm2:ghx},~\ref{lem:thm2:ghy}, with $\alpha=\max\left\{\frac{9}{10},\frac{1+p+p^2}{(1+p)(1+2p)}\right\}$.
\end{proof}

  \subsection{Proof of Theorem~\ref{thm:spin}}

The recursion we use in this proof is the same as those in proving Theorem~\ref{thm:1}.
The difference is that this time we use potential function $\Phi(x)=x$.

\begin{lemma}
\label{lem:3:dg}
If $c/2 \leq x,y,z \leq 2c$, $\lambda>0$, and $c \geq 4\sqrt{\sqrt{2}-1}=2.57$, then there is some constant $\alpha(\lambda,c)<1$ such that
$\frac{\left|\pderiv{g}{x}\right|\Phi\left(x\right)+\left|\pderiv{g}{y}\right|\Phi\left(y\right)+\left|\pderiv{g}{z}\right|\Phi\left(z\right)}{\Phi\left(g\left(x,y,z\right)\right)} \leq \alpha$.
\end{lemma}
\begin{proof}
\begin{align*}
& \frac{\left|\pderiv{g}{x}\right|\Phi\left(x\right)+\left|\pderiv{g}{y}\right|\Phi\left(y\right)+\left|\pderiv{g}{z}\right|\Phi\left(z\right)}{\Phi\left(g\left(x,y,z\right)\right)} \\
=& \frac{(\lambda x z+1) \left(x \left| \frac{-y z \lambda^2+c z \lambda+1}{(\lambda x z+1)^2}\right| +y \left| \frac{\lambda}{\lambda x z+1}\right| +z \left| \frac{\lambda x (-c+x+\lambda y)}{(\lambda x z+1)^2}\right| \right)}{c \lambda x z+\lambda y+x} \\
\leq & \frac{(\lambda x z+1) \left(x \frac{y z \lambda^2+c z \lambda+1}{(\lambda x z+1)^2} +y \frac{\lambda}{\lambda x z+1} +z \left| \frac{\lambda x (-c+x+\lambda y)}{(\lambda x z+1)^2}\right| \right)}{c \lambda x z+\lambda y+x}.
\end{align*}
Let
\begin{align}
g_1(x,y,z) &= \frac{(\lambda x z+1) \left(x \frac{y z \lambda^2+c z \lambda+1}{(\lambda x z+1)^2} +y \frac{\lambda}{\lambda x z+1} +z \frac{\lambda x (-c+x+\lambda y)}{(\lambda x z+1)^2} \right)}{c \lambda x z+\lambda y+x} \\
&= \frac{2 c \lambda x z+\lambda^2 x y z-\lambda x^2 z+\lambda y+x}{(\lambda x z+1) (c \lambda x z+\lambda y+x)} \label{eq:thm3:2:g1}
\end{align}
and
\begin{align}
g_2(x,y,z) &= \frac{(\lambda x z+1) \left(x \frac{y z \lambda^2+c z \lambda+1}{(\lambda x z+1)^2} +y \frac{\lambda}{\lambda x z+1} -z \frac{\lambda x (-c+x+\lambda y)}{(\lambda x z+1)^2} \right)}{c \lambda x z+\lambda y+x} \\
&= \frac{3 \lambda^2 x y z+\lambda x^2 z+\lambda y+x}{(\lambda x z+1) (c \lambda x z+\lambda y+x)} \label{eq:thm3:2:g2},
\end{align}
then
$$\frac{\left|\pderiv{g}{x}\right|\Phi\left(x\right)+\left|\pderiv{g}{y}\right|\Phi\left(y\right)+\left|\pderiv{g}{z}\right|\Phi\left(z\right)}{\Phi\left(g\left(x,y,z\right)\right)}
\leq \max\left\{g_1(x,y,z),g_2(x,y,z)\right\}.$$
The result follows immediately after Lemma~\ref{lem:thm3:2:g1} and Lemma~\ref{lem:thm3:2:g2}.
\end{proof}

\begin{lemma}
\label{lem:thm3:2:g1}
If $c/2 \leq x,y,z \leq 2c$, $c>0$, and $\lambda \geq \lambda_0> 0$, then $g_1$ (defined in \eqref{eq:thm3:2:g1}) satisfies
$g_1(x,y,z) \leq \frac{8 c^2 \lambda^2+\left(6 c^2+32\right) \lambda+8}{c^2 \left(c^2+8\right) \lambda^2+\left(6 c^2+32\right) \lambda+8} < 1$.
\end{lemma}
\begin{proof}
\begin{align}
g_1(x,y,z) &= \frac{2 c \lambda x z+\lambda^2 x y z-\lambda x^2 z+\lambda y+x}{(\lambda x z+1) (c \lambda x z+\lambda y+x)} \label{eq:thm3:2:8} \\
&\leq \frac{3 c^2 \lambda z+2 c \left(\lambda^2 y z+1\right)+4 \lambda y}{(c \lambda z+2) \left(c^2 \lambda z+c+2 \lambda y\right)} \label{eq:thm3:2:13} && \text{by Lemma~\ref{lem:thm3:2:8}} \\
&\leq \frac{4 c \lambda^2 z+\lambda (3 c z+8)+2}{(c \lambda z+2) (\lambda (c z+4)+1)} \label{eq:thm3:2:17} &&\text{by Lemma~\ref{lem:thm3:2:13}} \\
&\leq \frac{8 c^2 \lambda^2+\left(6 c^2+32\right) \lambda+8}{\left(c^2 \lambda+4\right) \left(\left(c^2+8\right) \lambda+2\right)} \label{eq:thm3:2:g1a} &&\text{by Lemma~\ref{lem:thm3:2:17}} \\
&\leq \frac{8 c^2 \lambda_0^2+\left(6 c^2+32\right) \lambda_0+8}{\left(c^2 \lambda_0+4\right) \left(\left(c^2+8\right) \lambda_0+2\right)} &&\text{by Lemma~\ref{lem:thm3:2:g1a}} \\
&= \frac{8 c^2 \lambda_0^2+\left(6 c^2+32\right) \lambda_0+8}{c^2 \left(c^2+8\right) \lambda_0^2+\left(6 c^2+32\right) \lambda_0+8}.
\end{align}
\end{proof}

\begin{lemma}
\label{lem:thm3:2:8}
\eqref{eq:thm3:2:8} is decreasing monotonically with respect to $x$ when $x \geq c/4$.
\end{lemma}
\begin{proof}
The derivative of \eqref{eq:thm3:2:8} with respect to $x$ is
$$-\frac{\lambda z \left(x^2 \left(2 c^2 \lambda^2 z^2+c \lambda z \left(\lambda^2 y z+4\right)+2 \lambda^2 y z+2\right)+2 \lambda x y (c \lambda z+2)-c \lambda y\right)}{(\lambda x z+1)^2 (c \lambda x z+\lambda y+x)^2}
\leq 0.$$
\end{proof}

\begin{lemma}
\label{lem:thm3:2:13}
\eqref{eq:thm3:2:13} is increasing monotonically with respect to $y$.
\end{lemma}
\begin{proof}
The derivative of \eqref{eq:thm3:2:13} with respect to $y$ is
$$\frac{2 c^3 \lambda^3 z^2}{(c \lambda z+2) \left(c^2 \lambda z+c+2 \lambda y\right)^2} \geq 0.$$
\end{proof}

\begin{lemma}
\label{lem:thm3:2:17}
\eqref{eq:thm3:2:17} is decreasing monotonically with respect to $z$.
\end{lemma}
\begin{proof}
The derivative of \eqref{eq:thm3:2:17} with respect to $z$ is
$$-\frac{c^2 \lambda^2 z \left(4 c \lambda^2 z+\lambda (3 c z+16)+4\right)}{(c \lambda z+2)^2 (\lambda (c z+4)+1)^2} \leq 0.$$
\end{proof}

\begin{lemma}
\label{lem:thm3:2:g1a}
\eqref{eq:thm3:2:g1a} is decreasing monotonically with respect to $\lambda$ when $\lambda \geq 0$.
\end{lemma}
\begin{proof}
The derivative of \eqref{eq:thm3:2:g1a} with respect to $\lambda$ is $-\frac{4c^2\lambda^2\left(8+\left(32+3c^2\right)\lambda+4c^2\lambda^2\right)}{\left(4+c^2\lambda\right)^2\left(2+\left(8+c^2\right)\lambda\right)^2}<0$.
\end{proof}

\begin{lemma}
\label{lem:thm3:2:g2}
If $c/2 \leq x,y,z \leq 2c$, $\lambda\geq\lambda_0>0$, and $c \geq 4\sqrt{\sqrt{2}-1}=2.57$, then $g_2$ (defined in \eqref{eq:thm3:2:g2}) satisfies
$g_2(x,y,z) \leq \frac{192 \lambda_0^2+20 \lambda_0+1}{576 \lambda_0^2+52 \lambda_0+1} <1$.
\end{lemma}
\begin{proof}
\begin{align}
g_2(x,y,z) &= \frac{3 \lambda^2 x y z+\lambda x^2 z+\lambda y+x}{(\lambda x z+1) (c \lambda x z+\lambda y+x)} \label{eq:thm3:2:25} \\
&\leq \frac{c^2 \lambda z+c \left(6 \lambda^2 y z+2\right)+4 \lambda y}{(c \lambda z+2) \left(c^2 \lambda z+c+2 \lambda y\right)} \label{eq:thm3:2:27} &&\text{by Lemma~\ref{lem:thm3:2:25}} \\
&\leq \frac{12 c \lambda^2 z+\lambda (c z+8)+2}{(c \lambda z+2) (\lambda (c z+4)+1)} \label{eq:thm3:2:29} &&\text{by Lemma~\ref{lem:thm3:2:27}} \\
&\leq \frac{24 c^2 \lambda^2+2 \left(c^2+16\right) \lambda+8}{\left(c^2 \lambda+4\right) \left(\left(c^2+8\right) \lambda+2\right)} \label{eq:thm3:2:34} &&\text{by Lemma~\ref{lem:thm3:2:29}} \\
&\leq \frac{192 \lambda^2+20 \lambda+1}{576 \lambda^2+52 \lambda+1} \label{eq:thm3:2:g2a} &&\text{by Lemma~\ref{lem:thm3:2:34}} \\
&\leq \frac{192 \lambda_0^2+20 \lambda_0+1}{576 \lambda_0^2+52 \lambda_0+1}. &&\text{by Lemma~\ref{lem:thm3:2:g2a}}
\end{align}
\end{proof}

\begin{lemma}
\label{lem:thm3:2:25}
\eqref{eq:thm3:2:25} is decreasing monotonically with respect to $x$ when $c \geq 2$, $x,z \geq c/2$, and $y \leq 2c$.
\end{lemma}
\begin{proof}
The derivative of \eqref{eq:thm3:2:25} with respect to $x$ is
\begin{align}
& -\frac{\lambda^2 y z \left(3 c \lambda^2 x^2 z^2+2 c \lambda x z+c+2 \lambda x^2 z-2 \lambda y\right)}{(\lambda x z+1)^2 (c \lambda x z+\lambda y+x)^2} \\
\leq & -\frac{\lambda^2 y z \left(3 c \lambda^2 x^2 z^2+ \frac{2c^3}{4}+c+\frac{2c^3}{8}-4c\right)}{(\lambda x z+1)^2 (c \lambda x z+\lambda y+x)^2} \\
\leq & -\frac{\lambda^2 y z \left(3 c \lambda^2 x^2 z^2+ \frac{3}{4}c\left(c^2-4\right)\right)}{(\lambda x z+1)^2 (c \lambda x z+\lambda y+x)^2} \\
\leq & 0.
\end{align}
\end{proof}

\begin{lemma}
\label{lem:thm3:2:27}
\eqref{eq:thm3:2:27} is increasing monotonically with respect to $y$.
\end{lemma}
\begin{proof}
The derivative of \eqref{eq:thm3:2:27} with respect to $y$ is
$$\frac{2 c^2 \lambda^2 z (3 c \lambda z+4)}{(c \lambda z+2) \left(c^2 \lambda z+c+2 \lambda y\right)^2}.$$
\end{proof}

\begin{lemma}
\label{lem:thm3:2:29}
\eqref{eq:thm3:2:29} is decreasing monotonically with respect to $z$ when $z \geq c/2$ and $c \geq 4 \sqrt{\sqrt{2}-1} = 2.57$.
\end{lemma}
\begin{proof}
The derivative of \eqref{eq:thm3:2:29} with respect to $z$ is
\begin{align*}
& -\frac{c \lambda \left(12 c^2 \lambda^3 z^2+\lambda^2 \left(c^2 z^2+16 c z-64\right)+4 c \lambda z+4\right)}{(c \lambda z+2)^2 (\lambda (c z+4)+1)^2} \\
\leq & -\frac{c \lambda \left(12 c^2 \lambda^3 z^2+\lambda^2 \left(\frac{c^4}{4} +8 c^2-64\right)+4 c \lambda z+4\right)}{(c \lambda z+2)^2 (\lambda (c z+4)+1)^2} \\
\leq & 0.
\end{align*}
\end{proof}

\begin{lemma}
\label{lem:thm3:2:34}
\eqref{eq:thm3:2:34} is decreasing monotonically with respect to $c$ when $c \geq 4 \sqrt{\sqrt{2}-1} = 2.57$.
\end{lemma}
\begin{proof}
The derivative of \eqref{eq:thm3:2:34} with respect to $c$ is
$$-\frac{4 c \lambda \left(12 c^4 \lambda^3+8 c^2 \lambda+\left(c^4+32 c^2-256\right) \lambda^2+16\right)}{\left(c^2 \lambda+4\right)^2 \left(\left(c^2+8\right) \lambda+2\right)^2} \leq 0.$$
\end{proof}

\begin{lemma}
\label{lem:thm3:2:g2a}
\eqref{eq:thm3:2:g2a} is decreasing monotonically when $\lambda\geq0$.
\end{lemma}
\begin{proof}
The derivative of \eqref{eq:thm3:2:g2a} with respect to $\lambda$ is $-\frac{32(48\lambda^2+24\lambda+1)}{(576\lambda^2+52\lambda+1)^2}<0$.
\end{proof}

\begin{lemma}
\label{lem:3:dghx}
If $x,y \geq c/2 \geq 1$, $c \geq 1.60$, and $\lambda\geq\lambda_0>0$, then there is some constant $\alpha(\lambda_0)<1$ such that
$\frac{\left|\pderiv{\hat{g}}{x}\right|\Phi\left(x\right)}{\Phi\left(\hat{g}\left(x,y\right)\right)} \leq
\alpha(\lambda_0)<1$.
\end{lemma}
\begin{proof}
\begin{align}
& \frac{\left|\pderiv{\hat{g}}{x}\right|\Phi\left(x\right)}{\Phi\left(\hat{g}\left(x,y\right)\right)} \\
=& \frac{x \left|c \lambda y-\lambda^2 y^2+1\right|}{(\lambda x y+1) (c \lambda x y+\lambda y+x)} \\
\leq & \frac{x \left(c \lambda y+\lambda^2 y^2+1\right)}{(\lambda x y+1) (c \lambda x y+\lambda y+x)} \label{eq:thm3:14} \\
\leq & \frac{2 c \left(c \lambda y+\lambda^2 y^2+1\right)}{(c \lambda y+2) \left(c^2 \lambda y+c+2 \lambda y\right)} \label{eq:thm3:16} && \text{by Lemma~\ref{lem:thm3:14}} \\
\leq & \frac{2 c^2 \lambda (\lambda+2)+8}{\left(c^2 \lambda+4\right) \left(\left(c^2+2\right) \lambda+2\right)} \label{eq:thm3:22} && \text{by Lemma~\ref{lem:thm3:16}} \\
\leq & \frac{\lambda^2+2 \lambda+2}{2 \lambda^2+5 \lambda+2} && \text{by Lemma~\ref{lem:thm3:22}} \\
\leq & \max\left\{\frac{\lambda_0^2+2\lambda_0+2}{2\lambda_0^2+5\lambda_0+2},\frac{1}{2}\right\}. && \text{by Lemma~\ref{lem:thm3:dghxa}}.
\end{align}
\end{proof}

\begin{lemma}
\label{lem:thm3:14}
\eqref{eq:thm3:14} is decreasing monotonically with respect to $x$ when $x \geq c/2 \geq 1$.
\end{lemma}
\begin{proof}
The derivative of \eqref{eq:thm3:14} with respect to $x$ is
$$-\frac{\lambda y \left(c \lambda y+\lambda^2 y^2+1\right) \left(x^2 (c \lambda y+1)-1\right)}{(\lambda x y+1)^2 (c \lambda x y+\lambda y+x)^2} \leq 0.$$
\end{proof}

\begin{lemma}
\label{lem:thm3:16}
\eqref{eq:thm3:16} is decreasing monotonically with respect to $y$ when $c \geq \sqrt{\frac{1}{2} \left(1+\sqrt{17}\right)} = 1.60$.
\end{lemma}
\begin{proof}
The derivative of \eqref{eq:thm3:16} with respect to $y$ is
$$-\frac{2 c \lambda \left(c^4 \lambda^2 y^2+2 c^3 \lambda y + c^2 + \left(c^4 - c^2 - 4\right) \lambda^2 y^2\right)}{(c \lambda y+2)^2 \left(c^2 \lambda y+c+2 \lambda y\right)^2} \leq 0.$$
\end{proof}

\begin{lemma}
\label{lem:thm3:22}
\eqref{eq:thm3:22} is decreasing monotonically with respect to $c$ when $c \geq \sqrt{2}$.
\end{lemma}
\begin{proof}
The derivative of \eqref{eq:thm3:22} with respect to $\lambda$ is
$$-\frac{4 c \lambda \left(c^4 \lambda^3+2 \left(c^4-4\right) \lambda^2+8 \left(c^2-2\right) \lambda+8\right)}{\left(c^2 \lambda+4\right)^2 \left(\left(c^2+2\right) \lambda+2\right)^2} \leq 0.$$
\end{proof}

\begin{lemma}
\label{lem:thm3:dghxa}
If $\lambda \geq \lambda_0>1$, then
$\frac{\lambda^2+2 \lambda+2}{2 \lambda^2+5 \lambda+2} \leq
\max\left\{\frac{\lambda_0^2+2 \lambda_0+2}{2 \lambda_0^2+5 \lambda_0+2},\frac{1}{2}\right\}<1$.
\end{lemma}
\begin{proof}
Since
$\deriv{}{\lambda}\frac{\lambda^2+2 \lambda+2}{2 \lambda^2+5 \lambda+2}
= \frac{\lambda^2-4\lambda-6}{(2\lambda^2+5\lambda+2)^2}$, it is decreasing and then increasing when $\lambda>1$.
Hence its maximum is achieved on either boundary.
\end{proof}

\begin{lemma}
\label{lem:3:dghy}
If $c/2 \leq x,y \leq c+2/c$ and $c>0$, then there exists a constant $\alpha(c)<1$ such that
$\frac{\left|\pderiv{\hat{g}}{y}\right|\Phi\left(y\right)}{\Phi\left(\hat{g}\left(x,y\right)\right)} \leq \alpha(c) <1$.
\end{lemma}
\begin{proof}
\begin{align*}
& \frac{\left|\pderiv{\hat{g}}{y}\right|\Phi\left(y\right)}{\Phi\left(\hat{g}\left(x,y\right)\right)} \\
=& \frac{\lambda y \left|c x-x^2+1\right|}{(\lambda x y+1) (c \lambda x y+\lambda y+x)} \\
\leq & \frac{\lambda y \left(c x+x^2+1\right)}{(\lambda x y+1) (c \lambda x y+\lambda y+x)}.
\end{align*}
Since $\lim_{\lambda \to +\infty}\frac{\left|\pderiv{\hat{g}}{y}\right|\Phi\left(y\right)}{\Phi\left(\hat{g}\left(x,y\right)\right)}=0$, there is some constant $\lambda_0(x,y,c)$ such that if $\lambda > \lambda_0$, then
$\frac{\left|\pderiv{\hat{g}}{y}\right|\Phi\left(y\right)}{\Phi\left(\hat{g}\left(x,y\right)\right)}<\frac{1}{2}$.
Let $\lambda_1(c)=\max_{c/2 \leq x,y \leq c+2/c} \lambda_0(x,y,c)$.
Then if $\lambda>\lambda_1$, we have 
$\frac{\left|\pderiv{\hat{g}}{y}\right|\Phi\left(y\right)}{\Phi\left(\hat{g}\left(x,y\right)\right)}<\frac{1}{2}$.
Otherwise,
\begin{align*}
& \frac{\left|\pderiv{\hat{g}}{y}\right|\Phi\left(y\right)}{\Phi\left(\hat{g}\left(x,y\right)\right)} \\
\leq & 1- \frac{x}{(\lambda x y+1) (c \lambda x y+\lambda y+x)}\\
\leq & 1-\frac{c}{2(4c^2\lambda_1(c)+1)(4c^3\lambda_1(c)+2c\lambda_1(c)+2c)}.
\end{align*}
The proof is done by setting $\alpha(c)=\max\left\{\frac{1}{2},1-\frac{c}{2(4c^2\lambda_1(c)+1)(4c^3\lambda_1(c)+2c\lambda_1(c)+2c)}\right\}$.
\end{proof}

\begin{lemma}
\label{lem:3:dh}
If $c/2 \leq x,\mu \leq c+2/c$ and $c>0$, then there exists a constant $\alpha(c)<1$ such that
$\frac{\left|\deriv{h}{x}\right|\Phi\left(x\right)}{\Phi\left(h\left(x\right)\right)} \leq
\alpha(c)<1$.
\end{lemma}
\begin{proof}
Since $h(x)=\hat{g}(\mu,x)$, the result follows immediately after Lemma~\ref{lem:3:dghy}.
\end{proof}

\begin{proof}[Proof of Theorem~\ref{thm:spin-holant}]
This theorem is an application of Lemma~\ref{lem:algo}.
Required conditions are verified in Lemma~\ref{lem:3:dg}, Lemma~\ref{lem:3:dghx}, Lemma~\ref{lem:3:dghy}, and Lemma~\ref{lem:3:dh}.
\end{proof}

}

\bibliography{refs}

\ifabs{
\newpage
\appendix
\ \newline
\ \newline
\ \newline
\ \newline
\ \newline
\ \newline
\ \newline
\ \newline
\ \newline
\ \newline
\ \newline
\ \newline
\ \newline
\ \newline
\ \newline

\center \textbf{\huge Appendix: Full Paper}
}{}
\end{document}